\documentclass[twoside,leqno,twocolumn]{article}
\usepackage{ltexpprt}
\usepackage{enumerate}
\usepackage{t1enc}
\usepackage{pgfplots}
\usepackage{tikz}
\usepackage{amsfonts}

\newcommand{\realrange}[2]{\left[#1, #2\right]}

\newcommand{\unitrange}[2]{\realrange{0}{1}}

\newcommand{\llabel}[1]{\label{\labelprefix:#1}}
\newcommand{\labelprefix}{} 

\newcommand{\discussionsize}{\small}

\newcommand{\frage}[1]{}

\newenvironment{code}{\noindent
\begin{tabbing}%
\hspace{2em}\=\hspace{2em}\=\hspace{2em}\=\hspace{2em}\=\hspace{2em}\=%
\hspace{2em}\=\hspace{2em}\=\hspace{2em}\=\hspace{2em}\=\hspace{2em}\=%
\kill}{\end{tabbing}}

\newcommand{\labelcommand}{}
\newcommand{\captiontext}{}
\newsavebox{\codeparam}
\newcounter{lineNumber}
\newenvironment{disscodepos}[3]{%
\renewcommand{\labelcommand}{#2}%
\renewcommand{\captiontext}{#3}%
\sbox{\codeparam}{\parbox{\textwidth}{#3}}%
\begin{figure}[#1]\begin{center}\begin{code}\setcounter{lineNumber}{1}}{%
\end{code}\end{center}\caption{\llabel{\labelcommand}\captiontext}\end{figure}}

{\end{disscodepos}}

\newdimen\endofsize\endofsize=0.5em
\def\endofbeweis{~\quad\hglue\hsize minus\hsize
                 \hbox{\vrule height \endofsize width
\endofsize}\par}

\usepackage{epsfig}
\usepackage{url}
\usepackage{amsmath}
\usepackage{amssymb}
\usepackage{latexsym}
\usepackage{todonotes}
\usepackage{hyperref}
 \usepackage[algo2e,ruled,vlined]{algorithm2e}
\usepackage{algorithm}
 \usepackage{algorithmic}
\usepackage{wrapfig}
\usepackage{graphicx}
\newcommand\textproc{\textsc}
\usepackage{numprint}
\npdecimalsign{.} 
\usepackage{makeidx}
\usepackage{booktabs}

\usepackage{doi}
\usepackage{pdflscape}

\newcommand{\etal}{et~al.}

\usepackage{color}
\usepackage{colortbl}
\usepackage{etoolbox}
\usepackage[left] {lineno}
\definecolor {infocolor} {rgb} {0.6,0.6,0.6}

\usepackage[referable]{threeparttablex}
\usepackage{tabu}
\usepackage{multirow}
\usepackage{makecell}
\usepackage{longtable}
\usepackage{pifont}
\usepackage{xcolor}
\usepackage{colortbl}
\usepackage{blindtext}

\newcommand{\mc}{\multicolumn}

\newif\ifFull
\Fullfalse

\newif\ifDoubleBlind
\DoubleBlindfalse

\newif\ifAppendix
\Appendixfalse

\let\doendproof\endproof
\renewcommand\endproof{~\hfill$\boxtimes$\doendproof}

\usepackage{csvsimple}

\usepackage[labelformat=simple]{subcaption}

\usepackage{afterpage}

\usepackage[square,numbers,sort&compress]{natbib}

\newcommand{\niceremark}[3]{\textcolor{red}{\textsc{#1 #2: }}\textcolor{blue}{\textsf{#3}}}

\newcommand{\conFast}{$C_{\text{fast}}$}
\newcommand{\conStrong}{$C_{\text{strong}}$}

\newcommand{\basicSparse}{\textproc{Basic-Sparse}}
\newcommand{\basicDense}{\textproc{Basic-Dense}}
\newcommand{\nonIncreasing}{\textproc{NonIncreasing}}
\newcommand{\cyclicStrong}{\textproc{Cyclic-Strong}}
\newcommand{\cyclicFast}{\textproc{Cyclic-Fast}}

\newcommand{\hils}{\textproc{HILS}}
\newcommand{\dynWVCOne}{\textproc{DynWVC1}}
\newcommand{\dynWVCTwo}{\textproc{DynWVC2}}
\newcommand{\dynWVC}{\textproc{DynWVC}}

\definecolor{lightergray}{rgb}{0.86, 0.86, 0.86}
\definecolor{darkred}{rgb}{0.5, 0.0, 0.0}
\definecolor{darkgreen}{rgb}{0.0, 0.5, 0.0}

\newif\ifRemoveComments
\RemoveCommentsfalse

\ifRemoveComments
\renewcommand{\niceremark}[3]{ }
\fi

\ifRemoveComments
\usepackage[final]{changes}
\else
\usepackage[draft]{changes}
\fi{}
\definechangesauthor[name={Darren Strash}, color=blue]{DS} 
\definechangesauthor[name={Sebastian Lamm}, color=green]{SL} 
\definechangesauthor[name={Bogdan Zavalnij}, color=yellow]{BZ}

\makeatletter
\def\BState{\STATE\hskip-\ALG@thistlm}
\makeatother

\makeatletter
\newcommand{\manuallabel}[1]{\def\@currentlabel{#1}}
\makeatother

\newtheorem{reduction}{Reduction}

\patchcmd{\thebibliography}{\list}{\fontsize{0.98em}{0.9\baselineskip}\selectfont\list}{}{}

\newcommand{\mytitle}{Boosting Data Reduction for the Maximum Weight Independent Set Problem Using Increasing Transformations}

\begin{document}
\title{\mytitle\ifDoubleBlind\else\thanks{Partially supported by DFG grants SA 933/10-2
    and by National Research, Development and Innovation Office
    NKFIH Fund No. SNN-135643.
}\fi{}}

\ifDoubleBlind
\author{(Anonymous)}
\else
\author{Alexander Gellner\thanks{Institute for Theoretical Informatics, Karlsruhe Institute of Technology, Karlsruhe, Germany.}, Sebastian Lamm\footnotemark[2], Christian Schulz\thanks{Faculty of Computer Science, University of Vienna, Vienna, Austria.}, Darren Strash\thanks{Department of Computer Science, Hamilton College, Clinton, NY, USA.}, Bogd\'an Zav\'alnij\thanks{Department of Combinatorics and Discrete Mathematics, Alfr\'ed R\'enyi Institute of Mathematics, Budapest, Hungary.}}
\fi{} 
\date{}
\maketitle

\begin{abstract}
Given a vertex-weighted graph, the maximum weight independent set problem asks for a pair-wise non-adjacent set of vertices such that the sum of their weights is maximum. The branch-and-reduce paradigm is the de facto standard approach to solve the problem to optimality in practice. In this paradigm, data reduction rules are applied to \emph{decrease} the problem size. These data reduction rules ensure that given an optimum solution on the new (smaller) input, one can quickly construct an optimum solution on the original input.

We introduce new generalized data reduction and transformation rules for the problem. A key feature of our work is that some transformation rules can \emph{increase} the size of the input. Surprisingly, these so-called
\emph{increasing transformations} can simplify
the problem and also open up the reduction space to yield even smaller
irreducible graphs later throughout the algorithm. In experiments,
our algorithm computes significantly smaller irreducible
graphs on all except one instance, solves more instances to optimality than previously
possible, is up to two orders of magnitude faster than the best state-of-the-art solver,
and finds higher-quality solutions than heuristic solvers DynWVC and HILS on many instances. While the \emph{increasing} transformations are only efficient enough for preprocessing at this time, we see this as a critical initial step towards a new \emph{branch-and-transform}~paradigm.
\end{abstract}
\newpage
\pagestyle{plain}

\section{Introduction}
\label{sec:intro}
Given a graph $G=(V,E)$ and a weight function $w~:~V~\rightarrow~\mathbb{R}^+$ that assigns positive weights to vertices, the goal of the \emph{maximum weight independent set} (MWIS) problem is to compute a set of pairwise non-adjacent vertices $I \subseteq V$, whose total weight is maximum.
The problem is $\mathcal{NP}$-hard~\cite{garey1974}, and has a wide-range of practical applications in areas such as map labeling~\cite{gemsa2014dynamiclabel,barth-2016}, coding theory~\cite{nurmela1997new, brouwer1990new}, combinatorial auctions~\cite{wu2015solving}, alignment of biological networks~\cite{ay2011submap}, workload scheduling for energy-efficient scheduling of disks~\cite{chou2011energy}, computer vision~\cite{ma2012maximum}, wireless communication~\cite{xu2010maximum}, and protein structure prediction~\cite{mascia2010predicting}.

In practice, graphs with hundreds of vertices can be solved with
traditional branch-and-bound
methods~\cite{balas1986finding,babel1994fast,warren2006combinatorial,butenko-trukhanov}. However,
until recently, even for medium-sized synthetic instances, the maximum
weight independent set problem remained largely infeasible. In stark
contrast, the unweighted variants can be quickly solved on 
\emph{large} real-world instances---even with millions of
vertices---in practice, by using
\emph{kernelization}~\cite{strash-power-2016,chang2017,hespe2018scalable}
or the \emph{branch-and-reduce} paradigm~\cite{akiba-tcs-2016}.
Kernelization iteratively applies \emph{reduction rules}, thereby
reducing the size of the input graph until an irreducible graph is
obtained.  This irreducible graph is usually called a \emph{kernel} if it has
size bounded by a function of a specified input
parameter. 
A solution is then calculated on this irreducible graph and extended
to a solution of the original instance by undoing the reduction rules.
Branch-and-reduce takes this process to the extreme: reduction rules
are exhaustively applied before branching in a branch-and-bound
algorithm. Branching then changes the remaining instance, opening up
the reduction space, allowing further reduction rules to be applied
before the next branching step.  For those instances that can't be
solved exactly, high-quality (and often exact) solutions can be found
by combining kernelization with either local
search~\cite{dahlum2016,chang2017} or evolutionary
algorithms~\cite{redumis-2017}.

Recently, Lamm~\etal~\cite{lamm-2019} discovered new reduction rules for the weighted problem that, in practice, are often able to calculate very small irreducible graphs on a wide number of instances.
Surprisingly, the branch-and-reduce algorithm is able to solve a large number of instances up to two orders of magnitude faster than existing (inexact) local search algorithms.
However, the algorithm still fails to compute small reduced graphs on some instances~\cite{lamm-2019}.
This is mainly due to their very specialized nature that searches the graph for specific subgraphs that can be removed.

\paragraph{Our Results.}
Unlike narrowly-defined data reduction rules, we engineer new generalized data
reduction and transformation rules. The transformation rules, in
contrast to data reduction rules, can also \emph{increase} the size of
the input. Surprisingly, this can simplify the problem and also open
up the reduction space to yield even smaller irreducible graphs later
throughout the algorithm~\cite{alexe2003struction}. 

More precisely, we engineer practically efficient variants of the stability number data reduction rule (called struction).
To the best of our knowledge, these are the first practical
implementations of the weighted struction rule that are able to handle a large variety of real-world~instances. 
Our algorithm exploits the full potential of the struction rule by also allowing the application of structions that may increase the number of vertices.
These new rules are integrated in the state-of-the-art branch-and-reduce algorithm by Lamm~\etal~\cite{lamm-2019} -- with and without the property that a data transformation rule can increase the size of the input.
Extensive experiments indicate that our algorithm calculates significantly smaller irreducible graphs than current state-of-the-art approach and, preprocessing with our transformations enables branch-and-reduce to solve many instances that were previously infeasible to solve to optimality. 

\section{Related Work}
\label{sec:related_work}
In the following we will first give a short overview of existing work on both exact and heuristic procedures, especially outlining how kernelization and preprocessing methods are used in state-of-the-art algorithms.
Furthermore, we present a detailed overview of the struction, as this transformation is the focus of this work.

\subsection{Exact Methods.}
Exact algorithms usually compute optimal solutions by systematically exploring the solution space.
A frequently used paradigm in exact algorithms for combinatorial optimization problems is called \emph{branch-and-bound}~\cite{ostergaard2002fast,warren2006combinatorial}.
In case of the MWIS problem, these types of algorithms compute optimal solutions by case distinctions in which vertices are either included into the current solution or excluded from it, branching into two or more subproblems and resulting in a search tree.
Over the years, branch-and-bound methods have been improved by new branching schemes or better pruning methods using upper and lower bounds to exclude specific subtrees~\cite{balas1986finding, babel1994fast,li2017minimization}.
In particular, Warren and Hicks~\cite{warren2006combinatorial} proposed three branch-and-bound algorithms that combine the use of weighted clique covers and a branching scheme first introduced by Balas and Yu~\cite{balas1986finding}.
Their first approach extends the algorithm by Babel~\cite{babel1994fast} by using a more intricate data structures to improve its performance.
The second one is an adaptation of the algorithm of Balas and Yu, which uses a weighted clique heuristic that yields structurally similar results to the heuristic of Balas and Yu.
The last algorithm is a hybrid version that combines both algorithms and is able to compute optimal solutions on graphs with hundreds of vertices.

In recent years, reduction rules have frequently been added to branch-and-bound methods yielding so-called \emph{branch-and-reduce} algorithms~\cite{akiba-tcs-2016}. 
These algorithms are able to improve the worst-case runtime of branch-and-bound algorithms by applications of reduction rules to the current graph before each branching step.
For the unweighted case, a large number of branch-and-reduce
algorithms have been developed in the past. The currently best exact
solver \cite{hespe2020wegotyoucovered}, which won the PACE challenge
2019~\cite{hespe2020wegotyoucovered, bogdan-pace, peaty-pace}, uses a portfolio of branch-and-reduce/bound solvers for the
complementary  problems. 
However, for a long time, virtually no weighted reduction rules were known, which is why hardly any branch-and-reduce algorithms exist for the MWIS problem.

To the best of our knowledge, the first and only branch-and-reduce algorithm for the weighted case was recently presented by Lamm~\etal~\cite{lamm-2019}.
The authors first introduce two meta-reductions called neighborhood removal and neighborhood folding, from which they derive a new set of weighted reduction rules.
On this foundation a branch-and-reduce algorithm is developed using pruning with weighted clique covers similar to the approach by Warren and Hicks~\cite{warren2006combinatorial} for upper bounds and an adapted version of the ARW local search~\cite{andrade-2012} for lower bounds.
The experimental evaluation shows that their algorithm can solve a large set of real-world instances and outperform state-of-the-art algorithms.

Finally, there are exact procedures which are either based on other extension of the branch-and-bound paradigm, e.g.~\cite{rebennack2011branch,warrier2005branch,warrier2007branch}, or on the reformulation into other $\mathcal{NP}$-complete problems, for which a variety of solvers already exist.
For instance, Xu~\etal~\cite{xu2016new} recently developed an algorithm called \texttt{SBMS}, which calculates an optimal solution for a given MWVC instance by solving a series of SAT instances.

\subsection{Heuristic Methods.}
A widely used heuristic approach is local search, which usually computes an initial solution and then tries to improve it by simple insertion, removal or swap operations. 
Although in theory local search generally offers no guarantees for the solution's quality, in practice they find high quality solutions significantly faster than exact procedures.

For unweighted graphs, the iterated local search (ARW) by Andrade~\etal~\cite{andrade-2012}, is a very successful heuristic.
It is based on so-called $(1,2)$-swaps which remove one vertex from the solution and add two new vertices to it, thus improving the current solution by one.
Their algorithm uses special data structures which find such a $(1,2)$-swap in linear time in the number of edges or prove that none exists.
Their algorithm is able to find (near-)optimal solutions for small to medium- sized instances in milliseconds, but struggles on massive instances with millions of vertices and edges.

The hybrid iterated local search (HILS) by Nogueira~\etal~\cite{hybrid-ils-2018} adapts the ARW algorithm for weighted graphs.
In addition to weighted $(1,2)$-swaps, it also uses $(\omega,1)$-swaps that add one vertex $v$ into the current solution and exclude its $\omega$ neighbors. 
These two types of neighborhoods are explored separately using variable neighborhood descent (VND).
In practice, their algorithm finds all known optimal solutions on well-known benchmark instances within milliseconds and outperforms other state-of-the-art local searches.

Two other local searches, DynWVC1 and DynWVC2, for the equivalent minimum weight vertex cover problem are presented by Cai~\etal~\cite{cai-dynwvc}.
Their algorithms extend the existing FastWVC heuristic~\cite{li2017efficient} by dynamic selection strategies for vertices to be removed from the current solution.
In practice, DynWVC1 outperforms previous MWVC heuristics on map labeling instances and large scale networks, and DynWVC2 provides further improvements on large scale networks but performs worse on map labeling instances.

Recently, Li~\etal~\cite{li2019numwvc} presented a local search algorithm for the minimum weight vertex cover (MWVC) problem, which is complementary to the MWIS problem. Their algorithm applies reduction rules during the construction phase of the initial solution.
Furthermore, they adapt the configuration checking approach~\cite{cai2011local} to the MWVC problem which is used to reduce cycling, i.e. returning to a solution that has been visited recently.
Finally, they develop a technique called self-adaptive-vertex-removing, which dynamically adjusts the number of removed vertices per iteration.
Experiments show that algorithm outperforms state-of-the-art approaches on both graphs of up to millions of nodes and real-world instances.

\subsection{Struction.}
Originally the struction (STability number RedUCTION) was introduced by Ebenegger~\etal~\cite{ebenegger1984pseudo} and was later improved by Alexe~\etal~\cite{alexe2003struction}.
This method is a graph transformation for unweighted graphs, that can be applied to an arbitrary vertex and reduces the stability number by exactly one.
Thus, by successive application of the struction, the independence number of a graph can be determined.
Ebenegger~\etal~also show that there is an equivalence between finding a maximum weight independent set and maximizing a pseudo Boolean function, i.e. a real-valued function with Boolean variables, which allows to derive the struction as a special case.
Finally, the authors present a generalization of the struction to weighted graphs.

On this basis, theoretical algorithms with polynomial time complexity for special graph classes have been developed~\cite{hammer1985stability, hammer1985struction, alexe2003struction}. 
These algorithms use additional reduction rules and a careful selection of vertices on which the struction is applied.

Hoke and Troyon~\cite{hoke1994struction} developed another form of the weighted struction, using the same equivalence found by Ebenegger~\etal~\cite{ebenegger1984pseudo}.
In particular, they derive the \emph{revised weighted struction}.
However, this type of struction can only be applied to claw-free graphs: graphs without an induced three-leaf star.
This transformation also removes a vertex~$v$ and its neighborhood, but is able to create fewer new vertices, since these are only created for pairs of non-adjacent neighbors whose combined weight is greater than the weight of $v$.

As far as we are aware, prior to this work, only few experiments with
struction variants exist and are limited to only small instances:
Ebenegger~\etal~and Alexe~\etal~evaluated the struction only on small
graphs with less than a hundred vertices for the unweighted
case~\cite{ebenegger1984pseudo,alexe2003struction}.  Furthermore, for
the weighted case, none of the previously proposed struction variants has been
evaluated so
far~\cite{ebenegger1984pseudo,hoke1994struction,alexe2003struction}. 

\section{Preliminaries}
\label{sec:prelim}
A graph $G = (V,E)$ consists of a \emph{vertex set} $V$ and an \emph{edge set} $E \subset V \times V$.
It is called \emph{undirected} if for each edge $(u,v) \in E$ the edge set also contains $(v,u)$.
We only consider undirected graphs without self loops, i.e.~$(v,v)~\not\in~E$, and therefore we denote edges by sets~$\{u,v\}$.
In addition, a graph is (vertex-)\emph{weighted} if a positive scalar weight $w(v)$ is assigned to each vertex $v \in V$.
The weight of a vertex set~$X~\subset~V$ is defined as~$w(X)=~\sum_{v~\in~X}~w(x)$. 

A graph $G'=(V',E')$ is a \emph{subgraph} of $G=(V,E)$ if
$V'~\subset~V$ and~$E'~\subset~E~\cap~(V'~\times~V')$
holds. 
Given a vertex set $U~\subset~V$ the \emph{induced subgraph} of $U$ is the graph $G'=(U,E')$ with~$E'~=~E~\cap~(U~\times~U)$ and is denoted by $G[U]$.
Two vertices $u,v$ are called \emph{adjacent} if $\{u,v\} \in E$.
The \emph{neighborhood} $N(v)$ of a vertex $v$ is the set of all vertices adjacent to $v$. 
$N[v]=N(v)\cup\{v\}$ is called the \emph{closed neighborhood} and $\overline{N}(v) = V \setminus N[v]$ the \emph{non-neighborhood} of $v$.
Finally, we denote the \emph{degree} of a vertex $v$ is $\delta(v) = |N(v)|$.

For a given graph $G=(V,E)$, a vertex set $I \subset V$ is an \emph{independent set} if all vertices $v \in I$ are pairwise not adjacent.
An independent set is called \emph{maximal} if it is not a subset of another independent set and \emph{maximum} if no other independent set has greater cardinality.
The independence number~$\alpha(G)~=~|I|$, sometimes also called stability number, of a graph~$G$ is the cardinality of a maximum independent set $I$.
Likewise, for a weighted graph $G$ an independent set $I$ has \emph{maximum weight}, if there is no independent set $I'$ with a weight $w(I')$ greater than $w(I)$.
The weighted independence number $\alpha_w(G) = w(I)$ of a weighted graph $G$ is defined as the weight of a maximum weight independent set $I$.
For a given weighted graph $G$, the \emph{maximum weight independent set problem} (MWIS) seeks a maximum weight independent set.

\subsection{Original Weighted Struction.}
\label{subsec:struction}
We now present the original weighted struction introduced by Ebenegger et al.~\cite{ebenegger1984pseudo}, on which we base our struction variants. In general, we apply a struction to a \emph{center vertex}~$v$ and denote its neighborhood by~$N(v)~=~\{1,2,...,p\}$.
All variants we use remove $v$ from the graph $G$, producing a new
graph $G'$, and reduce the weighted independence number of the graph $G$ by its weight, i.e.~$\alpha_w(G)~=~\alpha_w(G')~+~w(v)$.

For ease of presentation, we first introduce a method called \emph{layering}.
Layering describes the partitioning of a given set $M$ that contains vertices $v_{x,y}$, that are indexed by two parameters~$x~\in~X,\ y~\in~Y$.
The sets $X,\ Y$ either contain vertices or vertex sets.
For $k~\in~X$ a layer $L_k$ contains all vertices having $k$ as first parameter, i.e.~$L_k~=~\{v_{x,y}~\in~M~:~x~=~k\}$.
Conversely, the \emph{layer} of a vertex $v_{x,y}$ is~$L(v_{x,y})=k$. 

In order to apply the original struction by
Ebenegger~\etal~\cite{ebenegger1984pseudo}, the center vertex $v$ must
have minimum weight among its closed neighborhood.
The struction is then applied by removing $v$ and creating new vertices for each pair $i,j$ of non-adjacent vertices in $N(v)$. To guarantee that we can obtain an MWIS $I$ of $G$ using an MWIS $I'$ of $G'$ with $w(I)~=~w(I')~+~w(v)$, we also insert edges between the new and original vertices.
An example of this type of struction is given in Figure~\ref{fig:original_struction}\ifAppendix\else.\fi{}

In the following, we provide a formal definition of the original struction:
\begin{reduction}[Original Struction]
Let~$v~\in~V$ be a vertex with minimum weight~$w(v)$ among its closed neighborhood.
Transform the graph as follows:
\begin{itemize}
\item Remove~$v$, lower weight of each neighbor by~$w(v)$
\item For each pair of non-adjacent neighbors~$x<y$, create a new vertex~$v_{x,y}$ with weight~$w(v_{x,y})~:=~w(v)$
\item 
Insert edges between~$v_{q,x},v_{r,y}$ if either~$x$ and~$y$ are adjacent or $L_q \neq L_r$
\item 
Each vertex~$v_{x,y}$ is also connected to vertex ~$w~\in~V\setminus\{v\}$ adjacent to either~$x$ or~$y$
\end{itemize}
\end{reduction}

For a MWIS~$I'$ of~$G'$ we obtain a MWIS~$I$ of~$G$ as follows: If~$I'~\cap~N(v)~=~\emptyset$ applies, we have~$I~=~I'~\cup~\{v\}$, otherwise we remove the new vertices, i.e.~$I~=~I'~\cap~V$.
Furthermore we have~$\alpha_w(G)~=~\alpha_w(G')~+~w(v)$.

\ifAppendix
\else
\begin{figure*}[ht]
	\centering
	\begin{subfigure}[b]{0.25\textwidth}
		\centering
		\includegraphics[width=\textwidth]{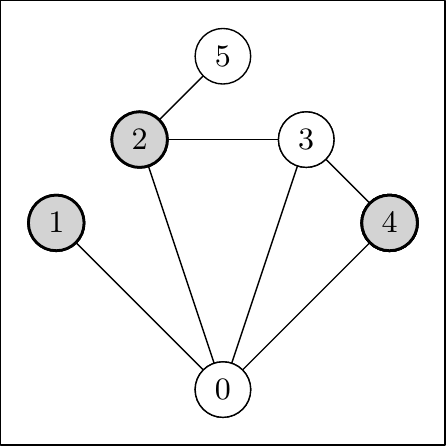}
		\caption{Original Graph}
	\end{subfigure}
	\begin{subfigure}[b]{0.25\textwidth}
		\centering
		\includegraphics[width=\textwidth]{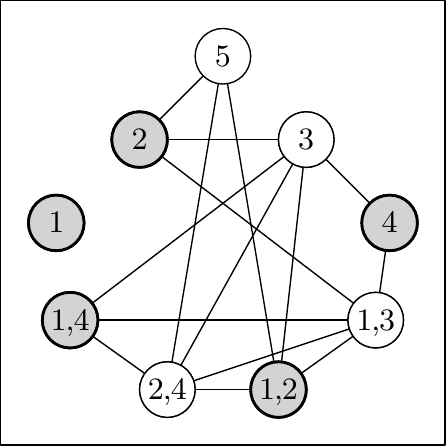}
		\caption{Original Struction}
		\label{fig:original_struction}
	\end{subfigure}
	\begin{subfigure}[b]{0.25\textwidth}
		\centering
		\includegraphics[width=\textwidth]{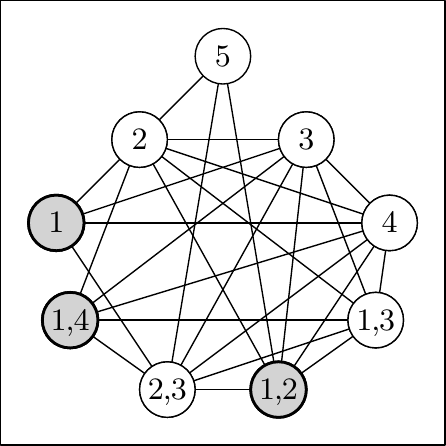}
		\caption{Modified Struction}
		\label{fig:modified_struction}
	\end{subfigure}
	\caption{Application of original struction and modified struction.
	Vertices representing the same independent set in the different graphs are highlighted in gray.}
\end{figure*}
\fi{} 

\section{New Weighted Struction Variants}
We now introduce three new struction variants: 
First, we deal with the fact that using the original weighted struction, an MWIS in the transformed graph might consist of more vertices than in the original graph.
We do so by using different weight assignments for the new vertices and inserting additional edges.
Second, we present a generalization of the revised weighted struction that can be applied to vertices of general graphs without the need to fulfill specific weight constraints.
However, this variant creates new vertices for independent sets in the neighborhood of a vertex $v$ whose weight is greater than $w(v)$.
Finally, we alleviate this issue by creating new vertices only for a specific subset of these independent sets in the third variant.

\subsection{Modified Weighted Struction.}
One caveat of the original struction is that the \emph{number of vertices} that are part of an MWIS in the transformed graph is generally larger than in the original graph.
The modified struction tries to alleviate this issue by ensuring that the number of vertices of an MWIS stays the same in both graphs. 
This is done by using a different weight assignment and inserting additional edges.
In particular, the newly created vertex for each pair of non-adjacent neighbors~$x,y \in N(v)$ with~$x~<~y$ is now assigned weight~$w(v_{x,y})~=~w(y)$ (instead of $w(v)$).
Furthermore, in addition to the edges created in the original struction, each neighbor~$k~\in~N(v)$ is connected to each vertex~$v_{x,y}$ belonging to a different layer than~$k$.
Finally, $N(v)$ is extended to a clique by adding edges between vertices~$x,y~\in~N(v)$.
For an MWIS~$I'$ of~$G'$ we now obtain an MWIS~$I$ of~$G$ as follows: If~$I'~\cap~N(v)~=~\emptyset$ applies, we have~$I~=~I'~\cup~\{v\}$, otherwise we obtain $I$ by replacing each new vertex~$v_{x,y}~\in~I'$ with the original vertex~$v_y$, i.e.~$I~=~(I'~\cap~V)~\cup~\{v_y~|~v_{x,y}~\in~I'~\setminus~V\}$.
As for the original struction, we have~$\alpha_w(G)~=~\alpha_w(G')~+~w(v)$.
An example of the modified struction is given in Figure~\ref{fig:modified_struction}\ifAppendix\else.\fi{} 
A proof of correctness can be found in Appendix~\ref{appendix:proofs}.

\subsection{Extended Weighted Struction.}
The extended struction removes the weight restriction for the vertex $v$ in the former variants.
Unlike the previous two structions, this variant considers independent sets of arbitrary size in the neighborhood $N(v)$.
In fact, we create new vertices for each independent set in $G[N(v)]$ if its weight is greater than $v$.
Note that this can result in up to $\mathcal{O}(2^{\delta(v)})$ new vertices.
An example application of the extended struction can be found in Figure~\ref{fig:extended_struction}{}.

\begin{reduction}[Extended Weighted Struction]
Let~$v~\in~V$ be an arbitrary vertex and~$C$ the set of all independent sets~$c$ in~$G[N(v)]$ with~$w(c)~>~w(v)$.
We derive the transformed graph~$G'$ as follows:
First, remove~$v$ together with its neighborhood and create a new vertex~$v_c$ with weight~$w(v_c)~=~w(c)~-~w(v)$ for each independent set~$c~\in~C$.
Each vertex~$v_c$ is then connected to each
non-neighbor~$w~\in~\overline{N}(v)$ adjacent to at least one
vertex in $c$. 
Finally, the vertices~$v_c$ are connected with each other, forming a clique.
For an MWIS~$I'$ of~$G'$ we obtain an MWIS~$I$ of~$G$ as follows:
If $I'~\setminus~V~=~\{v_c\}$ replace~$v_c$ with the vertices of the corresponding independent set $c$, i.e.~$I~=~(I'~\cap~V)~\cup~c$, otherwise~$I~=~I'~\cup~\{v\}$.
Furthermore we have~$\alpha_w(G)~=~\alpha_w(G')~+~w(v)$.
\end{reduction}

A proof of correctness is in Appendix~\ref{appendix:proofs}.

\ifAppendix
\else
\begin{figure*}[ht]
	\centering
	\begin{subfigure}[t]{0.25\textwidth}
		\centering
		\includegraphics[width=\textwidth]{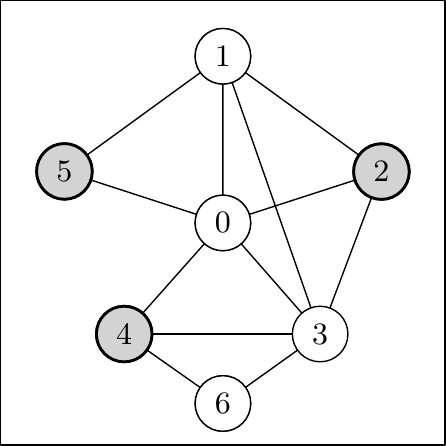}
		\caption{Original Graph}
	\end{subfigure}
	\begin{subfigure}[t]{0.25\textwidth}
		\centering
		\includegraphics[width=\textwidth]{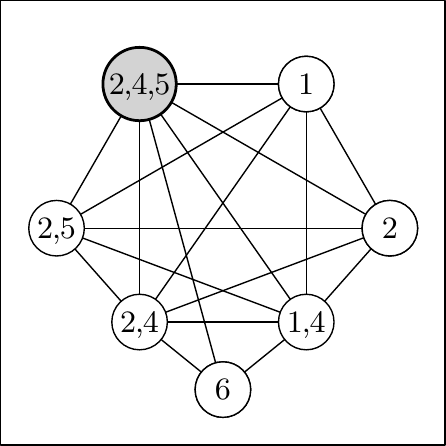}
		\caption{Extended Struction}
		\label{fig:extended_struction}
	\end{subfigure}
	\begin{subfigure}[t]{0.25\textwidth}
		\centering
		\includegraphics[width=\textwidth]{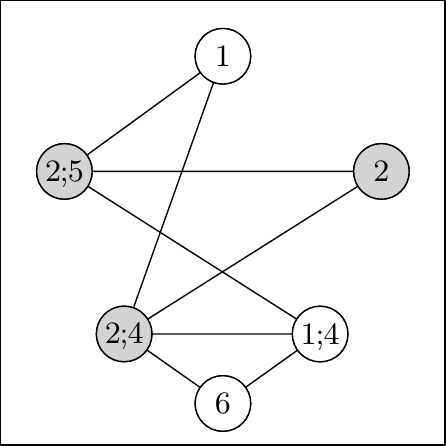}
		\caption{Extended Reduced Struction}
		\label{fig:extended_reduced_struction}
	\end{subfigure}
	\caption{Application of extended struction and extended reduced struction.
	Vertices representing the same independent set in the different graphs are highlighted in gray.
	We assume some weight constraints in the original graph for the construction in b) and c):
$w(1)~>w(0)$, $w(2)~>w(0)$ and $w(3)~+~w(4)~+~w(5)~\leq~w(0)$.}
\end{figure*}
\fi{}

\subsection{Extended Reduced Weighted Struction.}
The extended reduced struction is a variant of the extended struction, which can potentially reduce the number of newly created vertices.
For this purpose, only independent sets with weight ``just'' greater than $w(v)$ are considered.
In particular, this type of struction considers independent sets that have weight greater than $w(v)$, where each subset of this independent set has weight less than $w(v)$.
For this purpose, let $C$ be the set of all independent sets in
$G[N(v)]$ and $C'~\subseteq~C$ be the subset of independent sets for which there is no independent set in $C$ that has a weight greater than $w(v)$ and is a proper superset of $C'$.
We then use the same construction as for the extended struction, but only create new vertices for the set $C'$.
The resulting set of vertices is denoted by $V_C$.
However, since this construction might not be valid anymore, we also add additional vertices that are connected to each other by using layering.
To be more specific, for each pair of an independent set~$c~\in~C'$
and a vertex~$y~\in~N(v)$ we create a vertex~$v_{c,y}$ with
weight~$w(v_{c,y})~=~w(y)$, if~$c$ can be extended by~$y$, i.e.~$y$ is
not adjacent to any vertex~$v'~\in~c$. 
We denote this set of vertices $v_{c,y}$ by $V_E$.
We then insert edges between two vertices~$v_{c,y},v_{c',y'}$ if they either belong to different layers or $y$ and $y'$ have been adjacent.
Moreover, each vertex~$v_{c,y}$ is connected to each non-neighbor~$w~\in~\overline{N}(v)$, if~$w$ has been connected to either~$y$ or a vertex~$x~\in~c$.
Finally, we connect each vertex~$v_{c}$ to each vertex~$v_{c',y}$ belonging to a different layer than~$c$.
For an MWIS~$I'$ of~$G'$ we obtain an MWIS~$I$ of~$G$ as follows:
If $I'~\cap~V_C~=\emptyset$ applies, we set $I~=~I'~\cup~\{v\}$.
Otherwise, there is a single vertex $v_c \in I' \cap V_C$ that we replace with the vertices of its independent set $c$.
Moreover, we replace each vertex $v_{c,y} \in I' \cap V_E$ with the vertex $v_y$.
Altogether we have $I = (I' \cap V) \cup c \cup \{v_y \mid v_{c,y} \in I' \cap V_E \}$.
Furthermore we have~$\alpha_w(G)~=~\alpha_w(G')~+~w(v)$.
An example application of the extended struction can be found in Figure~\ref{fig:extended_reduced_struction}\ifAppendix\else.\fi{}

\ifAppendix
\else
\begin{algorithm}[t]
\SetAlgoLined
\begin{algorithmic}
	\STATE   \textbf{input} graph $G=(V,E)$, current solution weight $c$ (initially zero), best solution weight $\mathcal{W}$ (initially zero)
        \vspace*{-.45cm}
	\STATE   \textbf{procedure} Solve($G$, $c$, $\mathcal{W}$)
	\STATE   \quad $(G,c) \leftarrow$ Reduce$(G, c)$
	\STATE   \quad \textbf{if} $\mathcal{W} = 0$ \textbf{then} $\mathcal{W}  \leftarrow$ $c+\mathrm{ILS}(G)$ 
	\STATE   \quad \textbf{if} $c$ + UpperBound($G$) $\leq \mathcal{W}$ \textbf{then} \textbf{return} $\mathcal{W}$
	\STATE   \quad \textbf{if} $G$ is empty \textbf{then} \textbf{return} $\max\{\mathcal{W}, c\}$
	\STATE   \quad \textbf{if} $G$ is not connected \textbf{then}
	\STATE   \quad \quad \textbf{for all} $G_i \in $ Components($G$) \textbf{do}
	\STATE   \quad \quad \quad $c \leftarrow c + \text{Solve}$($G_i$, 0, 0)
	
	\STATE   \quad \quad \textbf{return} $\max(\mathcal{W},c)$
	\STATE   \quad $(G_1, c_1), (G_2, c_2) \leftarrow $ Branch$(G, c)$
	\STATE   \quad  \COMMENT{Run 1st case, update currently best solution}
	\STATE   \quad $\mathcal{W} \leftarrow $ Solve$(G_1, c_1, \mathcal{W})$ 
	\STATE   \quad \COMMENT{Use updated $\mathcal{W}$ to shrink the search space}
	\STATE   \quad $\mathcal{W}\leftarrow $ Solve$(G_2, c_2, \mathcal{W})$
	\STATE   \textbf{return} $\mathcal{W}$
\end{algorithmic}

\caption{Branch-and-Reduce Algorithm for MWIS}
\label{branchreducelabel}
  \label{branch_reduce_code}
\end{algorithm}
\fi{} 

\section{Practically Efficient Structions}
\label{sec:struction}
We now propose our two novel preprocessing algorithms for the MWIS problems based on the struction variants presented in the previous section.
Furthermore, we present how we integrate the different structions into the framework of Lamm~\etal~\cite{lamm-2019}, both as a preprocessing step and as a reduce step in branch-and-reduce, to more quickly compute optimal solutions. This takes an initial step towards a more general \emph{branch-and-transform} framework.

Since the different forms of the weighted struction do not necessarily reduce the number of vertices, we divide them (and existing reduction rules) into three different classes that are used throughout this paper: 
For \emph{decreasing transformations} (reductions) the transformed graph $G'$ has less vertices than the original graph $G$.
Note that all reduction rules used in the original framework of Lamm~\etal~\cite{lamm-2019} belong in this class.
Furthermore, we derive special cases of the structions, which also belong to this type.
Transformations where the number of vertices in the original graph
stays the same -- but reduce the size and weight of MWIS -- in the transformed graph are called \emph{plateau transformations}.
While plateau transformations cannot reduce the size of the graph, they can potentially produce new subgraphs which can then be reduced by other (decreasing) transformations.
Finally, \emph{increasing transformations} are transformations whose resulting graph has more vertices than the original graph.
Similar to plateau transformations, the idea is to reduce the resulting graph by further reduction rules and transformations.
However, since increasing structions can lead to a larger transformed graph, it is difficult to integrate them into algorithms that only apply non-increasing transformations.

\subsection{Non-Increasing Reduction Algorithm.}
In this section we show how to obtain decreasing and plateau transformations from the four different forms of structions presented in Section~\ref{subsec:struction}.
Based on these, an incremental preprocessing algorithm is proposed.

In general, when applying any form of struction, the number of vertices of the transformed graph $G'$ depends on the number of removed and newly created vertices.
However, it is difficult to estimate the number of resulting vertices in advance, since it varies depending on the number of possible independent sets present in the neighborhood of the center vertex.
Thus, when applying a struction variant, we generally keep track of the number of vertices that will be created.
If this number exceeds a given maximum value $n_{max}$, we discard the corresponding struction to ensure that not too many vertices are created.

We begin by taking a closer look at the structurally similar original weighted struction and modified weighted struction.
These variants reduce the number of vertices by at most one, since they only remove the center vertex $v$.
Therefore, decreasing or plateau structions can be obtained by setting $n_{max}~=~0$ or $n_{max}~=~1$.
However, note that this type of decreasing struction is already covered by the isolated weight transfer proposed by Lamm~\etal~\cite{lamm-2019}.

Looking at the two remaining struction variants, we see that they not only remove the center vertex $v$ but also its neighborhood $N(v)$.
Thus, the size of the graph can be reduced by up to $\delta(v) + 1$.
Decreasing or plateau structions can be obtained by using the corresponding struction variant with $n_{max}~=~\delta(v)$ or $n_{max}~=~\delta(v) + 1$.

The resulting rules can then easily be integrated into the kernelization algorithm used by Lamm~\etal~\cite{lamm-2019}.
However, since all struction variants are very general reduction rules, they tend to be expensive in terms of running time.
We therefore apply them after the faster localized reduction rules, but before the even more expensive critical set reduction.
To be specific, we use the following reduction rule order:~$R~=~[$\textproc{Neighborhood Removal}, \textproc{Degree Two Fold}, \textproc{Clique Reduction}, \textproc{Domination}, \textproc{Twin}, \textproc{Clique Neighborhood Removal}, \textproc{Generalized Fold}, \textproc{Decreasing Struction}, \textproc{Plateau Struction}, \textproc{Critical Weighted Independent Set}$]$.

\subsection{Cyclic Blow-Up Algorithm.}
Next, we extend the previous algorithm to also make use of increasing structions.
The main idea is to alternate between computing an irreducible graph using the previous (non-increasing) algorithm and then applying increasing structions while ensuring that the graph size does not increase too much.
The reasoning for this is that even though the graph size might increase, this can generate new and potentially reducible subgraphs, thus leading to an overall decrease in the graph size.

In the following, we say that a graph $K$ is better than a graph $K'$ if it has fewer vertices.
However, our algorithm can easily be adapted to match other quality criteria.
Pseudocode for our algorithm is given in Algorithm~\ref{cyclic_blow_up_code}.

Our algorithms maintains two graphs $K$ and $K^{\star}$.
$K^{\star}$ is the best graph found so far, i.e. the graph with the least number of vertices.
$K$ is the current graph, which we try to reduce to get a better graph $K^{\star}$.
Both graphs, $K$ and $K^{\star}$ are initialized with the graph obtained by applying the non-increasing reduction algorithm of the previous section.
The algorithm then alternates between two phases, a \emph{blow-up phase} (Section~\ref{blowup_phase}) and a \emph{reduction phase}.
During the blow-up phase a set of increasing structions is applied to $K$, resulting in a new graph~$K'$.
$K'$ is then reduced using the non-increasing algorithm, resulting in a graph~$K''$.
Next, we have to decide whether or not to use $K''$ or $K$ for the next iteration.
Note, that it can be advantageous to accept a graph $K''$ even if it has more vertices than $K$ to avoid local minima.
Nonetheless, we decided to only keep $K''$ if it is less vertices than $K$ as this strategy provided better results during preliminary experiments.
Finally, since we might not completely reduce the graph, we use a termination criterion, which will be discussed in Section~\ref{stopping_criterion}.

\begin{algorithm}[t]
\SetAlgoLined
\begin{algorithmic}
  \STATE \textbf{input} graph $G=(V,E)$, unsuccessful iteration threshold $X \in [1,\infty)$ 
	\STATE \textbf{procedure} CyclicBlowUp($G$, $X$)
	\STATE \quad $K \leftarrow$ Reduce$(G)$
	\STATE \quad $K^{\star} \leftarrow K$
	\STATE \quad count $\leftarrow 0$ 
  \STATE \quad \textbf{while} $|V(K)| < \alpha \cdot |V(K^{\star})|$ \textbf{and} count $< X$ \textbf{do}
	\STATE \quad \quad $K' \leftarrow$ BlowUp$(K)$ 
  \STATE \quad \quad \textbf{if} $K' = K$ \textbf{then}
  \STATE \quad \quad \quad \textbf{return} $K^{\star}$
  \STATE \quad \quad $K'' \leftarrow$ Reduce$(K')$
  \STATE \quad \quad $K \leftarrow$ Accept$(K'',K)$
    \STATE \quad \quad \textbf{if} $K < K^{\star}$ \textbf{then}
  \STATE \quad \quad \quad $K^{\star} \leftarrow K$
	\STATE \textbf{return} $K^{\star}$
\end{algorithmic}
\caption{Cyclic Blow-Up Algorithm}
\label{cyclic_blow_up_code}
\end{algorithm}

\subsubsection{Blow-Up Phase.}
\label{blowup_phase}
The starting point of the blow-up phase is an irreducible graph, where no more reductions (including decreasing structions) can be applied.
Next, we select a vertex $v$ from a candidate set $\mathcal{C}$.
This candidate set consists of all vertices in the current graph which have not been explicitly excluded for selection during the algorithm.
Vertex selection is a crucial part of our algorithm. 
Depending on the selected vertex the struction might create a large number of new vertices and the size of the transformed graph can increase drastically.
Thus, we discuss possible selection strategies in the next section.

Next, we apply a struction to the selected vertex $v$.
As for our previous algorithm, we keep track of the number of newly created and deleted vertices during this step.
In particular, if the struction would result in more than $n_{max}$ vertices, it is aborted.
In this case, the vertex $v$ is excluded from the candidate set.
The vertex $v$ will become viable again as soon as the corresponding
struction would create another transformed graph, i.e. its neighborhood~$N(v)$ changed.

After having applied a struction, we then proceed with the subsequent reduction phase.
It might also be possible to apply more than one struction during a blow-up phase.
However, one has to be careful to not let the size of the graph grow too large.

\paragraph{Vertex Selection Strategies.}
\label{vertex_selection_strategies}
The goal of the vertex selection procedure is to find an increasing struction that results in a new graph, which can then be reduced to an overall smaller graph.
In general, it is very difficult to estimate in advance to what extent the transformed graph can be reduced without actually performing the reduction phase.
Most of the following strategies therefore aim at increasing the size of the graph by only a few vertices.
The number of newly created vertices is determined by the number of independent sets in the neighborhood of $v$ having a total greater weight than $v$.
In general, determining this number is $\mathcal{NP}$-complete~\cite{prodinger1982fibonacci} and thus often infeasible to compute in practice.
In contrast, a much simpler selection strategy would be to choose vertices uniformly at random.
However, this can lead to structions that significantly increase the graph size.

Thus, in order to limit the size increase of a struction, we decided to use an approximation of the exact number of independent sets in the neighborhood of $v$.
In particular, we only consider independent sets up to a size of two.
This results in a lower bound $L$ for the number of independent sets~\cite{pedersen2006bounds} which can be computed in $\mathcal{O}(\Delta^2)$ time.
Since the lower bound $L$ can be far smaller than the actual number of newly created vertices, we use an additional \emph{tightness-check}: 
This check is passed if less than $L'~=\lceil~\beta~\cdot~L~\rceil$ new vertices with $\beta \in (1,\infty)$ are created by the corresponding struction.
Our strategy then works as follows:
We select a vertex $v$ with a minimal increase and perform the tightness-check.
If it fails, we know that at least $L'$ new vertices are created by the corresponding struction.
Therefore~$L'$ forms a tighter bound for the number of new vertices, and we reinsert $v$ to $C$ using the bound $L'$.
We then repeat this process until we find a vertex that passes the tightness-check.

\vspace*{-.25cm}
\subsection{Termination criteria.}
\label{stopping_criterion}
In general, the size of $K$ can decrease very slowly or even
exhibit oscillatory behavior.
This can cause the algorithm to take a long time to improve $K^{\star}$ or even not improve it at all.
For this purpose, one needs an appropriate termination criterion.
First, we want to avoid that the size of the current graph $K$ distances too much from that of the best graph $K^{\star}$.
Therefore we abort the algorithm as soon as the size of the current
graph exceeds the size of the best graph by a factor
$\alpha~\in~[1,\infty)$, that is if $|V(K)| \geq \alpha \cdot V(K^{\star})$.
Additionally, we also count the number of unsuccessful iterations, i.e. iterations in which the new graph has been rejected.
Our second criterion aborts the algorithm if this value exceeds some constant~$X~\in~[1,\infty)$.

\section{Experimental Evaluation}
\label{sec:experiments}
We now evaluate the impact and performance our preprocessing methods: 
First, we compare the performance of our algorithms with the two configurations used in the branch-and-reduce framework of Lamm~\etal~\cite{lamm-2019}. For this purpose, we examine the sizes of the reduced graphs, the number of instances solved, as well as the time required to do so. Second, we perform a broader comparison with other state-of-the-art algorithms, including heuristic approaches.

\subsection{Methodology and Setup.}\label{methodology}
We ran all the experiments on a machine with four Octa-Core Intel Xeon E5-4640 processors running at 2.4 GHz, 512 GB of main memory, 420 MB L3-Cache and 48256 KB L2-Cache.
The machine runs Ubuntu 18.04.4 and Linux kernel version 4.15.0-96.
All algorithms were implemented in C++11 and compiled with g++ version 7.5.0 with optimization flag -O3.
All algorithms were executed sequentially with a time limit of \numprint{1000} seconds.
The experiments for heuristic methods were performed with five different random seeds.
Furthermore, we present \emph{cactus plots} types of plots, which show the number of instances solved over time. 

\subsubsection{Algorithm Configuration.}\label{algorithms}
For our evaluation, we use both the non-increasing algorithm, as well as the cyclic blow-up algorithm.
In particular, we use two different configurations of the cyclic blow-up algorithm:
The first configuration, called~\conStrong, aims to achieve small reduced graphs.
For this purpose, we set the number of unsuccessful blow-up phases to
$X=64$, the number of vertices that a struction is allowed to create
to $n_{max}=$ \numprint{2048} and the maximum struction degree (the degree up to which we can apply structions)
$d_{max}=512$. 
In our preliminary experiments, this configuration was always able to compute the smallest reduced graphs.
The second configuration, called \conFast, aims to achieve a good trade-off between the reduced graph size and the time required to compute an optimal solution.
Thus we set $X=25$, $n_{max}=512$ and $d_{max}=64$. 
Finally, all our algorithms use $\beta=2$ for the tightness-check during vertex selection, as well as the the extended weighted struction, as
this struction variant achieved the best performance during preliminary experiments.

To measure the impact of our preprocessing methods on existing solvers, we add each configuration to the branch-and-reduce framework by Lamm~\etal~\cite{lamm-2019}.
This results in three solvers, which we call \cyclicFast, \cyclicStrong\ and \nonIncreasing\ in the following.
Note that each solver uses the corresponding configuration only for its initial preprocessing, whereas subsequent graph reductions only use decreasing transformations.
Finally, we have replaced the ILS local search used in the original
framework to compute lower bounds with the hybrid iterated local search (\hils) of
Nogueira~\etal~\cite{hybrid-ils-2018}. 
This resulted in slightly better runtimes during preliminary experiments, but had no impact on the number of instances that were solved.
 
\subsubsection{Instances.}\label{datasets}
We test all algorithms on a large corpus of sparse graphs, which mainly consists of instances found in previous works on the maximum (weight) independent set problem~\cite{cai-dynwvc,lamm-2019,andrade-2012}.
In particular, this corpus consists of real-world conflict graphs that were derived from OpenStreetMap (OSM)~\cite{OSMWEB,barth-2016,cai-dynwvc}, as well as a collection of large social networks from the Stanford Large Network Dataset Repository (SNAP)~\cite{snapnets}. 
One caveat of this corpus is that most its instances are unweighted.
Following the example of previous work~\cite{cai-dynwvc,li2017efficient,lamm-2019}, we alleviate this issue by assigning each vertex a random weight that is uniformly distributed in the interval $[1,200]$.
Furthermore, we extend this set of benchmark instances by also considering instances stemming from dual graphs of well-known triangle meshes (mesh)~\cite{sander2008efficient}, 3d meshes derived from simulations using the finite element method (FE), as well as instances for the maximum weight clique problem~\cite{glasgowinstances}.
However, the complements of the maximum weight clique instances are only somewhat sparse---and most are irreducible by our techniques.
This behavior has already been observed by Akiba and Iwata~\cite{akiba-tcs-2016} on similar instances. Therefore, we will omit these instances from our experiments.
An overview of all instances considered is given in Appendix~\ref{properties}.

\subsection{Comparison with Branch-and-Reduce.}\label{bnr_comp}

We begin by comparing our three solvers presented in Section~\ref{algorithms} with the two configurations, called \basicSparse\ and \basicDense\, of the branch-and-reduce framework outlined by Lamm~\etal~\cite{lamm-2019}.
Our comparison is divided into two parts: First, we consider the sizes of the irreducible graphs after the initial reduction phase. Second, we compare the number of solved instances and the time required to solve them.
A complete overview of the reduced graph sizes and running times for each algorithm is given in Appendix~\ref{bnr_tables}.

Table~\ref{graph_size_table_all} shows the sizes of the irreducible graphs after the initial reduction phase.
Note that we omit \basicDense\ as it always calculates equally sized or larger graphs than \basicSparse.

First, we note that with the exception of \texttt{fe\_ocean}, \cyclicStrong\ always produces the smallest reduced graphs.
For this particular instance, the usage of the struction limits the efficiency of the critical set reduction, resulting in a larger reduced graphs.
Furthermore, the greatest improvement can be found on the mesh instances, which are all completely reduced using \cyclicStrong.
In comparison, \basicSparse\ is not able to obtain an empty graph on a single of these instances and ends up with reduced graphs consisting of up to thousands of vertices.
Overall, \cyclicStrong\ is able to achieve an empty reduced graph on $60$ of the $87$ instances tested \---
an additional $48$ instances compared to the $22$ empty reduced graphs computed by \basicSparse.

If we compare the reduced graphs of~\cyclicStrong\ and \cyclicFast\, we see that they always have the same size on the mesh instances. 
However, the size of the reduced instances computed by~\cyclicFast\ on the other instance families is up to a few thousand vertices larger.
On the OSM instances for example, \cyclicFast\ calculates a reduced graph that has the same size as the one computed by \cyclicStrong\ on only~\numprint{16} out of~\numprint{34} instances with the largest difference being~\numprint{2216} vertices.

\begin{table*}[t]
\centering
\vspace*{-.5cm}
\scriptsize
\begin{ThreePartTable}
\begin{tabular}{ccccccccccc}
  \mc{1}{l|}{Graph} & \mc{1}{r}{$n$} & \mc{1}{r|}{$t_{r}$} & \mc{1}{r}{$n$} & \mc{1}{r|}{$t_{r}$} & \mc{1}{r}{$n$} & \mc{1}{r|}{$t_{r}$} & \mc{1}{r}{$n$} & \mc{1}{r|}{$t_{r}$} & \mc{1}{r}{$n$} & \mc{1}{r}{$t_{r}$} \\ 
  \hline

  \mc{1}{l|}{OSM instances} & \mc{2}{c|}{\basicDense} & \mc{2}{c|}{\basicSparse} & \mc{2}{c|}{\nonIncreasing} & \mc{2}{c|}{\cyclicFast} & \mc{2}{c}{\cyclicStrong} \\  
  \hline

  \rowcolor{lightergray} \mc{1}{l|}{\texttt{alabama-AM2}} & \mc{1}{r}{173} & \mc{1}{r|}{0.06} & \mc{1}{r}{173} & \mc{1}{r|}{0.07} & \mc{1}{r}{0} & \mc{1}{r|}{0.01} & \mc{1}{r}{0} & \mc{1}{r|}{0.01} & \mc{1}{r}{0} & \mc{1}{r}{0.01} \\  
  \mc{1}{l|}{\texttt{district-of-columbia-AM2}} & \mc{1}{r}{\numprint{6360}} & \mc{1}{r|}{11.86} & \mc{1}{r}{\numprint{6360}} & \mc{1}{r|}{14.39} & \mc{1}{r}{\numprint{5606}} & \mc{1}{r|}{0.85} & \mc{1}{r}{\numprint{1855}} & \mc{1}{r|}{2.51} & \mc{1}{r}{\numprint{1484}} & \mc{1}{r}{84.91} \\  
  \mc{1}{l|}{\texttt{florida-AM3}} & \mc{1}{r}{\numprint{1069}} & \mc{1}{r|}{31.52} & \mc{1}{r}{\numprint{1069}} & \mc{1}{r|}{35.20} & \mc{1}{r}{814} & \mc{1}{r|}{0.13} & \mc{1}{r}{661} & \mc{1}{r|}{0.44} & \mc{1}{r}{267} & \mc{1}{r}{42.26} \\  
  \mc{1}{l|}{\texttt{georgia-AM3}} & \mc{1}{r}{861} & \mc{1}{r|}{8.99} & \mc{1}{r}{861} & \mc{1}{r|}{10.14} & \mc{1}{r}{796} & \mc{1}{r|}{0.08} & \mc{1}{r}{587} & \mc{1}{r|}{0.69} & \mc{1}{r}{425} & \mc{1}{r}{12.84} \\  
  \mc{1}{l|}{\texttt{greenland-AM3}} & \mc{1}{r}{\numprint{3942}} & \mc{1}{r|}{3.81} & \mc{1}{r}{\numprint{3942}} & \mc{1}{r|}{24.77} & \mc{1}{r}{\numprint{3953}} & \mc{1}{r|}{3.94} & \mc{1}{r}{\numprint{3339}} & \mc{1}{r|}{10.27} & \mc{1}{r}{\numprint{3339}} & \mc{1}{r}{54.44} \\  
  \rowcolor{lightergray} \mc{1}{l|}{\texttt{new-hampshire-AM3}} & \mc{1}{r}{247} & \mc{1}{r|}{4.99} & \mc{1}{r}{247} & \mc{1}{r|}{5.69} & \mc{1}{r}{164} & \mc{1}{r|}{0.02} & \mc{1}{r}{0} & \mc{1}{r|}{0.07} & \mc{1}{r}{0} & \mc{1}{r}{0.09} \\  
  \rowcolor{lightergray} \mc{1}{l|}{\texttt{rhode-island-AM2}} & \mc{1}{r}{\numprint{1103}} & \mc{1}{r|}{0.55} & \mc{1}{r}{\numprint{1103}} & \mc{1}{r|}{0.68} & \mc{1}{r}{845} & \mc{1}{r|}{0.17} & \mc{1}{r}{0} & \mc{1}{r|}{0.53} & \mc{1}{r}{0} & \mc{1}{r}{4.57} \\  
  \rowcolor{lightergray} \mc{1}{l|}{\texttt{utah-AM3}} & \mc{1}{r}{568} & \mc{1}{r|}{8.21} & \mc{1}{r}{568} & \mc{1}{r|}{8.97} & \mc{1}{r}{396} & \mc{1}{r|}{0.03} & \mc{1}{r}{0} & \mc{1}{r|}{0.09} & \mc{1}{r}{0} & \mc{1}{r}{0.40} \\  

  \hline 
  \mc{1}{l|}{Empty graphs} & \mc{2}{r|}{0\% (0/34)} & \mc{2}{r|}{0\% (0/34)} & \mc{2}{r|}{11.8\% (4/34)} & \mc{2}{r|}{41.2\% (14/34)} & \mc{2}{r}{50\% (17/34)} \\  
  \hline \hline

  \mc{1}{l|}{SNAP instances} & \mc{2}{c|}{\basicDense} & \mc{2}{c|}{\basicSparse} & \mc{2}{c|}{\nonIncreasing} & \mc{2}{c|}{\cyclicFast} & \mc{2}{c}{\cyclicStrong} \\  
  \hline

  \mc{1}{l|}{\texttt{as-skitter}} & \mc{1}{r}{\numprint{26584}} & \mc{1}{r|}{25.82} & \mc{1}{r}{\numprint{8585}} & \mc{1}{r|}{36.69} & \mc{1}{r}{\numprint{3426}} & \mc{1}{r|}{4.75} & \mc{1}{r}{\numprint{2782}} & \mc{1}{r|}{5.50} & \mc{1}{r}{\numprint{2343}} & \mc{1}{r}{6.80} \\  
  \rowcolor{lightergray} \mc{1}{l|}{\texttt{ca-AstroPh}} & \mc{1}{r}{0} & \mc{1}{r|}{0.02} & \mc{1}{r}{0} & \mc{1}{r|}{0.02} & \mc{1}{r}{0} & \mc{1}{r|}{0.02} & \mc{1}{r}{0} & \mc{1}{r|}{0.03} & \mc{1}{r}{0} & \mc{1}{r}{0.03} \\  
  \rowcolor{lightergray} \mc{1}{l|}{\texttt{email-EuAll}} & \mc{1}{r}{0} & \mc{1}{r|}{0.08} & \mc{1}{r}{0} & \mc{1}{r|}{0.09} & \mc{1}{r}{0} & \mc{1}{r|}{0.06} & \mc{1}{r}{0} & \mc{1}{r|}{0.09} & \mc{1}{r}{0} & \mc{1}{r}{0.07} \\  
  \rowcolor{lightergray} \mc{1}{l|}{\texttt{p2p-Gnutella06}} & \mc{1}{r}{0} & \mc{1}{r|}{0.01} & \mc{1}{r}{0} & \mc{1}{r|}{0.01} & \mc{1}{r}{0} & \mc{1}{r|}{0.01} & \mc{1}{r}{0} & \mc{1}{r|}{0.01} & \mc{1}{r}{0} & \mc{1}{r}{0.01} \\  
  \rowcolor{lightergray} \mc{1}{l|}{\texttt{roadNet-PA}} & \mc{1}{r}{\numprint{133814}} & \mc{1}{r|}{2.43} & \mc{1}{r}{\numprint{35442}} & \mc{1}{r|}{7.73} & \mc{1}{r}{300} & \mc{1}{r|}{1.05} & \mc{1}{r}{0} & \mc{1}{r|}{1.19} & \mc{1}{r}{0} & \mc{1}{r}{1.14} \\  
  \mc{1}{l|}{\texttt{soc-LiveJournal1}} & \mc{1}{r}{\numprint{60041}} & \mc{1}{r|}{236.88} & \mc{1}{r}{\numprint{29508}} & \mc{1}{r|}{213.74} & \mc{1}{r}{\numprint{4319}} & \mc{1}{r|}{22.27} & \mc{1}{r}{\numprint{3530}} & \mc{1}{r|}{24.13} & \mc{1}{r}{\numprint{1314}} & \mc{1}{r}{37.77} \\  
  \mc{1}{l|}{\texttt{web-Google}} & \mc{1}{r}{\numprint{2810}} & \mc{1}{r|}{1.57} & \mc{1}{r}{\numprint{1254}} & \mc{1}{r|}{2.42} & \mc{1}{r}{361} & \mc{1}{r|}{1.75} & \mc{1}{r}{46} & \mc{1}{r|}{1.88} & \mc{1}{r}{46} & \mc{1}{r}{7.97} \\  
  \rowcolor{lightergray} \mc{1}{l|}{\texttt{wiki-Vote}} & \mc{1}{r}{477} & \mc{1}{r|}{0.03} & \mc{1}{r}{0} & \mc{1}{r|}{0.02} & \mc{1}{r}{0} & \mc{1}{r|}{0.02} & \mc{1}{r}{0} & \mc{1}{r|}{0.02} & \mc{1}{r}{0} & \mc{1}{r}{0.02} \\  

  \hline 
  \mc{1}{l|}{Empty graphs} & \mc{2}{r|}{58.1\% (18/31)} & \mc{2}{r|}{67.7\% (21/31)} & \mc{2}{r|}{67.7\% (21/34)} & \mc{2}{r|}{80.6\% (25/31)} & \mc{2}{r}{80.6\% (25/31)} \\  
  \hline \hline

  \mc{1}{l|}{mesh instances} & \mc{2}{c|}{\basicDense} & \mc{2}{c|}{\basicSparse} & \mc{2}{c|}{\nonIncreasing} & \mc{2}{c|}{\cyclicFast} & \mc{2}{c}{\cyclicStrong} \\  
  \hline

  \rowcolor{lightergray} \mc{1}{l|}{\texttt{buddha}} & \mc{1}{r}{\numprint{380315}} & \mc{1}{r|}{5.56} & \mc{1}{r}{\numprint{107265}} & \mc{1}{r|}{26.19} & \mc{1}{r}{86} & \mc{1}{r|}{1.83} & \mc{1}{r}{0} & \mc{1}{r|}{1.87} & \mc{1}{r}{0} & \mc{1}{r}{1.91} \\  
  \rowcolor{lightergray} \mc{1}{l|}{\texttt{dragon}} & \mc{1}{r}{\numprint{51885}} & \mc{1}{r|}{0.89} & \mc{1}{r}{\numprint{12893}} & \mc{1}{r|}{1.34} & \mc{1}{r}{0} & \mc{1}{r|}{0.18} & \mc{1}{r}{0} & \mc{1}{r|}{0.19} & \mc{1}{r}{0} & \mc{1}{r}{0.21} \\  
  \rowcolor{lightergray} \mc{1}{l|}{\texttt{ecat}} & \mc{1}{r}{\numprint{239787}} & \mc{1}{r|}{4.07} & \mc{1}{r}{\numprint{26270}} & \mc{1}{r|}{10.09} & \mc{1}{r}{274} & \mc{1}{r|}{2.12} & \mc{1}{r}{0} & \mc{1}{r|}{2.12} & \mc{1}{r}{0} & \mc{1}{r}{2.14} \\  
  \hline 
  \mc{1}{l|}{Empty graphs} & \mc{2}{r|}{0\% (0/15)} & \mc{2}{r|}{0\% (0/15)} & \mc{2}{r|}{66.7\% (10/15)} & \mc{2}{r|}{100\% (15/15)} & \mc{2}{r}{100\% (15/15)} \\  
  \hline \hline

  \mc{1}{l|}{FE instances} & \mc{2}{c|}{\basicDense} & \mc{2}{c|}{\basicSparse} & \mc{2}{c|}{\nonIncreasing} & \mc{2}{c|}{\cyclicFast} & \mc{2}{c}{\cyclicStrong} \\  
  \hline  

  \mc{1}{l|}{\texttt{fe\_ocean}} & \mc{1}{r}{\numprint{141283}} & \mc{1}{r|}{1.05} & \mc{1}{r}{0} & \mc{1}{r|}{5.94} & \mc{1}{r}{\numprint{138338}} & \mc{1}{r|}{8.90} & \mc{1}{r}{\numprint{138134}} & \mc{1}{r|}{9.61} & \mc{1}{r}{\numprint{138049}} & \mc{1}{r}{10.78} \\  
  \rowcolor{lightergray} \mc{1}{l|}{\texttt{fe\_sphere}} & \mc{1}{r}{\numprint{15269}} & \mc{1}{r|}{0.21} & \mc{1}{r}{\numprint{15269}} & \mc{1}{r|}{1.47} & \mc{1}{r}{\numprint{2961}} & \mc{1}{r|}{0.34} & \mc{1}{r}{147} & \mc{1}{r|}{0.62} & \mc{1}{r}{0} & \mc{1}{r}{0.75} \\  
  \hline

  \mc{1}{l|}{Empty graphs} & \mc{2}{r|}{0\% (0/7)} & \mc{2}{r|}{14.3\% (1/7)} & \mc{2}{r|}{0\% (0/7)} & \mc{2}{r|}{28.6\% (2/7)} & \mc{2}{r}{42.9\% (3/7)} \\  
\end{tabular}
\end{ThreePartTable}

\caption{Smallest irreducible graph found by each algorithm and time
  (in seconds) required to compute it. Rows are highlighted in gray if
  one of our algorithms is able to obtain an empty graph. 
}
\label{graph_size_table_all}
\end{table*}

\begin{figure*}[t]
	\centering
  \includegraphics[width=.85\linewidth]{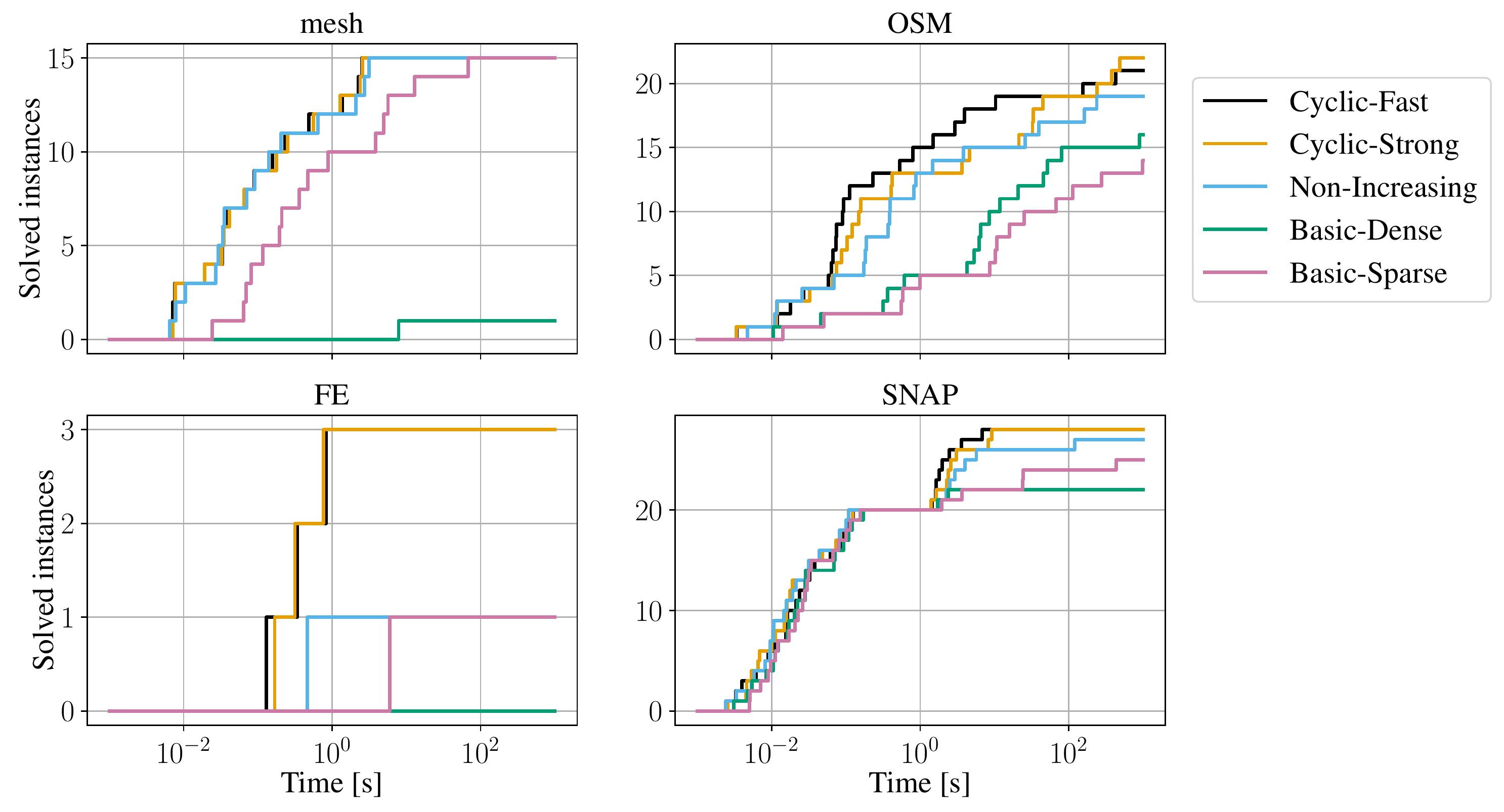}
	\caption{Cactus plots for the different instance families and evaluated solvers.}
	\label{cactus_plot}
\end{figure*}
\begin{table*}[t!]
\centering
\scriptsize
\begin{ThreePartTable}
\begin{tabu}{ccccccccc}

  \mc{1}{l}{Graph} & \mc{1}{|r}{$t_{max}$} & \mc{1}{r}{$w_{max}$} & \mc{1}{|r}{$t_{max}$} & \mc{1}{r}{$w_{max}$} & \mc{1}{|r}{$t_{max}$} & \mc{1}{r}{$w_{max}$} & \mc{1}{|r}{$t_{max}$} & \mc{1}{r}{$w_{max}$} \\ 
  \hline  

  \mc{1}{l}{OSM instances} & \mc{2}{|c}{DynWVC2} & \mc{2}{|c}{HILS} & \mc{2}{|c}{Cyclic-Fast} & \mc{2}{|c}{Cyclic-Strong} \\ 
  \hline  

  \mc{1}{l}{\cellcolor{lightgray!50}{\texttt{alabama-AM2}}} & \mc{1}{|r}{\cellcolor{lightgray!50}{0.24}} & \mc{1}{r}{\cellcolor{lightgray!50}{\numprint{174269}}} & \mc{1}{|r}{\cellcolor{lightgray!50}{0.03}} & \mc{1}{r}{\textbf{\cellcolor{lightgray!50}{\numprint{174309}}}} & \mc{1}{|r}{\cellcolor{lightgray!50}{0.01}} & \mc{1}{r}{\textbf{\cellcolor{lightgray!50}{\numprint{174309}}}} & \mc{1}{|r}{\cellcolor{lightgray!50}{0.01}} & \mc{1}{r}{\textbf{\cellcolor{lightgray!50}{\numprint{174309}}}} \\  

  \mc{1}{l}{\texttt{district-of-columbia-AM2}} & \mc{1}{|r}{915.18} & \mc{1}{r}{\numprint{208977}} & \mc{1}{|r}{400.69} & \mc{1}{r}{\textbf{\numprint{209132}}} & \mc{1}{|r}{4.21} & \mc{1}{r}{\textbf{\numprint{209132}}} & \mc{1}{|r}{84.21} & \mc{1}{r}{\numprint{209131}} \\  

  \mc{1}{l}{\cellcolor{lightgray!50}{\texttt{florida-AM3}}} & \mc{1}{|r}{\cellcolor{lightgray!50}{862.04}} & \mc{1}{r}{\cellcolor{lightgray!50}{\numprint{237120}}} & \mc{1}{|r}{\cellcolor{lightgray!50}{3.98}} & \mc{1}{r}{\textbf{\cellcolor{lightgray!50}{\numprint{237333}}}} & \mc{1}{|r}{\cellcolor{lightgray!50}{1.57}} & \mc{1}{r}{\textbf{\cellcolor{lightgray!50}{\numprint{237333}}}} & \mc{1}{|r}{\cellcolor{lightgray!50}{40.97}} & \mc{1}{r}{\textbf{\cellcolor{lightgray!50}{\numprint{237333}}}} \\  

  \mc{1}{l}{\cellcolor{lightgray!50}{\texttt{georgia-AM3}}} & \mc{1}{|r}{\cellcolor{lightgray!50}{1.31}} & \mc{1}{r}{\textbf{\cellcolor{lightgray!50}{\numprint{222652}}}} & \mc{1}{|r}{\cellcolor{lightgray!50}{0.04}} & \mc{1}{r}{\textbf{\cellcolor{lightgray!50}{\numprint{222652}}}} & \mc{1}{|r}{\cellcolor{lightgray!50}{0.98}} & \mc{1}{r}{\textbf{\cellcolor{lightgray!50}{\numprint{222652}}}} & \mc{1}{|r}{\cellcolor{lightgray!50}{12.97}} & \mc{1}{r}{\textbf{\cellcolor{lightgray!50}{\numprint{222652}}}} \\  

  \mc{1}{l}{\texttt{greenland-AM3}} & \mc{1}{|r}{640.46} & \mc{1}{r}{\numprint{14010}} & \mc{1}{|r}{1.18} & \mc{1}{r}{\textbf{\numprint{14011}}} & \mc{1}{|r}{10.95} & \mc{1}{r}{\textbf{\numprint{14011}}} & \mc{1}{|r}{58.24} & \mc{1}{r}{\numprint{14008}} \\  

  \mc{1}{l}{\cellcolor{lightgray!50}{\texttt{new-hampshire-AM3}}} & \mc{1}{|r}{\cellcolor{lightgray!50}{1.63}} & \mc{1}{r}{\textbf{\cellcolor{lightgray!50}{\numprint{116060}}}} & \mc{1}{|r}{\cellcolor{lightgray!50}{0.03}} & \mc{1}{r}{\textbf{\cellcolor{lightgray!50}{\numprint{116060}}}} & \mc{1}{|r}{\cellcolor{lightgray!50}{0.05}} & \mc{1}{r}{\textbf{\cellcolor{lightgray!50}{\numprint{116060}}}} & \mc{1}{|r}{\cellcolor{lightgray!50}{0.08}} & \mc{1}{r}{\textbf{\cellcolor{lightgray!50}{\numprint{116060}}}} \\  

  \mc{1}{l}{\cellcolor{lightgray!50}{\texttt{rhode-island-AM2}}} & \mc{1}{|r}{\cellcolor{lightgray!50}{13.90}} & \mc{1}{r}{\cellcolor{lightgray!50}{\numprint{184576}}} & \mc{1}{|r}{\cellcolor{lightgray!50}{0.24}} & \mc{1}{r}{\textbf{\cellcolor{lightgray!50}{\numprint{184596}}}} & \mc{1}{|r}{\cellcolor{lightgray!50}{0.41}} & \mc{1}{r}{\textbf{\cellcolor{lightgray!50}{\numprint{184596}}}} & \mc{1}{|r}{\cellcolor{lightgray!50}{4.37}} & \mc{1}{r}{\textbf{\cellcolor{lightgray!50}{\numprint{184596}}}} \\  

  \mc{1}{l}{\cellcolor{lightgray!50}{\texttt{utah-AM3}}} & \mc{1}{|r}{\cellcolor{lightgray!50}{136.90}} & \mc{1}{r}{\textbf{\cellcolor{lightgray!50}{\numprint{98847}}}} & \mc{1}{|r}{\cellcolor{lightgray!50}{0.07}} & \mc{1}{r}{\textbf{\cellcolor{lightgray!50}{\numprint{98847}}}} & \mc{1}{|r}{\cellcolor{lightgray!50}{0.09}} & \mc{1}{r}{\textbf{\cellcolor{lightgray!50}{\numprint{98847}}}} & \mc{1}{|r}{\cellcolor{lightgray!50}{0.27}} & \mc{1}{r}{\textbf{\cellcolor{lightgray!50}{\numprint{98847}}}} \\ 
  \hline  

  \mc{1}{l}{Solved instances} & \mc{2}{|r}{} & \mc{2}{|r}{} & \mc{2}{|r}{61.8\% (21/34)} & \mc{2}{|r}{64.7\% (22/34)} \\  

  \mc{1}{l}{Optimal weight} & \mc{2}{|r}{68.2\% (15/22)} & \mc{2}{|r}{100.0\% (22/22)} & \mc{2}{|r}{} & \mc{2}{|r}{} \\ 
  \hline \hline  

  \mc{1}{l}{SNAP instances} & \mc{2}{|c}{DynWVC2} & \mc{2}{|c}{HILS} & \mc{2}{|c}{Cyclic-Fast} & \mc{2}{|c}{Cyclic-Strong} \\ 
  \hline  

  \mc{1}{l}{\texttt{as-skitter}} & \mc{1}{|r}{383.97} & \mc{1}{r}{\numprint{123273938}} & \mc{1}{|r}{999.32} & \mc{1}{r}{\numprint{122658804}} & \mc{1}{|r}{346.69} & \mc{1}{r}{\numprint{124137148}} & \mc{1}{|r}{354.71} & \mc{1}{r}{\textbf{\numprint{124137365}}} \\  

  \mc{1}{l}{\cellcolor{lightgray!50}{\texttt{ca-AstroPh}}} & \mc{1}{|r}{\cellcolor{lightgray!50}{125.05}} & \mc{1}{r}{\cellcolor{lightgray!50}{\numprint{797480}}} & \mc{1}{|r}{\cellcolor{lightgray!50}{13.47}} & \mc{1}{r}{\textbf{\cellcolor{lightgray!50}{\numprint{797510}}}} & \mc{1}{|r}{\cellcolor{lightgray!50}{0.02}} & \mc{1}{r}{\textbf{\cellcolor{lightgray!50}{\numprint{797510}}}} & \mc{1}{|r}{\cellcolor{lightgray!50}{0.02}} & \mc{1}{r}{\textbf{\cellcolor{lightgray!50}{\numprint{797510}}}} \\  

  \mc{1}{l}{\cellcolor{lightgray!50}{\texttt{email-EuAl}l}} & \mc{1}{|r}{\cellcolor{lightgray!50}{132.62}} & \mc{1}{r}{\textbf{\cellcolor{lightgray!50}{\numprint{25286322}}}} & \mc{1}{|r}{\cellcolor{lightgray!50}{338.14}} & \mc{1}{r}{\textbf{\cellcolor{lightgray!50}{\numprint{25286322}}}} & \mc{1}{|r}{\cellcolor{lightgray!50}{0.07}} & \mc{1}{r}{\textbf{\cellcolor{lightgray!50}{\numprint{25286322}}}} & \mc{1}{|r}{\cellcolor{lightgray!50}{0.07}} & \mc{1}{r}{\textbf{\cellcolor{lightgray!50}{\numprint{25286322}}}} \\  

  \mc{1}{l}{\cellcolor{lightgray!50}{\texttt{p2p-Gnutella06}}} & \mc{1}{|r}{\cellcolor{lightgray!50}{186.97}} & \mc{1}{r}{\cellcolor{lightgray!50}{\numprint{548611}}} & \mc{1}{|r}{\cellcolor{lightgray!50}{1.29}} & \mc{1}{r}{\textbf{\cellcolor{lightgray!50}{\numprint{548612}}}} & \mc{1}{|r}{\cellcolor{lightgray!50}{0.01}} & \mc{1}{r}{\textbf{\cellcolor{lightgray!50}{\numprint{548612}}}} & \mc{1}{|r}{\cellcolor{lightgray!50}{0.01}} & \mc{1}{r}{\textbf{\cellcolor{lightgray!50}{\numprint{548612}}}} \\  

  \mc{1}{l}{\cellcolor{lightgray!50}{\texttt{roadNet-PA}}} & \mc{1}{|r}{\cellcolor{lightgray!50}{469.18}} & \mc{1}{r}{\cellcolor{lightgray!50}{\numprint{60990177}}} & \mc{1}{|r}{\cellcolor{lightgray!50}{999.94}} & \mc{1}{r}{\cellcolor{lightgray!50}{\numprint{60037011}}} & \mc{1}{|r}{\cellcolor{lightgray!50}{0.96}} & \mc{1}{r}{\textbf{\cellcolor{lightgray!50}{\numprint{61731589}}}} & \mc{1}{|r}{\cellcolor{lightgray!50}{1.04}} & \mc{1}{r}{\textbf{\cellcolor{lightgray!50}{\numprint{61731589}}}} \\  

  \mc{1}{l}{\texttt{soc-LiveJournal1}} & \mc{1}{|r}{999.99} & \mc{1}{r}{\numprint{279231875}} & \mc{1}{|r}{\numprint{1000.00}} & \mc{1}{r}{\numprint{255079926}} & \mc{1}{|r}{51.33} & \mc{1}{r}{\numprint{284036222}} & \mc{1}{|r}{44.19} & \mc{1}{r}{\textbf{\numprint{284036239}}} \\  

  \mc{1}{l}{\cellcolor{lightgray!50}{\texttt{web-Google}}} & \mc{1}{|r}{\cellcolor{lightgray!50}{324.65}} & \mc{1}{r}{\cellcolor{lightgray!50}{\numprint{56206250}}} & \mc{1}{|r}{\cellcolor{lightgray!50}{995.92}} & \mc{1}{r}{\cellcolor{lightgray!50}{\numprint{56008278}}} & \mc{1}{|r}{\cellcolor{lightgray!50}{1.72}} & \mc{1}{r}{\textbf{\cellcolor{lightgray!50}{\numprint{56326504}}}} & \mc{1}{|r}{\cellcolor{lightgray!50}{6.44}} & \mc{1}{r}{\textbf{\cellcolor{lightgray!50}{\numprint{56326504}}}} \\  

  \mc{1}{l}{\cellcolor{lightgray!50}{\texttt{wiki-Vote}}} & \mc{1}{|r}{\cellcolor{lightgray!50}{0.32}} & \mc{1}{r}{\textbf{\cellcolor{lightgray!50}{\numprint{500079}}}} & \mc{1}{|r}{\cellcolor{lightgray!50}{10.34}} & \mc{1}{r}{\textbf{\cellcolor{lightgray!50}{\numprint{500079}}}} & \mc{1}{|r}{\cellcolor{lightgray!50}{0.02}} & \mc{1}{r}{\textbf{\cellcolor{lightgray!50}{\numprint{500079}}}} & \mc{1}{|r}{\cellcolor{lightgray!50}{0.02}} & \mc{1}{r}{\textbf{\cellcolor{lightgray!50}{\numprint{500079}}}} \\ 
  \hline  

  \mc{1}{l}{Solved instances} & \mc{2}{|r}{} & \mc{2}{|r}{} & \mc{2}{|r}{90.3\% (28/31)} & \mc{2}{|r}{90.3\% (28/31)} \\  

  \mc{1}{l}{Optimal weight} & \mc{2}{|r}{28.6\% (8/28)} & \mc{2}{|r}{57.1\% (16/28)} & \mc{2}{|r}{} & \mc{2}{|r}{} \\ 
  \hline \hline  

  \mc{1}{l}{mesh instances} & \mc{2}{|c}{DynWVC2} & \mc{2}{|c}{HILS} & \mc{2}{|c}{Cyclic-Fast} & \mc{2}{|c}{Cyclic-Strong} \\ 
  \hline  

  \mc{1}{l}{\cellcolor{lightgray!50}{\texttt{buddha}}} & \mc{1}{|r}{\cellcolor{lightgray!50}{797.35}} & \mc{1}{r}{\cellcolor{lightgray!50}{\numprint{56757052}}} & \mc{1}{|r}{\cellcolor{lightgray!50}{999.94}} & \mc{1}{r}{\cellcolor{lightgray!50}{\numprint{55490134}}} & \mc{1}{|r}{\cellcolor{lightgray!50}{1.75}} & \mc{1}{r}{\textbf{\cellcolor{lightgray!50}{\numprint{57555880}}}} & \mc{1}{|r}{\cellcolor{lightgray!50}{1.77}} & \mc{1}{r}{\textbf{\cellcolor{lightgray!50}{\numprint{57555880}}}} \\  

  \mc{1}{l}{\cellcolor{lightgray!50}{\texttt{dragon}}} & \mc{1}{|r}{\cellcolor{lightgray!50}{981.51}} & \mc{1}{r}{\cellcolor{lightgray!50}{\numprint{7944042}}} & \mc{1}{|r}{\cellcolor{lightgray!50}{996.01}} & \mc{1}{r}{\cellcolor{lightgray!50}{\numprint{7940422}}} & \mc{1}{|r}{\cellcolor{lightgray!50}{0.21}} & \mc{1}{r}{\textbf{\cellcolor{lightgray!50}{\numprint{7956530}}}} & \mc{1}{|r}{\cellcolor{lightgray!50}{0.22}} & \mc{1}{r}{\textbf{\cellcolor{lightgray!50}{\numprint{7956530}}}} \\  

  \mc{1}{l}{\cellcolor{lightgray!50}{\texttt{ecat}}} & \mc{1}{|r}{\cellcolor{lightgray!50}{542.87}} & \mc{1}{r}{\cellcolor{lightgray!50}{\numprint{36129804}}} & \mc{1}{|r}{\cellcolor{lightgray!50}{999.91}} & \mc{1}{r}{\cellcolor{lightgray!50}{\numprint{35512644}}} & \mc{1}{|r}{\cellcolor{lightgray!50}{2.19}} & \mc{1}{r}{\textbf{\cellcolor{lightgray!50}{\numprint{36650298}}}} & \mc{1}{|r}{\cellcolor{lightgray!50}{2.29}} & \mc{1}{r}{\textbf{\cellcolor{lightgray!50}{\numprint{36650298}}}} \\ 
  \hline  

  \mc{1}{l}{Solved instances} & \mc{2}{|r}{} & \mc{2}{|r}{} & \mc{2}{|r}{100.0\% (15/15)} & \mc{2}{|r}{100.0\% (15/15)} \\  

  \mc{1}{l}{Optimal weight} & \mc{2}{|r}{0.0\% (0/15)} & \mc{2}{|r}{0.0\% (0/15)} & \mc{2}{|r}{} & \mc{2}{|r}{} \\ 
  \hline \hline  

  \mc{1}{l}{FE instances} & \mc{2}{|c}{DynWVC1} & \mc{2}{|c}{HILS} & \mc{2}{|c}{Cyclic-Fast} & \mc{2}{|c}{Cyclic-Strong} \\ 
  \hline  

  \mc{1}{l}{\texttt{fe\_ocean}} & \mc{1}{|r}{983.53} & \mc{1}{r}{\textbf{\numprint{7222521}}} & \mc{1}{|r}{999.57} & \mc{1}{r}{\numprint{7069279}} & \mc{1}{|r}{18.85} & \mc{1}{r}{\numprint{6591832}} & \mc{1}{|r}{19.04} & \mc{1}{r}{\numprint{6591537}} \\  

  \mc{1}{l}{\cellcolor{lightgray!50}{\texttt{fe\_sphere}}} & \mc{1}{|r}{\cellcolor{lightgray!50}{875.87}} & \mc{1}{r}{\cellcolor{lightgray!50}{\numprint{616978}}} & \mc{1}{|r}{\cellcolor{lightgray!50}{843.67}} & \mc{1}{r}{\cellcolor{lightgray!50}{\numprint{616528}}} & \mc{1}{|r}{\cellcolor{lightgray!50}{0.63}} & \mc{1}{r}{\textbf{\cellcolor{lightgray!50}{\numprint{617816}}}} & \mc{1}{|r}{\cellcolor{lightgray!50}{0.67}} & \mc{1}{r}{\textbf{\cellcolor{lightgray!50}{\numprint{617816}}}} \\ 
  \hline  

  \mc{1}{l}{Solved instances} & \mc{2}{|r}{} & \mc{2}{|r}{} & \mc{2}{|r}{42.9\% (3/7)} & \mc{2}{|r}{42.9\% (3/7)} \\  

  \mc{1}{l}{Optimal weight} & \mc{2}{|r}{0.0\% (0/3)} & \mc{2}{|r}{0.0\% (0/3)} & \mc{2}{|r}{} & \mc{2}{|r}{} \\  

\end{tabu}
\end{ThreePartTable}

\caption{Best solution found by each algorithm and time (in seconds) required to compute it. The global best solution is highlighted in bold. Rows are highlighted in gray if one of our exact solvers is able to solve the corresponding instances.}
\label{best_weight_table_all}
\end{table*}

Next, we examine the number of solved instances and the time required to solve them.
For this purpose, Figure~\ref{cactus_plot} shows cactus plots for the number of solved instances over time.
First, we can see that \cyclicStrong\ was able to solve the most instances overall (\numprint{68} out of \numprint{87} instances). 
To be more specific, \cyclicStrong\ was able to solve an additional \numprint{11} instances compared to \basicSparse\ and \basicDense.
Of these newly solved instances six are from the OSM family, three from the SNAP family and one additional instance from the FE family.

Comparing the time that our algorithms require to solve the instances with \basicSparse\ and \basicDense, we can see improvements on almost all instances.
Our \cyclicFast\ algorithm is able to find solutions up to an order of magnitude faster than \basicSparse\ and \basicDense\ on five mesh instances,~\numprint{13} OSM instances and three SNAP instances.
On the two OSM instances \texttt{pennsylvania-AM3} and \texttt{utah-AM3} as well as the SNAP instance \texttt{roadNet-CA}, we are up to two orders of magnitude faster.
We attribute this increase in performance to the much smaller reduced graph size, as often a smaller graph size tends to result in finding a solution faster.
Furthermore, the generalized fold reduction that is used in \basicSparse\ and \basicDense\ tends to increase the running time.
Thus, we omitted this reduction rule from our algorithm.

\subsection{Comparison with Heuristic Approaches.}\label{sota_comp}
In the following we provide a comparison of our algorithms with heuristic state-of-the-art approaches.
For this purpose, we also include the two local searches \dynWVC\ and \hils\ in addition to the two branch-and-reduce algorithms \basicSparse\ and \basicDense.
For \dynWVC\, we use both configurations \dynWVCOne\ and \dynWVCTwo\ described by Cai~\etal~\cite{cai-dynwvc}.
For all algorithms we compare both the best achievable weighted independent set, as well as their convergence behavior regarding solution quality.
An overview of the maximum weight and the minimum time required to obtain it is given in Table~\ref{best_weight_table_all}.
Furthermore, for each exact algorithm the number of solved instances is shown, whereas for heuristic methods the number of instances on which they are also able to find an optimal solution is given.
However, note that the heuristic methods tested are not able to verify the optimality of the solution they computed.
For the individual instance families, we list either \dynWVCOne\ or \dynWVCTwo\ depending on which of the two configurations provides better performance.
Finally, we omit \basicSparse, \basicDense\ and \nonIncreasing, as these are outperformed by either \cyclicFast\ and \cyclicStrong\ as presented in the previous section.
A complete overview of the solution sizes and running times for each algorithm is given in Appendix~\ref{sota_tables}.

Considering the OSM family, we see that our algorithms are able to compute an optimal solution on~\numprint{22} of the~\numprint{34} instances tested.
However, \hils\ was also able to calculate a solution of optimal size on all these instances.
Considering the OSM family, we can see that \hils\ calculates optimal solutions on all~\numprint{22} of the~\numprint{34} instances that can be solved by our algorithm \cyclicStrong.
In contrast, \dynWVC\ is only able to do so on~\numprint{15} of the~\numprint{22} instances.
Furthermore, on ten of the remaining~\numprint{12} instances which our algorithms are not able to solve, \hils\ is able to calculate the best solution.
Finally, when comparing the time required to compute the best solution, we find that \hils\ generally performs better than the other algorithms.

Looking at the SNAP instances, we have already seen that \cyclicFast\ and \cyclicStrong\ can solve~\numprint{28} of the~\numprint{31} instances optimally.
In contrast, \hils\ can only calculate optimal solutions on~\numprint{16} of these~\numprint{28} instances.
Furthermore, \dynWVC\ is able to obtain optimal solutions on eight of the solved instances.
For the three unsolved instances, \cyclicStrong\ computes the best solution on \texttt{as-skitter} and \texttt{soc-LiveJournal1}, while \dynWVCOne\ obtains it on \texttt{soc-pokec-relationships}.
In terms of running time, we can see that both \dynWVC\ and \hils\ are often orders of magnitude slower than our algorithms in achieving their optimal solution.

On the mesh instances, we can observe a similar pattern as for the SNAP instances.
Our algorithms \cyclicFast\ and \cyclicStrong\ are able to solve all instances optimally and always need less than three seconds to obtain them.
On the other hand, none of the evaluated local searches is able to compute an optimal solution on a single instance and are slower than our algorithms by \emph{orders of magnitude}.

Finally, on the FE family, neither \dynWVC\ nor \hils\ are able to obtain a solution of equal weight on any of the three instances solved by our algorithms.
However, considering the unsolved instances, our algorithms are only able to compute the best solution on \texttt{fe\_body}.
On all remaining instances, one of the two \dynWVC\ configurations calculates the best solution.

\section{Conclusion and Future Work}
\label{sec:conclusion}
In this work, we engineered two new algorithms for finding maximum weight independent sets.
Our algorithms use novel transformations that make heavy use of the struction method.
In general, the struction method can be classified as a transformation that reduces the independence number of a graph.  
One caveat of the struction method is that it does not guarantee the size of the graph reduces in tandem with its independence number, from which we derive so-called increasing transformations.
We introduced three different types of structions that aim to reduce the number of newly constructed vertices.
We then derived special cases of these struction variants that can be
efficiently applied in practice. 
Our experimental evaluation indicates that our techniques outperform existing
algorithms on a wide variety of instances. In particular, with the exception of a single instance, our transformations produce the smallest-known reduced graphs and, when performed before branch-and-reduce, solve more instances than existing exact algorithms---at times even solving instances faster than heuristic approaches. Of particular interest for future work is engineering increasing transformations that are efficient enough to be used throughout recursion---a more general \emph{branch-and-transform} technique. %
Further work includes evaluating the conic reduction~\cite{lozin2000conic} or clique reduction~\cite{lovasz1986matching}, which are similar to struction. 

\renewcommand\bibsection{\section*{\refname}}
\bibliographystyle{abbrvnat}
\bibliography{mwis}

\begin{thebibliography}{50}
\providecommand{\natexlab}[1]{#1}
\providecommand{\url}[1]{\texttt{#1}}
\expandafter\ifx\csname urlstyle\endcsname\relax
  \providecommand{\doi}[1]{doi: #1}\else
  \providecommand{\doi}{doi: \begingroup \urlstyle{rm}\Url}\fi

\bibitem[OSM()]{OSMWEB}
\emph{OpenStreetMap}.
\newblock URL \url{https://www.openstreetmap.org}.

\bibitem[Akiba and Iwata(2016)]{akiba-tcs-2016}
T.~Akiba and Y.~Iwata.
\newblock Branch-and-reduce exponential/{FPT} algorithms in practice: A case
  study of vertex cover.
\newblock \emph{Theoretical Computer Science}, 609, Part 1:\penalty0 211--225,
  2016.
\newblock \doi{10.1016/j.tcs.2015.09.023}.

\bibitem[Alexe et~al.(2003)Alexe, Hammer, Lozin, and
  de~Werra]{alexe2003struction}
G.~Alexe, P.~L. Hammer, V.~V. Lozin, and D.~de~Werra.
\newblock Struction revisited.
\newblock \emph{Discrete applied mathematics}, 132\penalty0 (1-3):\penalty0
  27--46, 2003.
\newblock \doi{10.1016/S0166-218X(03)00388-3}.

\bibitem[Andrade et~al.(2012)Andrade, Resende, and Werneck]{andrade-2012}
D.~V. Andrade, M.~G. Resende, and R.~F. Werneck.
\newblock Fast local search for the maximum independent set problem.
\newblock \emph{Journal of Heuristics}, 18\penalty0 (4):\penalty0 525--547,
  2012.
\newblock \doi{10.1007/s10732-012-9196-4}.

\bibitem[Ay et~al.(2011)Ay, Kellis, and Kahveci]{ay2011submap}
F.~Ay, M.~Kellis, and T.~Kahveci.
\newblock Submap: aligning metabolic pathways with subnetwork mappings.
\newblock \emph{Journal of computational biology}, 18\penalty0 (3):\penalty0
  219--235, 2011.
\newblock \doi{10.1089/cmb.2010.0280}.

\bibitem[Babel(1994)]{babel1994fast}
L.~Babel.
\newblock A fast algorithm for the maximum weight clique problem.
\newblock \emph{Computing}, 52\penalty0 (1):\penalty0 31--38, 1994.
\newblock \doi{10.1007/BF02243394}.

\bibitem[Balas and Yu(1986)]{balas1986finding}
E.~Balas and C.~S. Yu.
\newblock Finding a maximum clique in an arbitrary graph.
\newblock \emph{SIAM Journal on Computing}, 15\penalty0 (4):\penalty0
  1054--1068, 1986.
\newblock \doi{10.1137/0215075}.

\bibitem[Barth et~al.(2016)Barth, Niedermann, N\"{o}llenburg, and
  Strash]{barth-2016}
L.~Barth, B.~Niedermann, M.~N\"{o}llenburg, and D.~Strash.
\newblock Temporal map labeling: A new unified framework with experiments.
\newblock In \emph{Proceedings of the 24th ACM SIGSPATIAL International
  Conference on Advances in Geographic Information Systems}, GIS '16, pages
  23:1--23:10. ACM, 2016.
\newblock \doi{10.1145/2996913.2996957}.

\bibitem[{Brouwer} et~al.(1990){Brouwer}, {Shearer}, {Sloane}, and
  {Smith}]{brouwer1990new}
A.~E. {Brouwer}, J.~B. {Shearer}, N.~J.~A. {Sloane}, and W.~D. {Smith}.
\newblock A new table of constant weight codes.
\newblock \emph{IEEE Transactions on Information Theory}, 36\penalty0
  (6):\penalty0 1334--1380, 1990.
\newblock \doi{10.1109/18.59932}.

\bibitem[Butenko and Trukhanov(2007)]{butenko-trukhanov}
S.~Butenko and S.~Trukhanov.
\newblock Using critical sets to solve the maximum independent set problem.
\newblock \emph{Operations Research Letters}, 35\penalty0 (4):\penalty0
  519--524, 2007.
\newblock \doi{10.1016/j.orl.2006.07.004}.

\bibitem[Cai et~al.(2011)Cai, Su, and Sattar]{cai2011local}
S.~Cai, K.~Su, and A.~Sattar.
\newblock Local search with edge weighting and configuration checking
  heuristics for minimum vertex cover.
\newblock \emph{Artificial Intelligence}, 175\penalty0 (9-10):\penalty0
  1672--1696, 2011.
\newblock \doi{10.1016/j.artint.2011.03.003}.

\bibitem[Cai et~al.(2018)Cai, Hou, Lin, and Li]{cai-dynwvc}
S.~Cai, W.~Hou, J.~Lin, and Y.~Li.
\newblock Improving local search for minimum weight vertex cover by dynamic
  strategies.
\newblock In \emph{Proceedings of the Twenty-Seventh International Joint
  Conference on Artificial Intelligence ({IJCAI} 2018)}, pages 1412--1418,
  2018.
\newblock \doi{10.24963/ijcai.2018/196}.

\bibitem[Chang et~al.(2017)Chang, Li, and Zhang]{chang2017}
L.~Chang, W.~Li, and W.~Zhang.
\newblock Computing a near-maximum independent set in linear time by
  reducing-peeling.
\newblock \emph{Proceedings of the 2017 ACM International Conference on
  Management of Data (SIGMOD '17)}, pages 1181--1196, 2017.
\newblock \doi{10.1145/3035918.3035939}.

\bibitem[Chou et~al.(2011)Chou, Kim, and Rotem]{chou2011energy}
J.~Chou, J.~Kim, and D.~Rotem.
\newblock Energy-aware scheduling in disk storage systems.
\newblock In \emph{2011 31st International Conference on Distributed Computing
  Systems}, pages 423--433. IEEE, 2011.
\newblock \doi{10.1109/ICDCS.2011.40}.

\bibitem[Dahlum et~al.(2016)Dahlum, Lamm, Sanders, Schulz, Strash, and
  Werneck]{dahlum2016}
J.~Dahlum, S.~Lamm, P.~Sanders, C.~Schulz, D.~Strash, and R.~F. Werneck.
\newblock Accelerating local search for the maximum independent set problem.
\newblock In A.~V. Goldberg and A.~S. Kulikov, editors, \emph{Experimental
  Algorithms (SEA 2016)}, volume 9685 of \emph{LNCS}, pages 118--133. Springer,
  2016.
\newblock \doi{10.1007/978-3-319-38851-9_9}.

\bibitem[Ebenegger et~al.(1984)Ebenegger, Hammer, and
  De~Werra]{ebenegger1984pseudo}
C.~Ebenegger, P.~Hammer, and D.~De~Werra.
\newblock Pseudo-boolean functions and stability of graphs.
\newblock In \emph{North-Holland mathematics studies}, volume~95, pages 83--97.
  Elsevier, 1984.
\newblock \doi{10.1016/S0304-0208(08)72955-4}.

\bibitem[Garey et~al.(1974)Garey, Johnson, and Stockmeyer]{garey1974}
M.~R. Garey, D.~S. Johnson, and L.~Stockmeyer.
\newblock {Some Simplified {N}{P}-Complete Problems}.
\newblock In \emph{Proceedings of the 6th ACM Symposium on Theory of
  Computing}, STOC '74, pages 47--63. ACM, 1974.
\newblock \doi{10.1145/800119.803884}.

\bibitem[Gemsa et~al.(2014)Gemsa, N\"ollenburg, and
  Rutter]{gemsa2014dynamiclabel}
A.~Gemsa, M.~N\"ollenburg, and I.~Rutter.
\newblock Evaluation of labeling strategies for rotating maps.
\newblock In \emph{Experimental Algorithms (SEA'14)}, volume 8504 of
  \emph{LNCS}, pages 235--246. Springer, 2014.
\newblock \doi{10.1007/978-3-319-07959-2_20}.

\bibitem[Hammer et~al.(1985{\natexlab{a}})Hammer, Mahadev, and
  de~Werra]{hammer1985stability}
P.~L. Hammer, N.~V. Mahadev, and D.~de~Werra.
\newblock Stability in can-free graphs.
\newblock \emph{Journal of Combinatorial Theory, Series B}, 38\penalty0
  (1):\penalty0 23--30, 1985{\natexlab{a}}.
\newblock \doi{10.1016/0095-8956(85)90089-9}.

\bibitem[Hammer et~al.(1985{\natexlab{b}})Hammer, Mahadev, and
  de~Werra]{hammer1985struction}
P.~L. Hammer, N.~V.~R. Mahadev, and D.~de~Werra.
\newblock The struction of a graph: Application to cn-free graphs.
\newblock \emph{Combinatorica}, 5\penalty0 (2):\penalty0 141--147,
  1985{\natexlab{b}}.
\newblock \doi{10.1007/BF02579377}.

\bibitem[Hespe et~al.(2018)Hespe, Schulz, and Strash]{hespe2018scalable}
D.~Hespe, C.~Schulz, and D.~Strash.
\newblock Scalable kernelization for maximum independent sets.
\newblock In \emph{2018 Proceedings of the Twentieth Workshop on Algorithm
  Engineering and Experiments (ALENEX)}, pages 223--237. SIAM, 2018.
\newblock \doi{10.1137/1.9781611975055.19}.

\bibitem[Hespe et~al.(2020)Hespe, Lamm, Schulz, and
  Strash]{hespe2020wegotyoucovered}
D.~Hespe, S.~Lamm, C.~Schulz, and D.~Strash.
\newblock {WeGotYouCovered}: The winning solver from the {PACE} 2019 challenge,
  vertex cover track.
\newblock In \emph{2020 Proceedings of the SIAM Workshop on Combinatorial
  Scientific Computing}, pages 1--11. SIAM, 2020.
\newblock \doi{10.1137/1.9781611976229.1}.

\bibitem[Hoke and Troyon(1994)]{hoke1994struction}
K.~W. Hoke and M.~Troyon.
\newblock The struction algorithm for the maximum stable set problem revisited.
\newblock \emph{Discrete Mathematics}, 131\penalty0 (1-3):\penalty0 105--113,
  1994.
\newblock \doi{10.1016/0012-365X(94)90377-8}.

\bibitem[Lamm et~al.(2017)Lamm, Sanders, Schulz, Strash, and
  Werneck]{redumis-2017}
S.~Lamm, P.~Sanders, C.~Schulz, D.~Strash, and R.~F. Werneck.
\newblock Finding near-optimal independent sets at scale.
\newblock \emph{Journal of Heuristics}, 23\penalty0 (4):\penalty0 207--229,
  2017.
\newblock \doi{10.1007/s10732-017-9337-x}.

\bibitem[Lamm et~al.(2019)Lamm, Schulz, Strash, Williger, and Zhang]{lamm-2019}
S.~Lamm, C.~Schulz, D.~Strash, R.~Williger, and H.~Zhang.
\newblock Exactly solving the maximum weight independent set problem on large
  real-world graphs.
\newblock In \emph{2019 Proceedings of the Twenty-First Workshop on Algorithm
  Engineering and Experiments (ALENEX)}, pages 144--158. SIAM, 2019.
\newblock \doi{10.1137/1.9781611975499.12}.

\bibitem[Leskovec and Krevl(2014)]{snapnets}
J.~Leskovec and A.~Krevl.
\newblock {SNAP Datasets}: {Stanford} large network dataset collection.
\newblock URL \url{http://snap.stanford.edu/data}, June 2014.

\bibitem[Li et~al.(2017{\natexlab{a}})Li, Jiang, and
  Many{\`a}]{li2017minimization}
C.-M. Li, H.~Jiang, and F.~Many{\`a}.
\newblock On minimization of the number of branches in branch-and-bound
  algorithms for the maximum clique problem.
\newblock \emph{Computers \& Operations Research}, 84:\penalty0 1--15,
  2017{\natexlab{a}}.
\newblock \doi{10.1016/j.cor.2017.02.017}.

\bibitem[Li et~al.(2019)Li, Hu, Cai, Gao, Wang, and Yin]{li2019numwvc}
R.~Li, S.~Hu, S.~Cai, J.~Gao, Y.~Wang, and M.~Yin.
\newblock Numwvc: A novel local search for minimum weighted vertex cover
  problem.
\newblock \emph{Journal of the Operational Research Society}, pages 1--12,
  2019.
\newblock \doi{10.1080/01605682.2019.1621218}.

\bibitem[Li et~al.(2017{\natexlab{b}})Li, Cai, and Hou]{li2017efficient}
Y.~Li, S.~Cai, and W.~Hou.
\newblock An efficient local search algorithm for minimum weighted vertex cover
  on massive graphs.
\newblock In \emph{Asia-Pacific Conference on Simulated Evolution and Learning
  (SEAL 2017)}, volume 10593 of \emph{LNCS}, pages 145--157.
  2017{\natexlab{b}}.
\newblock \doi{10.1007/978-3-319-68759-9_13}.

\bibitem[Lov{\'a}sz and Plummer(1986)]{lovasz1986matching}
L.~Lov{\'a}sz and M.~D. Plummer.
\newblock \emph{Matching theory}, volume 121 of \emph{North-Holland Mathematics
  Studies}, pages 471--482.
\newblock North-Holland, 1986.
\newblock \doi{10.1016/S0304-0208(08)73648-X}.

\bibitem[Lozin(2000)]{lozin2000conic}
V.~V. Lozin.
\newblock Conic reduction of graphs for the stable set problem.
\newblock \emph{Discrete Mathematics}, 222\penalty0 (1-3):\penalty0 199--211,
  2000.
\newblock \doi{10.1016/S0012-365X(99)00408-2}.

\bibitem[Ma and Latecki(2012)]{ma2012maximum}
T.~Ma and L.~J. Latecki.
\newblock Maximum weight cliques with mutex constraints for video object
  segmentation.
\newblock In \emph{2012 IEEE Conference on Computer Vision and Pattern
  Recognition}, pages 670--677. IEEE, 2012.
\newblock \doi{10.1109/CVPR.2012.6247735}.

\bibitem[Mascia et~al.(2010)Mascia, Cilia, Brunato, and
  Passerini]{mascia2010predicting}
F.~Mascia, E.~Cilia, M.~Brunato, and A.~Passerini.
\newblock Predicting structural and functional sites in proteins by searching
  for maximum-weight cliques.
\newblock In \emph{Twenty-Fourth AAAI Conference on Artificial Intelligence},
  2010.

\bibitem[Nogueira et~al.(2018)Nogueira, Pinheiro, and
  Subramanian]{hybrid-ils-2018}
B.~Nogueira, R.~G.~S. Pinheiro, and A.~Subramanian.
\newblock A hybrid iterated local search heuristic for the maximum weight
  independent set problem.
\newblock \emph{Optimization Letters}, 12\penalty0 (3):\penalty0 567--583,
  2018.
\newblock \doi{10.1007/s11590-017-1128-7}.

\bibitem[Nurmela et~al.(1997)Nurmela, Kaikkonen, and Ostergard]{nurmela1997new}
K.~J. Nurmela, M.~K. Kaikkonen, and P.~Ostergard.
\newblock New constant weight codes from linear permutation groups.
\newblock \emph{IEEE Transactions on Information Theory}, 43\penalty0
  (5):\penalty0 1623--1630, 1997.
\newblock \doi{10.1109/18.623163}.

\bibitem[{\"O}sterg{\aa}rd(2002)]{ostergaard2002fast}
P.~R. {\"O}sterg{\aa}rd.
\newblock A fast algorithm for the maximum clique problem.
\newblock \emph{Discrete Applied Mathematics}, 120\penalty0 (1-3):\penalty0
  197--207, 2002.
\newblock \doi{10.1016/S0166-218X(01)00290-6}.

\bibitem[Pedersen et~al.(2006)Pedersen, Vestergaard,
  et~al.]{pedersen2006bounds}
A.~S.~P. Pedersen, P.~D. Vestergaard, et~al.
\newblock Bounds on the number of vertex independent sets in a graph.
\newblock \emph{Taiwanese Journal of Mathematics}, 10\penalty0 (6):\penalty0
  1575--1587, 2006.
\newblock \doi{10.11650/twjm/1500404576}.

\bibitem[Prodinger and Tichy(1982)]{prodinger1982fibonacci}
H.~Prodinger and R.~Tichy.
\newblock Fibonacci numbers of graphs.
\newblock \emph{The Fibonacci Quarterly}, 20\penalty0 (1):\penalty0 16--21,
  1982.

\bibitem[Prosser and Trimble(2019)]{peaty-pace}
P.~Prosser and J.~Trimble.
\newblock \emph{Peaty: An exact solver for the vertex cover problem}, 2019.

\bibitem[Rebennack et~al.(2011)Rebennack, Oswald, Theis, Seitz, Reinelt, and
  Pardalos]{rebennack2011branch}
S.~Rebennack, M.~Oswald, D.~O. Theis, H.~Seitz, G.~Reinelt, and P.~M. Pardalos.
\newblock A branch and cut solver for the maximum stable set problem.
\newblock \emph{Journal of combinatorial optimization}, 21\penalty0
  (4):\penalty0 434--457, 2011.
\newblock \doi{10.1007/s10878-009-9264-3}.

\bibitem[Sander et~al.(2008)Sander, Nehab, Chlamtac, and
  Hoppe]{sander2008efficient}
P.~V. Sander, D.~Nehab, E.~Chlamtac, and H.~Hoppe.
\newblock Efficient traversal of mesh edges using adjacency primitives.
\newblock \emph{ACM Transactions on Graphics (TOG)}, 27\penalty0 (5):\penalty0
  1--9, 2008.
\newblock \doi{10.1145/1409060.1409097}.

\bibitem[Strash(2016)]{strash-power-2016}
D.~Strash.
\newblock On the power of simple reductions for the maximum independent set
  problem.
\newblock In T.~N. Dinh and M.~T. Thai, editors, \emph{Computing and
  Combinatorics (COCOON'16)}, volume 9797 of \emph{LNCS}, pages 345--356. 2016.
\newblock \doi{10.1007/978-3-319-42634-1_28}.

\bibitem[Szab\'{o} and Zavalnij(2019)]{bogdan-pace}
S.~Szab\'{o} and B.~Zavalnij.
\newblock Combining algorithms for vertex cover and clique search.
\newblock In \emph{Proceedings of the 22nd International Multiconference
  INFORMATION SOCIETY -- IS 2019, Volume I: Middle-European Conference on
  Applied Theoretical Computer Science}, pages 71--74, 2019.

\bibitem[Trimble(2017)]{glasgowinstances}
J.~Trimble.
\newblock Maximum weight clique instances.
\newblock doi:
  \href{https://doi.org/10.5281/zenodo.848647}{10.5281/zenodo.848647}, Aug.
  2017.

\bibitem[Warren and Hicks(2006)]{warren2006combinatorial}
J.~S. Warren and I.~V. Hicks.
\newblock Combinatorial branch-and-bound for the maximum weight independent set
  problem.
\newblock 2006.
\newblock URL \url{https://www.caam.rice.edu/~ivhicks/jeff.rev.pdf}.

\bibitem[Warrier(2007)]{warrier2007branch}
D.~Warrier.
\newblock \emph{A branch, price, and cut approach to solving the maximum
  weighted independent set problem}.
\newblock PhD thesis, Texas A\&M University, 2007.

\bibitem[Warrier et~al.(2005)Warrier, Wilhelm, Warren, and
  Hicks]{warrier2005branch}
D.~Warrier, W.~E. Wilhelm, J.~S. Warren, and I.~V. Hicks.
\newblock A branch-and-price approach for the maximum weight independent set
  problem.
\newblock \emph{Networks: An International Journal}, 46\penalty0 (4):\penalty0
  198--209, 2005.
\newblock \doi{10.1002/net.20088}.

\bibitem[Wu and Hao(2015)]{wu2015solving}
Q.~Wu and J.-K. Hao.
\newblock Solving the winner determination problem via a weighted maximum
  clique heuristic.
\newblock \emph{Expert Systems with Applications}, 42\penalty0 (1):\penalty0
  355--365, 2015.
\newblock \doi{10.1016/j.eswa.2014.07.027}.

\bibitem[Xu et~al.(2016)Xu, Kumar, and Koenig]{xu2016new}
H.~Xu, T.~S. Kumar, and S.~Koenig.
\newblock A new solver for the minimum weighted vertex cover problem.
\newblock In \emph{International Conference on AI and OR Techniques in
  Constriant Programming for Combinatorial Optimization Problems}, pages
  392--405. Springer, 2016.
\newblock \doi{10.1007/978-3-319-33954-2_28}.

\bibitem[Xu et~al.(2010)Xu, Tang, and Wan]{xu2010maximum}
X.~Xu, S.~Tang, and P.-J. Wan.
\newblock Maximum weighted independent set of links under physical interference
  model.
\newblock In \emph{International Conference on Wireless Algorithms, Systems,
  and Applications}, pages 68--74. Springer, 2010.
\newblock \doi{10.1007/978-3-642-14654-1_8}.

\end{thebibliography}

\begin{appendix}
{} 



{} 

\section{Graph Properties}\label{properties}

\begin{table}[h!]
\scriptsize
\centering
\begin{ThreePartTable}
\begin{tabu}{ccc}

  \mc{1}{l|}{Graph} & \mc{1}{r|}{$|V|$} & \mc{1}{r}{|E|} \\ 
\midrule
  \mc{1}{l|}{\texttt{fe\_4elt2}}  & \mc{1}{r|}{\numprint{11143}}  & \mc{1}{r}{\numprint{65636}} \\

  \mc{1}{l|}{\texttt{fe\_body}}   & \mc{1}{r|}{\numprint{45087}}  & \mc{1}{r}{\numprint{327468}} \\

  \mc{1}{l|}{\texttt{fe\_ocean}}  & \mc{1}{r|}{\numprint{143437}} & \mc{1}{r}{\numprint{819186}} \\

  \mc{1}{l|}{\texttt{fe\_pwt}}    & \mc{1}{r|}{\numprint{36519}}  & \mc{1}{r}{\numprint{289588}} \\

  \mc{1}{l|}{\texttt{fe\_rotor}}  & \mc{1}{r|}{\numprint{99617}}  & \mc{1}{r}{\numprint{1324862}} \\

  \mc{1}{l|}{\texttt{fe\_sphere}} & \mc{1}{r|}{\numprint{16386}}  & \mc{1}{r}{\numprint{98304}} \\

  \mc{1}{l|}{\texttt{fe\_tooth}}  & \mc{1}{r|}{\numprint{78136}}  & \mc{1}{r}{\numprint{905182}} \\

\end{tabu}
\end{ThreePartTable}

\caption{Properties of FE instances}
\label{basic_property_table_fe}
\end{table}

\begin{table}[h!]
\scriptsize
\centering
\begin{ThreePartTable}
\begin{tabu}{ccc}

  \mc{1}{l|}{Graph} & \mc{1}{r|}{$|V|$} & \mc{1}{r}{|E|} \\ 
\midrule

  \mc{1}{l|}{\texttt{beethoven}} & \mc{1}{r|}{\numprint{4419}}    & \mc{1}{r}{\numprint{12982}} \\

  \mc{1}{l|}{\texttt{blob}}      & \mc{1}{r|}{\numprint{16068}}   & \mc{1}{r}{\numprint{48204}} \\

  \mc{1}{l|}{\texttt{buddha}}    & \mc{1}{r|}{\numprint{1087716}} & \mc{1}{r}{\numprint{3263148}} \\

  \mc{1}{l|}{\texttt{bunny}}     & \mc{1}{r|}{\numprint{68790}}   & \mc{1}{r}{\numprint{206034}} \\

  \mc{1}{l|}{\texttt{cow}}       & \mc{1}{r|}{\numprint{5036}}    & \mc{1}{r}{\numprint{14732}} \\

  \mc{1}{l|}{\texttt{dragon}}    & \mc{1}{r|}{\numprint{150000}}  & \mc{1}{r}{\numprint{450000}} \\

  \mc{1}{l|}{\texttt{dragonsub}} & \mc{1}{r|}{\numprint{600000}}  & \mc{1}{r}{\numprint{1800000}} \\

  \mc{1}{l|}{\texttt{ecat}}      & \mc{1}{r|}{\numprint{684496}}  & \mc{1}{r}{\numprint{2053488}} \\

  \mc{1}{l|}{\texttt{face}}      & \mc{1}{r|}{\numprint{22871}}   & \mc{1}{r}{\numprint{68108}} \\

  \mc{1}{l|}{\texttt{fandisk}}   & \mc{1}{r|}{\numprint{8634}}    & \mc{1}{r}{\numprint{25636}} \\

  \mc{1}{l|}{\texttt{feline}}    & \mc{1}{r|}{\numprint{41262}}   & \mc{1}{r}{\numprint{123786}} \\

  \mc{1}{l|}{\texttt{gameguy}}   & \mc{1}{r|}{\numprint{42623}}   & \mc{1}{r}{\numprint{127700}} \\

  \mc{1}{l|}{\texttt{gargoyle}}  & \mc{1}{r|}{\numprint{20000}}   & \mc{1}{r}{\numprint{60000}} \\

  \mc{1}{l|}{\texttt{turtle}}    & \mc{1}{r|}{\numprint{267534}}  & \mc{1}{r}{\numprint{802356}} \\

  \mc{1}{l|}{\texttt{venus}}     & \mc{1}{r|}{\numprint{5672}}    & \mc{1}{r}{\numprint{17016}} \\

  \end{tabu}
\end{ThreePartTable}

\caption{Properties of mesh instances}
\label{basic_property_table_mesh}
\end{table}

\begin{table}[h!]
\centering
\scriptsize
\begin{ThreePartTable}
\begin{tabu}{ccc}

  \mc{1}{l|}{Graph} & \mc{1}{r|}{$|V|$} & \mc{1}{r}{|E|} \\ 
\midrule
  \mc{1}{l|}{\texttt{as-skitter}}              & \mc{1}{r|}{\numprint{1696415}} & \mc{1}{r}{\numprint{22190596}} \\

  \mc{1}{l|}{\texttt{ca-AstroPh}}              & \mc{1}{r|}{\numprint{18772}}   & \mc{1}{r}{\numprint{396100}} \\

  \mc{1}{l|}{\texttt{ca-CondMat}}              & \mc{1}{r|}{\numprint{23133}}   & \mc{1}{r}{\numprint{186878}} \\

  \mc{1}{l|}{\texttt{ca-GrQc}}                 & \mc{1}{r|}{\numprint{5242}}    & \mc{1}{r}{\numprint{28968}} \\

  \mc{1}{l|}{\texttt{ca-HepPh}}                & \mc{1}{r|}{\numprint{12008}}   & \mc{1}{r}{\numprint{236978}} \\

  \mc{1}{l|}{\texttt{ca-HepTh}}                & \mc{1}{r|}{\numprint{9877}}    & \mc{1}{r}{\numprint{51946}} \\

  \mc{1}{l|}{\texttt{email-Enron}}             & \mc{1}{r|}{\numprint{36692}}   & \mc{1}{r}{\numprint{367662}} \\

  \mc{1}{l|}{\texttt{email-EuAll}}             & \mc{1}{r|}{\numprint{265214}}  & \mc{1}{r}{\numprint{728962}} \\

  \mc{1}{l|}{\texttt{p2p-Gnutella04}}          & \mc{1}{r|}{\numprint{10876}}   & \mc{1}{r}{\numprint{79988}} \\

  \mc{1}{l|}{\texttt{p2p-Gnutella05}}          & \mc{1}{r|}{\numprint{8846}}    & \mc{1}{r}{\numprint{63678}} \\

  \mc{1}{l|}{\texttt{p2p-Gnutella06}}          & \mc{1}{r|}{\numprint{8717}}    & \mc{1}{r}{\numprint{63050}} \\

  \mc{1}{l|}{\texttt{p2p-Gnutella08}}          & \mc{1}{r|}{\numprint{6301}}    & \mc{1}{r}{\numprint{41554}} \\

  \mc{1}{l|}{\texttt{p2p-Gnutella09}}          & \mc{1}{r|}{\numprint{8114}}    & \mc{1}{r}{\numprint{52026}} \\

  \mc{1}{l|}{\texttt{p2p-Gnutella24}}          & \mc{1}{r|}{\numprint{26518}}   & \mc{1}{r}{\numprint{130738}} \\

  \mc{1}{l|}{\texttt{p2p-Gnutella25}}          & \mc{1}{r|}{\numprint{22687}}   & \mc{1}{r}{\numprint{109410}} \\

  \mc{1}{l|}{\texttt{p2p-Gnutella30}}          & \mc{1}{r|}{\numprint{36682}}   & \mc{1}{r}{\numprint{176656}} \\

  \mc{1}{l|}{\texttt{p2p-Gnutella31}}          & \mc{1}{r|}{\numprint{62586}}   & \mc{1}{r}{\numprint{295784}} \\

  \mc{1}{l|}{\texttt{roadNet-CA}}              & \mc{1}{r|}{\numprint{1965206}} & \mc{1}{r}{\numprint{5533214}} \\

  \mc{1}{l|}{\texttt{roadNet-PA}}              & \mc{1}{r|}{\numprint{1088092}} & \mc{1}{r}{\numprint{3083796}} \\

  \mc{1}{l|}{\texttt{roadNet-TX}}              & \mc{1}{r|}{\numprint{1379917}} & \mc{1}{r}{\numprint{3843320}} \\

  \mc{1}{l|}{\texttt{soc-Epinions1}}           & \mc{1}{r|}{\numprint{75879}}   & \mc{1}{r}{\numprint{811480}} \\

  \mc{1}{l|}{\texttt{soc-LiveJournal1}}        & \mc{1}{r|}{\numprint{4847571}} & \mc{1}{r}{\numprint{85702474}} \\

  \mc{1}{l|}{\texttt{soc-Slashdot0811}}        & \mc{1}{r|}{\numprint{77360}}   & \mc{1}{r}{\numprint{938360}} \\

  \mc{1}{l|}{\texttt{soc-Slashdot0902}}        & \mc{1}{r|}{\numprint{82168}}   & \mc{1}{r}{\numprint{1008460}} \\

  \mc{1}{l|}{\texttt{soc-pokec-relationships}} & \mc{1}{r|}{\numprint{1632803}} & \mc{1}{r}{\numprint{44603928}} \\

  \mc{1}{l|}{\texttt{web-BerkStan}}            & \mc{1}{r|}{\numprint{685230}}  & \mc{1}{r}{\numprint{13298940}} \\

  \mc{1}{l|}{\texttt{web-Google}}              & \mc{1}{r|}{\numprint{875713}}  & \mc{1}{r}{\numprint{8644102}} \\

  \mc{1}{l|}{\texttt{web-NotreDame}}           & \mc{1}{r|}{\numprint{325729}}  & \mc{1}{r}{\numprint{2180216}} \\

  \mc{1}{l|}{\texttt{web-Stanford}}            & \mc{1}{r|}{\numprint{281903}}  & \mc{1}{r}{\numprint{3985272}} \\

  \mc{1}{l|}{\texttt{wiki-Talk}}               & \mc{1}{r|}{\numprint{2394385}} & \mc{1}{r}{\numprint{9319130}} \\

  \mc{1}{l|}{\texttt{wiki-Vote}}               & \mc{1}{r|}{\numprint{7115}}    & \mc{1}{r}{\numprint{201524}} \\

\end{tabu}
\end{ThreePartTable}

\caption{Properties of SNAP instances}
\label{basic_property_table_snap}
\end{table}

\begin{table}[h!]
\centering
\scriptsize
\begin{ThreePartTable}
\begin{tabu}{ccc}

  \mc{1}{l|}{Graph} & \mc{1}{r|}{$|V|$} & \mc{1}{r}{|E|} \\ 
  \midrule

\mc{1}{l|}{\texttt{alabama-AM2}}              & \mc{1}{r|}{\numprint{1164}}  & \mc{1}{r}{\numprint{38772}} \\

\mc{1}{l|}{\texttt{alabama-AM3}}              & \mc{1}{r|}{\numprint{3504}}  & \mc{1}{r}{\numprint{619328}} \\

\mc{1}{l|}{\texttt{district-of-columbia-AM1}} & \mc{1}{r|}{\numprint{2500}}  & \mc{1}{r}{\numprint{49302}} \\

\mc{1}{l|}{\texttt{district-of-columbia-AM2}} & \mc{1}{r|}{\numprint{13597}} & \mc{1}{r}{\numprint{3219590}} \\

\mc{1}{l|}{\texttt{district-of-columbia-AM3}} & \mc{1}{r|}{\numprint{46221}} & \mc{1}{r}{\numprint{55458274}} \\

\mc{1}{l|}{\texttt{florida-AM2}}              & \mc{1}{r|}{\numprint{1254}}  & \mc{1}{r}{\numprint{33872}} \\

\mc{1}{l|}{\texttt{florida-AM3}}              & \mc{1}{r|}{\numprint{2985}}  & \mc{1}{r}{\numprint{308086}} \\

\mc{1}{l|}{\texttt{georgia-AM3}}              & \mc{1}{r|}{\numprint{1680}}  & \mc{1}{r}{\numprint{148252}} \\

\mc{1}{l|}{\texttt{greenland-AM3}}            & \mc{1}{r|}{\numprint{4986}}  & \mc{1}{r}{\numprint{7304722}} \\

\mc{1}{l|}{\texttt{hawaii-AM2}}               & \mc{1}{r|}{\numprint{2875}}  & \mc{1}{r}{\numprint{530316}} \\

\mc{1}{l|}{\texttt{hawaii-AM3}}               & \mc{1}{r|}{\numprint{28006}} & \mc{1}{r}{\numprint{98889842}} \\

\mc{1}{l|}{\texttt{idaho-AM3}}                & \mc{1}{r|}{\numprint{4064}}  & \mc{1}{r}{\numprint{7848160}} \\

\mc{1}{l|}{\texttt{kansas-AM3}}               & \mc{1}{r|}{\numprint{2732}}  & \mc{1}{r}{\numprint{1613824}} \\

\mc{1}{l|}{\texttt{kentucky-AM2}}             & \mc{1}{r|}{\numprint{2453}}  & \mc{1}{r}{\numprint{1286856}} \\

\mc{1}{l|}{\texttt{kentucky-AM3}}             & \mc{1}{r|}{\numprint{19095}} & \mc{1}{r}{\numprint{119067260}} \\

\mc{1}{l|}{\texttt{louisiana-AM3}}            & \mc{1}{r|}{\numprint{1162}}  & \mc{1}{r}{\numprint{74154}} \\

\mc{1}{l|}{\texttt{maryland-AM3}}             & \mc{1}{r|}{\numprint{1018}}  & \mc{1}{r}{\numprint{190830}} \\

\mc{1}{l|}{\texttt{massachusetts-AM2}}        & \mc{1}{r|}{\numprint{1339}}  & \mc{1}{r}{\numprint{70898}} \\

\mc{1}{l|}{\texttt{massachusetts-AM3}}        & \mc{1}{r|}{\numprint{3703}}  & \mc{1}{r}{\numprint{1102982}} \\

\mc{1}{l|}{\texttt{mexico-AM3}}               & \mc{1}{r|}{\numprint{1096}}  & \mc{1}{r}{\numprint{94262}} \\

\mc{1}{l|}{\texttt{new-hampshire-AM3}}        & \mc{1}{r|}{\numprint{1107}}  & \mc{1}{r}{\numprint{36042}} \\

\mc{1}{l|}{\texttt{north-carolina-AM3}}       & \mc{1}{r|}{\numprint{1557}}  & \mc{1}{r}{\numprint{473478}} \\

\mc{1}{l|}{\texttt{oregon-AM2}}               & \mc{1}{r|}{\numprint{1325}}  & \mc{1}{r}{\numprint{115034}} \\

\mc{1}{l|}{\texttt{oregon-AM3}}               & \mc{1}{r|}{\numprint{5588}}  & \mc{1}{r}{\numprint{5825402}} \\

\mc{1}{l|}{\texttt{pennsylvania-AM3}}         & \mc{1}{r|}{\numprint{1148}}  & \mc{1}{r}{\numprint{52928}} \\

\mc{1}{l|}{\texttt{rhode-island-AM2}}         & \mc{1}{r|}{\numprint{2866}}  & \mc{1}{r}{\numprint{590976}} \\

\mc{1}{l|}{\texttt{rhode-island-AM3}}         & \mc{1}{r|}{\numprint{15124}} & \mc{1}{r}{\numprint{25244438}} \\

\mc{1}{l|}{\texttt{utah-AM3}}                 & \mc{1}{r|}{\numprint{1339}}  & \mc{1}{r}{\numprint{85744}} \\

\mc{1}{l|}{\texttt{vermont-AM3}}              & \mc{1}{r|}{\numprint{3436}}  & \mc{1}{r}{\numprint{2272328}} \\

\mc{1}{l|}{\texttt{virginia-AM2}}             & \mc{1}{r|}{\numprint{2279}}  & \mc{1}{r}{\numprint{120080}} \\

\mc{1}{l|}{\texttt{virginia-AM3}}             & \mc{1}{r|}{\numprint{6185}}  & \mc{1}{r}{\numprint{1331806}} \\

\mc{1}{l|}{\texttt{washington-AM2}}           & \mc{1}{r|}{\numprint{3025}}  & \mc{1}{r}{\numprint{304898}} \\

\mc{1}{l|}{\texttt{washington-AM3}}           & \mc{1}{r|}{\numprint{10022}} & \mc{1}{r}{\numprint{4692426}} \\

\mc{1}{l|}{\texttt{west-virginia-AM3}}        & \mc{1}{r|}{\numprint{1185}}  & \mc{1}{r}{\numprint{251240}} \\

\end{tabu}
\end{ThreePartTable}

\caption{Properties of OSM instances}
\label{basic_property_table_osm}
\end{table}
\vfill
\pagebreak

\section{Proofs}
\label{appendix:proofs}

\begin{lemma} After using the modified weighted struction, the equality
$\alpha_w(G)~=~\alpha_w(G')~+~w(v)$ holds.
\end{lemma}

\begin{proof}
Any maximal independent set in $G$ must contain either $v$ or some
nodes from $N(v)$. Let $I$ be a maximum weight independent set in
$G$, and $\alpha_w(G)$ its weight. If $v\in I$, then there is an
independent set $I^{\prime}$ in $G^{\prime}$ with weight
$\alpha_w(G)-w(v)$, namely $I^{\prime}=I\setminus v$. If
$v\notin I$, that is there are some nodes from $N(v)$ in $I$, then
there is an independent set $I^{\prime}$ in $G^{\prime}$ with weight
$\alpha_w(G)-w(v)$. Let denote the nodes from $I$ that are in
$N(v)$ by $\{x, y_1, y_2,\ldots,y_p\},
x<y_1<y_2<\cdots<y_p$. The independent set $I^{\prime}$ with weight
$\alpha_w(G)-w(v)$ can be constructed, namely $I^{\prime} =
I\setminus \{x, y_1, y_2,\ldots,y_p\} \cup
\{x^{\prime}, v^{\prime}_{x,y_1},
v^{\prime}_{x,y_2},\ldots,v^{\prime}_{x,y_p}\}$.
 
Let $I^{\prime}$ be a maximum weight independent set in $G^{\prime}$
and $\alpha_w(G^{\prime})$ its weight. If any new nodes $x^{\prime},
v^{\prime}_{x,y_1}, v^{\prime}_{x,y_2},\ldots,v^{\prime}_{x,y_p}\in
I^{\prime}$, then there is an independent set $I$ in $G$ with weight
$\alpha_w(G^{\prime})+w(v)$, namely $I=I^{\prime}\setminus
\{x^{\prime}, v^{\prime}_{x,y_1},
v^{\prime}_{x,y_2},\ldots,v^{\prime}_{x,y_p}\} \cup \{x, y_1,
y_2,\ldots,y_p\}$. If there is no new node in $I^{\prime}$,
then there is an independent set $I$ in $G$ with weight
$\alpha_w(G^{\prime})+w(v)$, namely $I=I^{\prime}\cup \{v\}$.
\end{proof}

\begin{lemma} After using the extended weighted struction, the equality
$\alpha_w(G)~=~\alpha_w(G')~+~w(v)$ holds.
\end{lemma}

\begin{proof}
Any maximal independent set in $G$ must contain either $v$ or some
nodes from $N(v)$. Let $I$ be a maximum weight independent set in
$G$, and $\alpha_w(G)$ its weight. If $v\in I$, then there is an
independent set $I^{\prime}$ in $G^{\prime}$ with weight
$\alpha_w(G)-w(v)$, namely $I^{\prime}=I\setminus v$. If
$v\notin I$, that is there are some nodes from $N(v)$ in $I$, then
there is an independent set $I^{\prime}$ in $G^{\prime}$ with weight
$\alpha_w(G)-w(v)$. Let denote the nodes from $I$ that are in
$N(v)$ by $\{u_1, u_2,\ldots,u_t\}$. The independent set
$I^{\prime}$ with weight $\alpha_w(G)-w(v)$ can be constructed from
$I$ by deleting nodes $\{u_1, u_2,\ldots,u_t\}$ and adding the one new
node that corresponds to the independent set containing these nodes.
 
Let $I^{\prime}$ be a maximum weight independent set in $G^{\prime}$
and $\alpha_w(G^{\prime})$ its weight. If any new node $y\in
I^{\prime}$ (let nodes $\{u_1, u_2,\ldots,u_t\}$ from $G$ be the nodes
that correspond to the independent set that is represented by $y$),
then there is an independent set $I$ in $G$ with weight
$\alpha_w(G^{\prime})+w(v)$, namely $I=I^{\prime}\setminus y \cup
\{u_1, u_2,\ldots,u_t\}$. If there is no new node in $I^{\prime}$,
then there is an independent set $I$ in $G$ with weight
$\alpha_w(G^{\prime})+w(v)$, namely $I=I^{\prime}\cup \{v\}$.
\end{proof}

\begin{lemma} After using the extended reduced weighted struction, the equality
$\alpha_w(G)~=~\alpha_w(G')~+~w(v)$ holds.
\end{lemma}

\begin{proof}
Concludes from the proof for extended weighted struction and the
modified weighted struction.
\end{proof}
\onecolumn 
\clearpage

\begin{landscape}
\section{Branch-and-Reduce Comparison}\label{bnr_tables}
\begin{table}[h!]
\centering
\scriptsize
\begin{ThreePartTable}
\begin{tabular}{cccccccccccccccc}

  \mc{1}{l|}{Graph} & \mc{1}{r}{$n$} & \mc{1}{r}{$t_{r}$} & \mc{1}{r|}{$t_{t}$} & \mc{1}{r}{$n$} & \mc{1}{r}{$t_{r}$} & \mc{1}{r|}{$t_{t}$} & \mc{1}{r}{$n$} & \mc{1}{r}{$t_{r}$} & \mc{1}{r|}{$t_{t}$} & \mc{1}{r}{$n$} & \mc{1}{r}{$t_{r}$} & \mc{1}{r|}{$t_{t}$} & \mc{1}{r}{$n$} & \mc{1}{r}{$t_{r}$} & \mc{1}{r}{$t_{t}$} \\ 
  \hline  

  \mc{1}{l|}{FE instances} & \mc{3}{c|}{\basicDense} & \mc{3}{c|}{\basicSparse} & \mc{3}{c|}{\nonIncreasing} & \mc{3}{c|}{\cyclicFast} & \mc{3}{c}{\cyclicStrong} \\  
  \hline  

  \rowcolor{lightergray} \mc{1}{l|}{\texttt{fe\_4elt2}} & \mc{1}{r}{\numprint{8580}} & \mc{1}{r}{0.29} & \mc{1}{r|}{-} & \mc{1}{r}{\numprint{8578}} & \mc{1}{r}{0.87} & \mc{1}{r|}{-} & \mc{1}{r}{562} & \mc{1}{r}{0.10} & \mc{1}{r|}{-} & \mc{1}{r}{0} & \mc{1}{r}{0.12} & \mc{1}{r|}{\textbf{0.13}} & \mc{1}{r}{0} & \mc{1}{r}{0.16} & \mc{1}{r}{0.17} \\  
  \mc{1}{l|}{\texttt{fe\_body}} & \mc{1}{r}{\numprint{16107}} & \mc{1}{r}{0.69} & \mc{1}{r|}{-} & \mc{1}{r}{\numprint{15992}} & \mc{1}{r}{3.40} & \mc{1}{r|}{-} & \mc{1}{r}{\numprint{1162}} & \mc{1}{r}{0.16} & \mc{1}{r|}{-} & \mc{1}{r}{625} & \mc{1}{r}{0.44} & \mc{1}{r|}{-} & \mc{1}{r}{553} & \mc{1}{r}{0.94} & \mc{1}{r}{-} \\  
  \mc{1}{l|}{\texttt{fe\_ocean}} & \mc{1}{r}{\numprint{141283}} & \mc{1}{r}{1.05} & \mc{1}{r|}{-} & \mc{1}{r}{0} & \mc{1}{r}{5.94} & \mc{1}{r|}{\textbf{5.99}} & \mc{1}{r}{\numprint{138338}} & \mc{1}{r}{8.90} & \mc{1}{r|}{-} & \mc{1}{r}{\numprint{138134}} & \mc{1}{r}{9.61} & \mc{1}{r|}{-} & \mc{1}{r}{\numprint{138049}} & \mc{1}{r}{10.78} & \mc{1}{r}{-} \\  
  \mc{1}{l|}{\texttt{fe\_pwt}} & \mc{1}{r}{\numprint{34521}} & \mc{1}{r}{0.46} & \mc{1}{r|}{-} & \mc{1}{r}{\numprint{34521}} & \mc{1}{r}{2.70} & \mc{1}{r|}{-} & \mc{1}{r}{\numprint{25550}} & \mc{1}{r}{0.78} & \mc{1}{r|}{-} & \mc{1}{r}{\numprint{20241}} & \mc{1}{r}{1.80} & \mc{1}{r|}{-} & \mc{1}{r}{\numprint{14107}} & \mc{1}{r}{5.65} & \mc{1}{r}{-} \\  
  \mc{1}{l|}{\texttt{fe\_rotor}} & \mc{1}{r}{\numprint{98271}} & \mc{1}{r}{9.80} & \mc{1}{r|}{-} & \mc{1}{r}{\numprint{98271}} & \mc{1}{r}{24.47} & \mc{1}{r|}{-} & \mc{1}{r}{\numprint{91946}} & \mc{1}{r}{4.80} & \mc{1}{r|}{-} & \mc{1}{r}{\numprint{91634}} & \mc{1}{r}{4.82} & \mc{1}{r|}{-} & \mc{1}{r}{\numprint{89647}} & \mc{1}{r}{11.11} & \mc{1}{r}{-} \\  
  \rowcolor{lightergray} \mc{1}{l|}{\texttt{fe\_sphere}} & \mc{1}{r}{\numprint{15269}} & \mc{1}{r}{0.21} & \mc{1}{r|}{-} & \mc{1}{r}{\numprint{15269}} & \mc{1}{r}{1.47} & \mc{1}{r|}{-} & \mc{1}{r}{\numprint{2961}} & \mc{1}{r}{0.34} & \mc{1}{r|}{-} & \mc{1}{r}{147} & \mc{1}{r}{0.62} & \mc{1}{r|}{\textbf{0.83}} & \mc{1}{r}{0} & \mc{1}{r}{0.75} & \mc{1}{r}{0.77} \\  
  \rowcolor{lightergray} \mc{1}{l|}{\texttt{fe\_tooth}} & \mc{1}{r}{\numprint{10922}} & \mc{1}{r}{1.69} & \mc{1}{r|}{-} & \mc{1}{r}{\numprint{10801}} & \mc{1}{r}{3.79} & \mc{1}{r|}{-} & \mc{1}{r}{15} & \mc{1}{r}{0.41} & \mc{1}{r|}{0.46} & \mc{1}{r}{0} & \mc{1}{r}{0.30} & \mc{1}{r|}{\textbf{0.34}} & \mc{1}{r}{0} & \mc{1}{r}{0.28} & \mc{1}{r}{0.32} \\  
  \hline \hline  

  \mc{1}{l|}{OSM instances} & \mc{3}{c|}{\basicDense} & \mc{3}{c|}{\basicSparse} & \mc{3}{c|}{\nonIncreasing} & \mc{3}{c|}{\cyclicFast} & \mc{3}{c|}{\cyclicStrong} \\  
  \hline  

  \rowcolor{lightergray} \mc{1}{l|}{\texttt{alabama-AM2}} & \mc{1}{r}{173} & \mc{1}{r}{0.06} & \mc{1}{r|}{0.31} & \mc{1}{r}{173} & \mc{1}{r}{0.07} & \mc{1}{r|}{0.55} & \mc{1}{r}{0} & \mc{1}{r}{0.01} & \mc{1}{r|}{\textbf{0.01}} & \mc{1}{r}{0} & \mc{1}{r}{0.01} & \mc{1}{r|}{0.01} & \mc{1}{r}{0} & \mc{1}{r}{0.01} & \mc{1}{r}{0.01} \\  

  \rowcolor{lightergray} \mc{1}{l|}{\texttt{alabama-AM3}} & \mc{1}{r}{\numprint{1614}} & \mc{1}{r}{12.05} & \mc{1}{r|}{-} & \mc{1}{r}{\numprint{1614}} & \mc{1}{r}{14.37} & \mc{1}{r|}{-} & \mc{1}{r}{\numprint{1288}} & \mc{1}{r}{0.34} & \mc{1}{r|}{-} & \mc{1}{r}{456} & \mc{1}{r}{1.45} & \mc{1}{r|}{\textbf{3.94}} & \mc{1}{r}{0} & \mc{1}{r}{33.11} & \mc{1}{r}{33.16} \\  

  \mc{1}{l|}{\texttt{district-of-columbia-AM1}} & \mc{1}{r}{800} & \mc{1}{r}{1.22} & \mc{1}{r|}{-} & \mc{1}{r}{800} & \mc{1}{r}{1.28} & \mc{1}{r|}{-} & \mc{1}{r}{367} & \mc{1}{r}{0.03} & \mc{1}{r|}{39.81} & \mc{1}{r}{185} & \mc{1}{r}{0.41} & \mc{1}{r|}{\textbf{0.80}} & \mc{1}{r}{0} & \mc{1}{r}{3.65} & \mc{1}{r}{3.66} \\  

  \rowcolor{lightergray} \mc{1}{l|}{\texttt{district-of-columbia-AM2}} & \mc{1}{r}{\numprint{6360}} & \mc{1}{r}{11.86} & \mc{1}{r|}{-} & \mc{1}{r}{\numprint{6360}} & \mc{1}{r}{14.39} & \mc{1}{r|}{-} & \mc{1}{r}{\numprint{5606}} & \mc{1}{r}{0.85} & \mc{1}{r|}{-} & \mc{1}{r}{\numprint{1855}} & \mc{1}{r}{2.51} & \mc{1}{r|}{-} & \mc{1}{r}{\numprint{1484}} & \mc{1}{r}{84.91} & \mc{1}{r}{-} \\  

  \mc{1}{l|}{\texttt{district-of-columbia-AM3}} & \mc{1}{r}{\numprint{33367}} & \mc{1}{r}{63.23} & \mc{1}{r|}{-} & \mc{1}{r}{\numprint{33367}} & \mc{1}{r}{358.14} & \mc{1}{r|}{-} & \mc{1}{r}{\numprint{32320}} & \mc{1}{r}{33.68} & \mc{1}{r|}{-} & \mc{1}{r}{\numprint{28842}} & \mc{1}{r}{66.67} & \mc{1}{r|}{-} & \mc{1}{r}{\numprint{25031}} & \mc{1}{r}{441.44} & \mc{1}{r}{-} \\  

  \rowcolor{lightergray} \mc{1}{l|}{\texttt{florida-AM2}} & \mc{1}{r}{41} & \mc{1}{r}{0.01} & \mc{1}{r|}{0.01} & \mc{1}{r}{41} & \mc{1}{r}{0.01} & \mc{1}{r|}{0.01} & \mc{1}{r}{0} & \mc{1}{r}{0.00} & \mc{1}{r|}{0.00} & \mc{1}{r}{0} & \mc{1}{r}{0.00} & \mc{1}{r|}{\textbf{0.00}} & \mc{1}{r}{0} & \mc{1}{r}{0.00} & \mc{1}{r}{0.00} \\  

  \mc{1}{l|}{\texttt{florida-AM3}} & \mc{1}{r}{\numprint{1069}} & \mc{1}{r}{31.52} & \mc{1}{r|}{45.81} & \mc{1}{r}{\numprint{1069}} & \mc{1}{r}{35.20} & \mc{1}{r|}{-} & \mc{1}{r}{814} & \mc{1}{r}{0.13} & \mc{1}{r|}{3.85} & \mc{1}{r}{661} & \mc{1}{r}{0.44} & \mc{1}{r|}{\textbf{2.93}} & \mc{1}{r}{267} & \mc{1}{r}{42.26} & \mc{1}{r}{45.05} \\  

  \mc{1}{l|}{\texttt{georgia-AM3}} & \mc{1}{r}{861} & \mc{1}{r}{8.99} & \mc{1}{r|}{892.17} & \mc{1}{r}{861} & \mc{1}{r}{10.14} & \mc{1}{r|}{-} & \mc{1}{r}{796} & \mc{1}{r}{0.08} & \mc{1}{r|}{25.97} & \mc{1}{r}{587} & \mc{1}{r}{0.69} & \mc{1}{r|}{\textbf{10.35}} & \mc{1}{r}{425} & \mc{1}{r}{12.84} & \mc{1}{r}{32.53} \\  

  \mc{1}{l|}{\texttt{greenland-AM3}} & \mc{1}{r}{\numprint{3942}} & \mc{1}{r}{3.81} & \mc{1}{r|}{-} & \mc{1}{r}{\numprint{3942}} & \mc{1}{r}{24.77} & \mc{1}{r|}{-} & \mc{1}{r}{\numprint{3953}} & \mc{1}{r}{3.94} & \mc{1}{r|}{-} & \mc{1}{r}{\numprint{3339}} & \mc{1}{r}{10.27} & \mc{1}{r|}{-} & \mc{1}{r}{\numprint{3339}} & \mc{1}{r}{54.44} & \mc{1}{r}{-} \\  

  \rowcolor{lightergray} \mc{1}{l|}{\texttt{hawaii-AM2}} & \mc{1}{r}{428} & \mc{1}{r}{2.08} & \mc{1}{r|}{4.27} & \mc{1}{r}{428} & \mc{1}{r}{2.15} & \mc{1}{r|}{10.22} & \mc{1}{r}{262} & \mc{1}{r}{0.07} & \mc{1}{r|}{0.18} & \mc{1}{r}{0} & \mc{1}{r}{0.09} & \mc{1}{r|}{\textbf{0.09}} & \mc{1}{r}{0} & \mc{1}{r}{0.10} & \mc{1}{r}{0.10} \\  

  \mc{1}{l|}{\texttt{hawaii-AM3}} & \mc{1}{r}{\numprint{24436}} & \mc{1}{r}{70.38} & \mc{1}{r|}{-} & \mc{1}{r}{\numprint{24436}} & \mc{1}{r}{743.04} & \mc{1}{r|}{-} & \mc{1}{r}{\numprint{24184}} & \mc{1}{r}{98.22} & \mc{1}{r|}{-} & \mc{1}{r}{\numprint{22997}} & \mc{1}{r}{118.52} & \mc{1}{r|}{-} & \mc{1}{r}{\numprint{21087}} & \mc{1}{r}{632.02} & \mc{1}{r}{-} \\  

  \mc{1}{l|}{\texttt{idaho-AM3}} & \mc{1}{r}{\numprint{3208}} & \mc{1}{r}{3.17} & \mc{1}{r|}{-} & \mc{1}{r}{\numprint{3208}} & \mc{1}{r}{29.91} & \mc{1}{r|}{-} & \mc{1}{r}{\numprint{3204}} & \mc{1}{r}{6.96} & \mc{1}{r|}{-} & \mc{1}{r}{\numprint{3160}} & \mc{1}{r}{8.74} & \mc{1}{r|}{-} & \mc{1}{r}{\numprint{2909}} & \mc{1}{r}{33.77} & \mc{1}{r}{-} \\  

  \mc{1}{l|}{\texttt{kansas-AM3}} & \mc{1}{r}{\numprint{1605}} & \mc{1}{r}{2.46} & \mc{1}{r|}{-} & \mc{1}{r}{\numprint{1605}} & \mc{1}{r}{4.81} & \mc{1}{r|}{-} & \mc{1}{r}{\numprint{1550}} & \mc{1}{r}{0.49} & \mc{1}{r|}{-} & \mc{1}{r}{903} & \mc{1}{r}{2.46} & \mc{1}{r|}{\textbf{430.93}} & \mc{1}{r}{860} & \mc{1}{r}{41.61} & \mc{1}{r}{489.15} \\  

  \rowcolor{lightergray} \mc{1}{l|}{\texttt{kentucky-AM2}} & \mc{1}{r}{442} & \mc{1}{r}{2.05} & \mc{1}{r|}{11.85} & \mc{1}{r}{442} & \mc{1}{r}{2.19} & \mc{1}{r|}{67.28} & \mc{1}{r}{183} & \mc{1}{r}{0.20} & \mc{1}{r|}{0.39} & \mc{1}{r}{0} & \mc{1}{r}{0.22} & \mc{1}{r|}{\textbf{0.23}} & \mc{1}{r}{0} & \mc{1}{r}{0.41} & \mc{1}{r}{0.42} \\  

  \mc{1}{l|}{\texttt{kentucky-AM3}} & \mc{1}{r}{\numprint{16871}} & \mc{1}{r}{109.47} & \mc{1}{r|}{-} & \mc{1}{r}{\numprint{16871}} & \mc{1}{r}{\numprint{3344.67}} & \mc{1}{r|}{-} & \mc{1}{r}{\numprint{16807}} & \mc{1}{r}{237.86} & \mc{1}{r|}{-} & \mc{1}{r}{\numprint{15947}} & \mc{1}{r}{298.49} & \mc{1}{r|}{-} & \mc{1}{r}{\numprint{15684}} & \mc{1}{r}{705.46} & \mc{1}{r}{-} \\  

  \rowcolor{lightergray} \mc{1}{l|}{\texttt{louisiana-AM3}} & \mc{1}{r}{382} & \mc{1}{r}{4.56} & \mc{1}{r|}{6.55} & \mc{1}{r}{382} & \mc{1}{r}{5.04} & \mc{1}{r|}{25.22} & \mc{1}{r}{349} & \mc{1}{r}{0.03} & \mc{1}{r|}{0.82} & \mc{1}{r}{0} & \mc{1}{r}{0.07} & \mc{1}{r|}{\textbf{0.07}} & \mc{1}{r}{0} & \mc{1}{r}{0.16} & \mc{1}{r}{0.16} \\  

  \rowcolor{lightergray} \mc{1}{l|}{\texttt{maryland-AM3}} & \mc{1}{r}{187} & \mc{1}{r}{7.59} & \mc{1}{r|}{8.49} & \mc{1}{r}{187} & \mc{1}{r}{8.65} & \mc{1}{r|}{10.73} & \mc{1}{r}{335} & \mc{1}{r}{0.03} & \mc{1}{r|}{0.19} & \mc{1}{r}{0} & \mc{1}{r}{0.11} & \mc{1}{r|}{\textbf{0.11}} & \mc{1}{r}{0} & \mc{1}{r}{0.15} & \mc{1}{r}{0.15} \\  

  \rowcolor{lightergray} \mc{1}{l|}{\texttt{massachusetts-AM2}} & \mc{1}{r}{196} & \mc{1}{r}{0.04} & \mc{1}{r|}{0.36} & \mc{1}{r}{196} & \mc{1}{r}{0.04} & \mc{1}{r|}{0.58} & \mc{1}{r}{193} & \mc{1}{r}{0.02} & \mc{1}{r|}{0.07} & \mc{1}{r}{0} & \mc{1}{r}{0.06} & \mc{1}{r|}{\textbf{0.06}} & \mc{1}{r}{0} & \mc{1}{r}{0.07} & \mc{1}{r}{0.07} \\  

  \mc{1}{l|}{\texttt{massachusetts-AM3}} & \mc{1}{r}{\numprint{2008}} & \mc{1}{r}{9.42} & \mc{1}{r|}{-} & \mc{1}{r}{\numprint{2008}} & \mc{1}{r}{12.62} & \mc{1}{r|}{-} & \mc{1}{r}{\numprint{1928}} & \mc{1}{r}{0.36} & \mc{1}{r|}{-} & \mc{1}{r}{\numprint{1636}} & \mc{1}{r}{1.08} & \mc{1}{r|}{-} & \mc{1}{r}{\numprint{1632}} & \mc{1}{r}{31.83} & \mc{1}{r}{-} \\  

  \rowcolor{lightergray} \mc{1}{l|}{\texttt{mexico-AM3}} & \mc{1}{r}{620} & \mc{1}{r}{25.29} & \mc{1}{r|}{80.23} & \mc{1}{r}{620} & \mc{1}{r}{27.52} & \mc{1}{r|}{991.99} & \mc{1}{r}{514} & \mc{1}{r}{0.03} & \mc{1}{r|}{\textbf{1.47}} & \mc{1}{r}{483} & \mc{1}{r}{0.28} & \mc{1}{r|}{1.50} & \mc{1}{r}{0} & \mc{1}{r}{21.03} & \mc{1}{r}{21.30} \\  

  \rowcolor{lightergray} \mc{1}{l|}{\texttt{new-hampshire-AM3}} & \mc{1}{r}{247} & \mc{1}{r}{4.99} & \mc{1}{r|}{6.19} & \mc{1}{r}{247} & \mc{1}{r}{5.69} & \mc{1}{r|}{15.89} & \mc{1}{r}{164} & \mc{1}{r}{0.02} & \mc{1}{r|}{0.17} & \mc{1}{r}{0} & \mc{1}{r}{0.07} & \mc{1}{r|}{\textbf{0.07}} & \mc{1}{r}{0} & \mc{1}{r}{0.09} & \mc{1}{r}{0.09} \\  

  \mc{1}{l|}{\texttt{north-carolina-AM3}} & \mc{1}{r}{\numprint{1178}} & \mc{1}{r}{0.69} & \mc{1}{r|}{-} & \mc{1}{r}{\numprint{1178}} & \mc{1}{r}{1.22} & \mc{1}{r|}{-} & \mc{1}{r}{\numprint{1146}} & \mc{1}{r}{0.25} & \mc{1}{r|}{-} & \mc{1}{r}{\numprint{1144}} & \mc{1}{r}{0.43} & \mc{1}{r|}{-} & \mc{1}{r}{700} & \mc{1}{r}{47.38} & \mc{1}{r}{379.088} \\  

  \rowcolor{lightergray} \mc{1}{l|}{\texttt{oregon-AM2}} & \mc{1}{r}{35} & \mc{1}{r}{0.04} & \mc{1}{r|}{0.05} & \mc{1}{r}{35} & \mc{1}{r}{0.05} & \mc{1}{r|}{0.05} & \mc{1}{r}{0} & \mc{1}{r}{0.01} & \mc{1}{r|}{\textbf{0.01}} & \mc{1}{r}{0} & \mc{1}{r}{0.02} & \mc{1}{r|}{0.02} & \mc{1}{r}{0} & \mc{1}{r}{0.01} & \mc{1}{r}{0.01} \\  

  \mc{1}{l|}{\texttt{oregon-AM3}} & \mc{1}{r}{\numprint{3670}} & \mc{1}{r}{9.95} & \mc{1}{r|}{-} & \mc{1}{r}{\numprint{3670}} & \mc{1}{r}{34.95} & \mc{1}{r|}{-} & \mc{1}{r}{\numprint{3584}} & \mc{1}{r}{3.92} & \mc{1}{r|}{-} & \mc{1}{r}{\numprint{3417}} & \mc{1}{r}{6.21} & \mc{1}{r|}{-} & \mc{1}{r}{\numprint{2721}} & \mc{1}{r}{38.72} & \mc{1}{r}{-} \\  

  \rowcolor{lightergray} \mc{1}{l|}{\texttt{pennsylvania-AM3}} & \mc{1}{r}{315} & \mc{1}{r}{16.69} & \mc{1}{r|}{20.71} & \mc{1}{r}{315} & \mc{1}{r}{19.39} & \mc{1}{r|}{113.87} & \mc{1}{r}{317} & \mc{1}{r}{0.03} & \mc{1}{r|}{0.39} & \mc{1}{r}{0} & \mc{1}{r}{0.07} & \mc{1}{r|}{\textbf{0.07}} & \mc{1}{r}{0} & \mc{1}{r}{0.12} & \mc{1}{r}{0.12} \\  

  \rowcolor{lightergray} \mc{1}{l|}{\texttt{rhode-island-AM2}} & \mc{1}{r}{\numprint{1103}} & \mc{1}{r}{0.55} & \mc{1}{r|}{-} & \mc{1}{r}{\numprint{1103}} & \mc{1}{r}{0.68} & \mc{1}{r|}{-} & \mc{1}{r}{845} & \mc{1}{r}{0.17} & \mc{1}{r|}{163.07} & \mc{1}{r}{0} & \mc{1}{r}{0.53} & \mc{1}{r|}{\textbf{0.53}} & \mc{1}{r}{0} & \mc{1}{r}{4.57} & \mc{1}{r}{4.58} \\  

  \mc{1}{l|}{\texttt{rhode-island-AM3}} & \mc{1}{r}{\numprint{13031}} & \mc{1}{r}{7.75} & \mc{1}{r|}{-} & \mc{1}{r}{\numprint{13031}} & \mc{1}{r}{193.76} & \mc{1}{r|}{-} & \mc{1}{r}{\numprint{12934}} & \mc{1}{r}{26.54} & \mc{1}{r|}{-} & \mc{1}{r}{\numprint{12653}} & \mc{1}{r}{29.75} & \mc{1}{r|}{-} & \mc{1}{r}{\numprint{12653}} & \mc{1}{r}{59.69} & \mc{1}{r}{-} \\  

  \rowcolor{lightergray} \mc{1}{l|}{\texttt{utah-AM3}} & \mc{1}{r}{568} & \mc{1}{r}{8.21} & \mc{1}{r|}{51.91} & \mc{1}{r}{568} & \mc{1}{r}{8.97} & \mc{1}{r|}{276.27} & \mc{1}{r}{396} & \mc{1}{r}{0.03} & \mc{1}{r|}{0.87} & \mc{1}{r}{0} & \mc{1}{r}{0.09} & \mc{1}{r|}{\textbf{0.09}} & \mc{1}{r}{0} & \mc{1}{r}{0.40} & \mc{1}{r}{0.41} \\  

  \mc{1}{l|}{\texttt{vermont-AM3}} & \mc{1}{r}{\numprint{2630}} & \mc{1}{r}{4.79} & \mc{1}{r|}{-} & \mc{1}{r}{\numprint{2630}} & \mc{1}{r}{9.82} & \mc{1}{r|}{-} & \mc{1}{r}{\numprint{2289}} & \mc{1}{r}{0.97} & \mc{1}{r|}{-} & \mc{1}{r}{\numprint{2069}} & \mc{1}{r}{1.37} & \mc{1}{r|}{-} & \mc{1}{r}{\numprint{2045}} & \mc{1}{r}{55.28} & \mc{1}{r}{-} \\  

  \rowcolor{lightergray} \mc{1}{l|}{\texttt{virginia-AM2}} & \mc{1}{r}{237} & \mc{1}{r}{0.13} & \mc{1}{r|}{0.61} & \mc{1}{r}{237} & \mc{1}{r}{0.12} & \mc{1}{r|}{0.99} & \mc{1}{r}{0} & \mc{1}{r}{0.03} & \mc{1}{r|}{\textbf{0.03}} & \mc{1}{r}{0} & \mc{1}{r}{0.03} & \mc{1}{r|}{0.03} & \mc{1}{r}{0} & \mc{1}{r}{0.03} & \mc{1}{r}{0.03} \\  

  \mc{1}{l|}{\texttt{virginia-AM3}} & \mc{1}{r}{\numprint{3867}} & \mc{1}{r}{34.13} & \mc{1}{r|}{-} & \mc{1}{r}{\numprint{3867}} & \mc{1}{r}{39.74} & \mc{1}{r|}{-} & \mc{1}{r}{\numprint{3738}} & \mc{1}{r}{0.40} & \mc{1}{r|}{-} & \mc{1}{r}{\numprint{2827}} & \mc{1}{r}{1.28} & \mc{1}{r|}{-} & \mc{1}{r}{\numprint{2547}} & \mc{1}{r}{81.67} & \mc{1}{r}{-} \\  

  \rowcolor{lightergray} \mc{1}{l|}{\texttt{washington-AM2}} & \mc{1}{r}{382} & \mc{1}{r}{0.24} & \mc{1}{r|}{5.31} & \mc{1}{r}{382} & \mc{1}{r}{0.18} & \mc{1}{r|}{8.58} & \mc{1}{r}{171} & \mc{1}{r}{0.05} & \mc{1}{r|}{0.37} & \mc{1}{r}{0} & \mc{1}{r}{0.06} & \mc{1}{r|}{\textbf{0.06}} & \mc{1}{r}{0} & \mc{1}{r}{0.07} & \mc{1}{r}{0.07} \\  

  \mc{1}{l|}{\texttt{washington-AM3}} & \mc{1}{r}{\numprint{8030}} & \mc{1}{r}{50.21} & \mc{1}{r|}{-} & \mc{1}{r}{\numprint{8030}} & \mc{1}{r}{67.00} & \mc{1}{r|}{-} & \mc{1}{r}{\numprint{7649}} & \mc{1}{r}{2.19} & \mc{1}{r|}{-} & \mc{1}{r}{\numprint{6895}} & \mc{1}{r}{3.12} & \mc{1}{r|}{-} & \mc{1}{r}{\numprint{6159}} & \mc{1}{r}{73.52} & \mc{1}{r}{-} \\  

  \mc{1}{l|}{\texttt{west-virginia-AM3}} & \mc{1}{r}{991} & \mc{1}{r}{10.69} & \mc{1}{r|}{-} & \mc{1}{r}{991} & \mc{1}{r}{12.13} & \mc{1}{r|}{-} & \mc{1}{r}{970} & \mc{1}{r}{0.08} & \mc{1}{r|}{238.39} & \mc{1}{r}{890} & \mc{1}{r}{0.33} & \mc{1}{r|}{\textbf{155.49}} & \mc{1}{r}{881} & \mc{1}{r}{38.73} & \mc{1}{r}{241.68} \\  

\end{tabular}
\end{ThreePartTable}

\caption{Obtained irreducible graph sizes $n$, time $t_r$ (in seconds) needed to obtain them and total solving time $t_t$ (in seconds) on FE and OSM instances.
The global best solving time $t_t$ is highlighted in bold.Rows are highlighted in gray if one of our algorithms is able to obtain an empty graph.}
\label{bnr_table_fe}
\end{table}

\begin{table*}[h]
\centering
\scriptsize
\begin{ThreePartTable}
\begin{tabular}{cccccccccccccccc}

  \mc{1}{l|}{Graph} & \mc{1}{r}{$n$} & \mc{1}{r}{$t_{r}$} & \mc{1}{r|}{$t_{t}$} & \mc{1}{r}{$n$} & \mc{1}{r}{$t_{r}$} & \mc{1}{r|}{$t_{t}$} & \mc{1}{r}{$n$} & \mc{1}{r}{$t_{r}$} & \mc{1}{r|}{$t_{t}$} & \mc{1}{r}{$n$} & \mc{1}{r}{$t_{r}$} & \mc{1}{r|}{$t_{t}$} & \mc{1}{r}{$n$} & \mc{1}{r}{$t_{r}$} & \mc{1}{r}{$t_{t}$} \\ 
  \hline  

  \mc{1}{l|}{mesh instances} & \mc{3}{c|}{\basicDense} & \mc{3}{c|}{\basicSparse} & \mc{3}{c|}{\nonIncreasing} & \mc{3}{c|}{\cyclicFast} & \mc{3}{c}{\cyclicStrong} \\  
  \hline  

  \rowcolor{lightergray} \mc{1}{l|}{\texttt{beethoven}} & \mc{1}{r}{\numprint{1254}} & \mc{1}{r}{0.02} & \mc{1}{r|}{7.86} & \mc{1}{r}{427} & \mc{1}{r}{0.02} & \mc{1}{r|}{0.08} & \mc{1}{r}{0} & \mc{1}{r}{0.01} & \mc{1}{r|}{\textbf{0.01}} & \mc{1}{r}{0} & \mc{1}{r}{0.01} & \mc{1}{r|}{0.01} & \mc{1}{r}{0} & \mc{1}{r}{0.01} & \mc{1}{r}{0.01} \\  

  \rowcolor{lightergray} \mc{1}{l|}{\texttt{blob}} & \mc{1}{r}{\numprint{5746}} & \mc{1}{r}{0.08} & \mc{1}{r|}{-} & \mc{1}{r}{\numprint{1464}} & \mc{1}{r}{0.06} & \mc{1}{r|}{0.20} & \mc{1}{r}{0} & \mc{1}{r}{0.03} & \mc{1}{r|}{\textbf{0.03}} & \mc{1}{r}{0} & \mc{1}{r}{0.03} & \mc{1}{r|}{0.04} & \mc{1}{r}{0} & \mc{1}{r}{0.03} & \mc{1}{r}{0.03} \\  

  \rowcolor{lightergray} \mc{1}{l|}{\texttt{buddha}} & \mc{1}{r}{\numprint{380315}} & \mc{1}{r}{5.56} & \mc{1}{r|}{-} & \mc{1}{r}{\numprint{107265}} & \mc{1}{r}{26.19} & \mc{1}{r|}{67.85} & \mc{1}{r}{86} & \mc{1}{r}{1.83} & \mc{1}{r|}{2.74} & \mc{1}{r}{0} & \mc{1}{r}{1.87} & \mc{1}{r|}{\textbf{2.26}} & \mc{1}{r}{0} & \mc{1}{r}{1.91} & \mc{1}{r}{2.39} \\  

  \rowcolor{lightergray} \mc{1}{l|}{\texttt{bunny}} & \mc{1}{r}{\numprint{24580}} & \mc{1}{r}{0.34} & \mc{1}{r|}{-} & \mc{1}{r}{\numprint{3290}} & \mc{1}{r}{0.56} & \mc{1}{r|}{0.89} & \mc{1}{r}{0} & \mc{1}{r}{0.12} & \mc{1}{r|}{\textbf{0.14}} & \mc{1}{r}{0} & \mc{1}{r}{0.13} & \mc{1}{r|}{0.16} & \mc{1}{r}{0} & \mc{1}{r}{0.15} & \mc{1}{r}{0.18} \\  

  \rowcolor{lightergray} \mc{1}{l|}{\texttt{cow}} & \mc{1}{r}{\numprint{1916}} & \mc{1}{r}{0.02} & \mc{1}{r|}{-} & \mc{1}{r}{513} & \mc{1}{r}{0.02} & \mc{1}{r|}{0.06} & \mc{1}{r}{0} & \mc{1}{r}{0.01} & \mc{1}{r|}{0.01} & \mc{1}{r}{0} & \mc{1}{r}{0.01} & \mc{1}{r|}{\textbf{0.01}} & \mc{1}{r}{0} & \mc{1}{r}{0.01} & \mc{1}{r}{0.01} \\  

  \rowcolor{lightergray} \mc{1}{l|}{\texttt{dragon}} & \mc{1}{r}{\numprint{51885}} & \mc{1}{r}{0.89} & \mc{1}{r|}{-} & \mc{1}{r}{\numprint{12893}} & \mc{1}{r}{1.34} & \mc{1}{r|}{3.83} & \mc{1}{r}{0} & \mc{1}{r}{0.18} & \mc{1}{r|}{\textbf{0.21}} & \mc{1}{r}{0} & \mc{1}{r}{0.19} & \mc{1}{r|}{0.23} & \mc{1}{r}{0} & \mc{1}{r}{0.21} & \mc{1}{r}{0.25} \\  

  \rowcolor{lightergray} \mc{1}{l|}{\texttt{dragonsub}} & \mc{1}{r}{\numprint{218779}} & \mc{1}{r}{2.60} & \mc{1}{r|}{-} & \mc{1}{r}{\numprint{19470}} & \mc{1}{r}{4.15} & \mc{1}{r|}{5.66} & \mc{1}{r}{506} & \mc{1}{r}{1.03} & \mc{1}{r|}{2.08} & \mc{1}{r}{0} & \mc{1}{r}{1.13} & \mc{1}{r|}{\textbf{1.36}} & \mc{1}{r}{0} & \mc{1}{r}{1.07} & \mc{1}{r}{1.28} \\  

  \rowcolor{lightergray} \mc{1}{l|}{\texttt{ecat}} & \mc{1}{r}{\numprint{239787}} & \mc{1}{r}{4.07} & \mc{1}{r|}{-} & \mc{1}{r}{\numprint{26270}} & \mc{1}{r}{10.09} & \mc{1}{r|}{12.93} & \mc{1}{r}{274} & \mc{1}{r}{2.12} & \mc{1}{r|}{3.16} & \mc{1}{r}{0} & \mc{1}{r}{2.12} & \mc{1}{r|}{\textbf{2.51}} & \mc{1}{r}{0} & \mc{1}{r}{2.14} & \mc{1}{r}{2.56} \\  

  \rowcolor{lightergray} \mc{1}{l|}{\texttt{face}} & \mc{1}{r}{\numprint{7588}} & \mc{1}{r}{0.09} & \mc{1}{r|}{-} & \mc{1}{r}{\numprint{1540}} & \mc{1}{r}{0.10} & \mc{1}{r|}{0.21} & \mc{1}{r}{0} & \mc{1}{r}{0.03} & \mc{1}{r|}{0.04} & \mc{1}{r}{0} & \mc{1}{r}{0.03} & \mc{1}{r|}{\textbf{0.03}} & \mc{1}{r}{0} & \mc{1}{r}{0.03} & \mc{1}{r}{0.04} \\  

  \rowcolor{lightergray} \mc{1}{l|}{\texttt{fandisk}} & \mc{1}{r}{\numprint{2851}} & \mc{1}{r}{0.05} & \mc{1}{r|}{-} & \mc{1}{r}{336} & \mc{1}{r}{0.03} & \mc{1}{r|}{0.07} & \mc{1}{r}{51} & \mc{1}{r}{0.02} & \mc{1}{r|}{0.03} & \mc{1}{r}{0} & \mc{1}{r}{0.02} & \mc{1}{r|}{\textbf{0.02}} & \mc{1}{r}{0} & \mc{1}{r}{0.02} & \mc{1}{r}{0.02} \\  

  \rowcolor{lightergray} \mc{1}{l|}{\texttt{feline}} & \mc{1}{r}{\numprint{14817}} & \mc{1}{r}{0.20} & \mc{1}{r|}{-} & \mc{1}{r}{\numprint{2743}} & \mc{1}{r}{0.25} & \mc{1}{r|}{0.47} & \mc{1}{r}{0} & \mc{1}{r}{0.08} & \mc{1}{r|}{0.09} & \mc{1}{r}{0} & \mc{1}{r}{0.08} & \mc{1}{r|}{\textbf{0.09}} & \mc{1}{r}{0} & \mc{1}{r}{0.08} & \mc{1}{r}{0.09} \\  

  \rowcolor{lightergray} \mc{1}{l|}{\texttt{gameguy}} & \mc{1}{r}{\numprint{13959}} & \mc{1}{r}{0.17} & \mc{1}{r|}{-} & \mc{1}{r}{312} & \mc{1}{r}{0.10} & \mc{1}{r|}{0.12} & \mc{1}{r}{0} & \mc{1}{r}{0.06} & \mc{1}{r|}{0.07} & \mc{1}{r}{0} & \mc{1}{r}{0.06} & \mc{1}{r|}{\textbf{0.07}} & \mc{1}{r}{0} & \mc{1}{r}{0.06} & \mc{1}{r}{0.07} \\  

  \rowcolor{lightergray} \mc{1}{l|}{\texttt{gargoyle}} & \mc{1}{r}{\numprint{6512}} & \mc{1}{r}{0.15} & \mc{1}{r|}{-} & \mc{1}{r}{\numprint{1819}} & \mc{1}{r}{0.14} & \mc{1}{r|}{0.36} & \mc{1}{r}{0} & \mc{1}{r}{0.03} & \mc{1}{r|}{0.03} & \mc{1}{r}{0} & \mc{1}{r}{0.03} & \mc{1}{r|}{\textbf{0.03}} & \mc{1}{r}{0} & \mc{1}{r}{0.03} & \mc{1}{r}{0.03} \\  

  \rowcolor{lightergray} \mc{1}{l|}{\texttt{turtle}} & \mc{1}{r}{\numprint{91624}} & \mc{1}{r}{1.17} & \mc{1}{r|}{-} & \mc{1}{r}{\numprint{16095}} & \mc{1}{r}{1.92} & \mc{1}{r|}{4.98} & \mc{1}{r}{186} & \mc{1}{r}{0.42} & \mc{1}{r|}{0.65} & \mc{1}{r}{0} & \mc{1}{r}{0.41} & \mc{1}{r|}{\textbf{0.49}} & \mc{1}{r}{0} & \mc{1}{r}{0.47} & \mc{1}{r}{0.56} \\  

  \rowcolor{lightergray} \mc{1}{l|}{\texttt{venus}} & \mc{1}{r}{\numprint{1898}} & \mc{1}{r}{0.02} & \mc{1}{r|}{-} & \mc{1}{r}{175} & \mc{1}{r}{0.01} & \mc{1}{r|}{0.02} & \mc{1}{r}{0} & \mc{1}{r}{0.01} & \mc{1}{r|}{0.01} & \mc{1}{r}{0} & \mc{1}{r}{0.01} & \mc{1}{r|}{\textbf{0.01}} & \mc{1}{r}{0} & \mc{1}{r}{0.01} & \mc{1}{r}{0.01} \\  

  \hline \hline  

  \mc{1}{l|}{SNAP instances} & \mc{3}{c|}{\basicDense} & \mc{3}{c|}{\basicSparse} & \mc{3}{c|}{\nonIncreasing} & \mc{3}{c|}{\cyclicFast} & \mc{3}{c}{\cyclicStrong} \\  
  \hline  

  \mc{1}{l|}{\texttt{as-skitter}} & \mc{1}{r}{\numprint{26584}} & \mc{1}{r}{25.82} & \mc{1}{r|}{-} & \mc{1}{r}{\numprint{8585}} & \mc{1}{r}{36.69} & \mc{1}{r|}{-} & \mc{1}{r}{\numprint{3426}} & \mc{1}{r}{4.75} & \mc{1}{r|}{-} & \mc{1}{r}{\numprint{2782}} & \mc{1}{r}{5.50} & \mc{1}{r|}{-} & \mc{1}{r}{\numprint{2343}} & \mc{1}{r}{6.80} & \mc{1}{r}{-} \\  

  \rowcolor{lightergray} \mc{1}{l|}{\texttt{ca-AstroPh}} & \mc{1}{r}{0} & \mc{1}{r}{0.02} & \mc{1}{r|}{\textbf{0.03}} & \mc{1}{r}{0} & \mc{1}{r}{0.02} & \mc{1}{r|}{0.03} & \mc{1}{r}{0} & \mc{1}{r}{0.02} & \mc{1}{r|}{0.03} & \mc{1}{r}{0} & \mc{1}{r}{0.03} & \mc{1}{r|}{0.04} & \mc{1}{r}{0} & \mc{1}{r}{0.03} & \mc{1}{r}{0.03} \\  

  \rowcolor{lightergray} \mc{1}{l|}{\texttt{ca-CondMat}} & \mc{1}{r}{0} & \mc{1}{r}{0.02} & \mc{1}{r|}{0.03} & \mc{1}{r}{0} & \mc{1}{r}{0.01} & \mc{1}{r|}{0.02} & \mc{1}{r}{0} & \mc{1}{r}{0.01} & \mc{1}{r|}{\textbf{0.02}} & \mc{1}{r}{0} & \mc{1}{r}{0.03} & \mc{1}{r|}{0.03} & \mc{1}{r}{0} & \mc{1}{r}{0.01} & \mc{1}{r}{0.02} \\  

  \rowcolor{lightergray} \mc{1}{l|}{\texttt{ca-GrQc}} & \mc{1}{r}{0} & \mc{1}{r}{0.00} & \mc{1}{r|}{0.00} & \mc{1}{r}{0} & \mc{1}{r}{0.00} & \mc{1}{r|}{0.00} & \mc{1}{r}{0} & \mc{1}{r}{0.00} & \mc{1}{r|}{\textbf{0.00}} & \mc{1}{r}{0} & \mc{1}{r}{0.00} & \mc{1}{r|}{0.00} & \mc{1}{r}{0} & \mc{1}{r}{0.00} & \mc{1}{r}{0.00} \\  

  \rowcolor{lightergray} \mc{1}{l|}{\texttt{ca-HepPh}} & \mc{1}{r}{0} & \mc{1}{r}{0.01} & \mc{1}{r|}{0.02} & \mc{1}{r}{0} & \mc{1}{r}{0.01} & \mc{1}{r|}{0.02} & \mc{1}{r}{0} & \mc{1}{r}{0.01} & \mc{1}{r|}{\textbf{0.01}} & \mc{1}{r}{0} & \mc{1}{r}{0.01} & \mc{1}{r|}{0.02} & \mc{1}{r}{0} & \mc{1}{r}{0.01} & \mc{1}{r}{0.01} \\  

  \rowcolor{lightergray} \mc{1}{l|}{\texttt{ca-HepTh}} & \mc{1}{r}{0} & \mc{1}{r}{0.01} & \mc{1}{r|}{0.01} & \mc{1}{r}{0} & \mc{1}{r}{0.00} & \mc{1}{r|}{0.01} & \mc{1}{r}{0} & \mc{1}{r}{0.01} & \mc{1}{r|}{\textbf{0.01}} & \mc{1}{r}{0} & \mc{1}{r}{0.01} & \mc{1}{r|}{0.01} & \mc{1}{r}{0} & \mc{1}{r}{0.00} & \mc{1}{r}{0.00} \\  

  \rowcolor{lightergray} \mc{1}{l|}{\texttt{email-Enron}} & \mc{1}{r}{0} & \mc{1}{r}{0.02} & \mc{1}{r|}{\textbf{0.03}} & \mc{1}{r}{0} & \mc{1}{r}{0.02} & \mc{1}{r|}{0.03} & \mc{1}{r}{0} & \mc{1}{r}{0.04} & \mc{1}{r|}{0.04} & \mc{1}{r}{0} & \mc{1}{r}{0.03} & \mc{1}{r|}{0.03} & \mc{1}{r}{0} & \mc{1}{r}{0.03} & \mc{1}{r}{0.03} \\  

  \rowcolor{lightergray} \mc{1}{l|}{\texttt{email-EuAll}} & \mc{1}{r}{0} & \mc{1}{r}{0.08} & \mc{1}{r|}{0.17} & \mc{1}{r}{0} & \mc{1}{r}{0.09} & \mc{1}{r|}{0.16} & \mc{1}{r}{0} & \mc{1}{r}{0.06} & \mc{1}{r|}{\textbf{0.08}} & \mc{1}{r}{0} & \mc{1}{r}{0.09} & \mc{1}{r|}{0.13} & \mc{1}{r}{0} & \mc{1}{r}{0.07} & \mc{1}{r}{0.10} \\  

  \rowcolor{lightergray} \mc{1}{l|}{\texttt{p2p-Gnutella04}} & \mc{1}{r}{0} & \mc{1}{r}{0.01} & \mc{1}{r|}{0.01} & \mc{1}{r}{0} & \mc{1}{r}{0.01} & \mc{1}{r|}{0.01} & \mc{1}{r}{0} & \mc{1}{r}{0.01} & \mc{1}{r|}{\textbf{0.01}} & \mc{1}{r}{0} & \mc{1}{r}{0.01} & \mc{1}{r|}{0.01} & \mc{1}{r}{0} & \mc{1}{r}{0.01} & \mc{1}{r}{0.01} \\  

  \rowcolor{lightergray} \mc{1}{l|}{\texttt{p2p-Gnutella05}} & \mc{1}{r}{0} & \mc{1}{r}{0.01} & \mc{1}{r|}{\textbf{0.01}} & \mc{1}{r}{0} & \mc{1}{r}{0.01} & \mc{1}{r|}{0.01} & \mc{1}{r}{0} & \mc{1}{r}{0.01} & \mc{1}{r|}{0.01} & \mc{1}{r}{0} & \mc{1}{r}{0.01} & \mc{1}{r|}{0.01} & \mc{1}{r}{0} & \mc{1}{r}{0.01} & \mc{1}{r}{0.01} \\  

  \rowcolor{lightergray} \mc{1}{l|}{\texttt{p2p-Gnutella06}} & \mc{1}{r}{0} & \mc{1}{r}{0.01} & \mc{1}{r|}{0.01} & \mc{1}{r}{0} & \mc{1}{r}{0.01} & \mc{1}{r|}{0.01} & \mc{1}{r}{0} & \mc{1}{r}{0.01} & \mc{1}{r|}{\textbf{0.01}} & \mc{1}{r}{0} & \mc{1}{r}{0.01} & \mc{1}{r|}{0.01} & \mc{1}{r}{0} & \mc{1}{r}{0.01} & \mc{1}{r}{0.01} \\  

  \rowcolor{lightergray} \mc{1}{l|}{\texttt{p2p-Gnutella08}} & \mc{1}{r}{0} & \mc{1}{r}{0.00} & \mc{1}{r|}{0.00} & \mc{1}{r}{0} & \mc{1}{r}{0.00} & \mc{1}{r|}{0.01} & \mc{1}{r}{0} & \mc{1}{r}{0.00} & \mc{1}{r|}{0.00} & \mc{1}{r}{0} & \mc{1}{r}{0.00} & \mc{1}{r|}{\textbf{0.00}} & \mc{1}{r}{0} & \mc{1}{r}{0.00} & \mc{1}{r}{0.00} \\  

  \rowcolor{lightergray} \mc{1}{l|}{\texttt{p2p-Gnutella09}} & \mc{1}{r}{0} & \mc{1}{r}{0.00} & \mc{1}{r|}{0.01} & \mc{1}{r}{0} & \mc{1}{r}{0.01} & \mc{1}{r|}{0.01} & \mc{1}{r}{0} & \mc{1}{r}{0.00} & \mc{1}{r|}{0.01} & \mc{1}{r}{0} & \mc{1}{r}{0.00} & \mc{1}{r|}{\textbf{0.00}} & \mc{1}{r}{0} & \mc{1}{r}{0.00} & \mc{1}{r}{0.01} \\  

  \rowcolor{lightergray} \mc{1}{l|}{\texttt{p2p-Gnutella24}} & \mc{1}{r}{0} & \mc{1}{r}{0.01} & \mc{1}{r|}{0.02} & \mc{1}{r}{0} & \mc{1}{r}{0.02} & \mc{1}{r|}{0.03} & \mc{1}{r}{0} & \mc{1}{r}{0.01} & \mc{1}{r|}{\textbf{0.01}} & \mc{1}{r}{0} & \mc{1}{r}{0.01} & \mc{1}{r|}{0.02} & \mc{1}{r}{0} & \mc{1}{r}{0.01} & \mc{1}{r}{0.01} \\  

  \rowcolor{lightergray} \mc{1}{l|}{\texttt{p2p-Gnutella25}} & \mc{1}{r}{10} & \mc{1}{r}{0.01} & \mc{1}{r|}{0.02} & \mc{1}{r}{0} & \mc{1}{r}{0.01} & \mc{1}{r|}{0.02} & \mc{1}{r}{0} & \mc{1}{r}{0.01} & \mc{1}{r|}{\textbf{0.01}} & \mc{1}{r}{0} & \mc{1}{r}{0.01} & \mc{1}{r|}{0.02} & \mc{1}{r}{0} & \mc{1}{r}{0.02} & \mc{1}{r}{0.02} \\  

  \rowcolor{lightergray} \mc{1}{l|}{\texttt{p2p-Gnutella30}} & \mc{1}{r}{0} & \mc{1}{r}{0.01} & \mc{1}{r|}{0.02} & \mc{1}{r}{0} & \mc{1}{r}{0.02} & \mc{1}{r|}{0.03} & \mc{1}{r}{0} & \mc{1}{r}{0.02} & \mc{1}{r|}{0.02} & \mc{1}{r}{0} & \mc{1}{r}{0.02} & \mc{1}{r|}{\textbf{0.02}} & \mc{1}{r}{0} & \mc{1}{r}{0.01} & \mc{1}{r}{0.02} \\  

  \rowcolor{lightergray} \mc{1}{l|}{\texttt{p2p-Gnutella31}} & \mc{1}{r}{0} & \mc{1}{r}{0.04} & \mc{1}{r|}{0.07} & \mc{1}{r}{0} & \mc{1}{r}{0.04} & \mc{1}{r|}{0.07} & \mc{1}{r}{0} & \mc{1}{r}{0.03} & \mc{1}{r|}{\textbf{0.03}} & \mc{1}{r}{0} & \mc{1}{r}{0.05} & \mc{1}{r|}{0.06} & \mc{1}{r}{0} & \mc{1}{r}{0.04} & \mc{1}{r}{0.05} \\  

  \rowcolor{lightergray} \mc{1}{l|}{\texttt{roadNet-CA}} & \mc{1}{r}{\numprint{234433}} & \mc{1}{r}{3.96} & \mc{1}{r|}{-} & \mc{1}{r}{\numprint{66406}} & \mc{1}{r}{20.51} & \mc{1}{r|}{437.62} & \mc{1}{r}{478} & \mc{1}{r}{2.14} & \mc{1}{r|}{5.70} & \mc{1}{r}{0} & \mc{1}{r}{2.42} & \mc{1}{r|}{\textbf{3.57}} & \mc{1}{r}{0} & \mc{1}{r}{2.59} & \mc{1}{r}{3.07} \\  

  \rowcolor{lightergray} \mc{1}{l|}{\texttt{roadNet-PA}} & \mc{1}{r}{\numprint{133814}} & \mc{1}{r}{2.43} & \mc{1}{r|}{-} & \mc{1}{r}{\numprint{35442}} & \mc{1}{r}{7.73} & \mc{1}{r|}{23.86} & \mc{1}{r}{300} & \mc{1}{r}{1.05} & \mc{1}{r|}{2.24} & \mc{1}{r}{0} & \mc{1}{r}{1.19} & \mc{1}{r|}{\textbf{1.44}} & \mc{1}{r}{0} & \mc{1}{r}{1.14} & \mc{1}{r}{1.40} \\  

  \rowcolor{lightergray} \mc{1}{l|}{\texttt{roadNet-TX}} & \mc{1}{r}{\numprint{153985}} & \mc{1}{r}{2.65} & \mc{1}{r|}{-} & \mc{1}{r}{\numprint{40350}} & \mc{1}{r}{10.49} & \mc{1}{r|}{24.30} & \mc{1}{r}{882} & \mc{1}{r}{1.23} & \mc{1}{r|}{3.98} & \mc{1}{r}{0} & \mc{1}{r}{1.32} & \mc{1}{r|}{\textbf{1.64}} & \mc{1}{r}{0} & \mc{1}{r}{1.34} & \mc{1}{r}{1.65} \\  

  \rowcolor{lightergray} \mc{1}{l|}{\texttt{soc-Epinions1}} & \mc{1}{r}{7} & \mc{1}{r}{0.05} & \mc{1}{r|}{\textbf{0.07}} & \mc{1}{r}{0} & \mc{1}{r}{0.06} & \mc{1}{r|}{0.08} & \mc{1}{r}{0} & \mc{1}{r}{0.08} & \mc{1}{r|}{0.10} & \mc{1}{r}{0} & \mc{1}{r}{0.07} & \mc{1}{r|}{0.08} & \mc{1}{r}{0} & \mc{1}{r}{0.07} & \mc{1}{r}{0.08} \\  

  \mc{1}{l|}{\texttt{soc-LiveJournal1}} & \mc{1}{r}{\numprint{60041}} & \mc{1}{r}{236.88} & \mc{1}{r|}{-} & \mc{1}{r}{\numprint{29508}} & \mc{1}{r}{213.74} & \mc{1}{r|}{-} & \mc{1}{r}{\numprint{4319}} & \mc{1}{r}{22.27} & \mc{1}{r|}{-} & \mc{1}{r}{\numprint{3530}} & \mc{1}{r}{24.13} & \mc{1}{r|}{-} & \mc{1}{r}{\numprint{1314}} & \mc{1}{r}{37.77} & \mc{1}{r}{-} \\  

  \rowcolor{lightergray} \mc{1}{l|}{\texttt{soc-Slashdot0811}} & \mc{1}{r}{0} & \mc{1}{r}{0.08} & \mc{1}{r|}{0.11} & \mc{1}{r}{0} & \mc{1}{r}{0.08} & \mc{1}{r|}{0.11} & \mc{1}{r}{0} & \mc{1}{r}{0.07} & \mc{1}{r|}{\textbf{0.08}} & \mc{1}{r}{0} & \mc{1}{r}{0.07} & \mc{1}{r|}{0.09} & \mc{1}{r}{0} & \mc{1}{r}{0.06} & \mc{1}{r}{0.07} \\  

  \rowcolor{lightergray} \mc{1}{l|}{\texttt{soc-Slashdot0902}} & \mc{1}{r}{0} & \mc{1}{r}{0.07} & \mc{1}{r|}{\textbf{0.09}} & \mc{1}{r}{0} & \mc{1}{r}{0.07} & \mc{1}{r|}{0.10} & \mc{1}{r}{0} & \mc{1}{r}{0.09} & \mc{1}{r|}{0.11} & \mc{1}{r}{0} & \mc{1}{r}{0.08} & \mc{1}{r|}{0.10} & \mc{1}{r}{0} & \mc{1}{r}{0.10} & \mc{1}{r}{0.12} \\  

  \mc{1}{l|}{\texttt{soc-pokec-relationships}} & \mc{1}{r}{\numprint{926346}} & \mc{1}{r}{299.11} & \mc{1}{r|}{-} & \mc{1}{r}{\numprint{898779}} & \mc{1}{r}{\numprint{1013.39}} & \mc{1}{r|}{-} & \mc{1}{r}{\numprint{808542}} & \mc{1}{r}{188.57} & \mc{1}{r|}{-} & \mc{1}{r}{\numprint{807412}} & \mc{1}{r}{217.83} & \mc{1}{r|}{-} & \mc{1}{r}{\numprint{807395}} & \mc{1}{r}{388.57} & \mc{1}{r}{-} \\  

  \mc{1}{l|}{\texttt{web-BerkStan}} & \mc{1}{r}{\numprint{36637}} & \mc{1}{r}{6.58} & \mc{1}{r|}{-} & \mc{1}{r}{\numprint{16661}} & \mc{1}{r}{8.70} & \mc{1}{r|}{-} & \mc{1}{r}{\numprint{1999}} & \mc{1}{r}{6.86} & \mc{1}{r|}{120.05} & \mc{1}{r}{151} & \mc{1}{r}{6.46} & \mc{1}{r|}{\textbf{6.83}} & \mc{1}{r}{151} & \mc{1}{r}{7.89} & \mc{1}{r}{8.25} \\  

  \mc{1}{l|}{\texttt{web-Google}} & \mc{1}{r}{\numprint{2810}} & \mc{1}{r}{1.57} & \mc{1}{r|}{\textbf{2.40}} & \mc{1}{r}{\numprint{1254}} & \mc{1}{r}{2.42} & \mc{1}{r|}{3.66} & \mc{1}{r}{361} & \mc{1}{r}{1.75} & \mc{1}{r|}{2.95} & \mc{1}{r}{46} & \mc{1}{r}{1.88} & \mc{1}{r|}{2.47} & \mc{1}{r}{46} & \mc{1}{r}{7.97} & \mc{1}{r}{9.24} \\  

  \mc{1}{l|}{\texttt{web-NotreDame}} & \mc{1}{r}{\numprint{13464}} & \mc{1}{r}{1.03} & \mc{1}{r|}{-} & \mc{1}{r}{\numprint{6052}} & \mc{1}{r}{2.03} & \mc{1}{r|}{-} & \mc{1}{r}{\numprint{2460}} & \mc{1}{r}{0.40} & \mc{1}{r|}{-} & \mc{1}{r}{\numprint{2061}} & \mc{1}{r}{0.56} & \mc{1}{r|}{\textbf{1.60}} & \mc{1}{r}{117} & \mc{1}{r}{2.44} & \mc{1}{r}{2.57} \\  

  \rowcolor{lightergray} \mc{1}{l|}{\texttt{web-Stanford}} & \mc{1}{r}{\numprint{14153}} & \mc{1}{r}{1.81} & \mc{1}{r|}{-} & \mc{1}{r}{\numprint{3325}} & \mc{1}{r}{2.45} & \mc{1}{r|}{-} & \mc{1}{r}{112} & \mc{1}{r}{2.25} & \mc{1}{r|}{2.50} & \mc{1}{r}{0} & \mc{1}{r}{1.80} & \mc{1}{r|}{\textbf{1.99}} & \mc{1}{r}{0} & \mc{1}{r}{2.17} & \mc{1}{r}{2.38} \\  

  \rowcolor{lightergray} \mc{1}{l|}{\texttt{wiki-Talk}} & \mc{1}{r}{0} & \mc{1}{r}{1.00} & \mc{1}{r|}{\textbf{1.71}} & \mc{1}{r}{0} & \mc{1}{r}{1.32} & \mc{1}{r|}{1.96} & \mc{1}{r}{0} & \mc{1}{r}{1.26} & \mc{1}{r|}{1.84} & \mc{1}{r}{0} & \mc{1}{r}{1.24} & \mc{1}{r|}{1.80} & \mc{1}{r}{0} & \mc{1}{r}{1.67} & \mc{1}{r}{2.28} \\  

  \rowcolor{lightergray} \mc{1}{l|}{\texttt{wiki-Vote}} & \mc{1}{r}{477} & \mc{1}{r}{0.03} & \mc{1}{r|}{0.12} & \mc{1}{r}{0} & \mc{1}{r}{0.02} & \mc{1}{r|}{0.03} & \mc{1}{r}{0} & \mc{1}{r}{0.02} & \mc{1}{r|}{\textbf{0.02}} & \mc{1}{r}{0} & \mc{1}{r}{0.02} & \mc{1}{r|}{0.02} & \mc{1}{r}{0} & \mc{1}{r}{0.02} & \mc{1}{r}{0.02} \\

\end{tabular}
\end{ThreePartTable}

\caption{Obtained irreducible graph sizes $n$, time $t_r$ (in seconds) needed to obtain them and total solving time $t_t$ (in seconds) on mesh and SNAP instances.
The global best solving time $t_t$ is highlighted in bold. Rows are highlighted in gray if one of our algorithms is able to obtain an empty graph.}
\label{bnr_table_snap}
\end{table*}
\clearpage

\section{State-of-the-Art Comparison}\label{sota_tables}
\begin{table}[h!]
\centering
\scriptsize
\begin{ThreePartTable}
\begin{tabu}{ccccccccccccc}

  \mc{1}{l}{Graph} & \mc{1}{|r}{$t_{max}$} & \mc{1}{r}{$w_{max}$} & \mc{1}{|r}{$t_{max}$} & \mc{1}{r}{$w_{max}$} & \mc{1}{|r}{$t_{max}$} & \mc{1}{r}{$w_{max}$} & \mc{1}{|r}{$t_{max}$} & \mc{1}{r}{$w_{max}$} & \mc{1}{|r}{$t_{max}$} & \mc{1}{r}{$w_{max}$} & \mc{1}{|r}{$t_{max}$} & \mc{1}{r}{$w_{max}$} \\
  \hline  

  \mc{1}{l}{FE instances} & \mc{2}{|c}{DynWVC1} & \mc{2}{|c}{DynWVC2} & \mc{2}{|c}{HILS} & \mc{2}{|c}{Cyclic-Fast} & \mc{2}{|c}{Cyclic-Strong} & \mc{2}{|c}{Non-Increasing} \\
  \hline  

  \mc{1}{l}{\cellcolor{lightgray!50}{\texttt{fe\_4elt2}}} & \mc{1}{|r}{\cellcolor{lightgray!50}{961.12}} & \mc{1}{r}{\cellcolor{lightgray!50}{\numprint{427755}}} & \mc{1}{|r}{\cellcolor{lightgray!50}{974.87}} & \mc{1}{r}{\cellcolor{lightgray!50}{\numprint{427755}}} & \mc{1}{|r}{\cellcolor{lightgray!50}{759.23}} & \mc{1}{r}{\cellcolor{lightgray!50}{\numprint{427646}}} & \mc{1}{|r}{\cellcolor{lightgray!50}{0.11}} & \mc{1}{r}{\textbf{\cellcolor{lightgray!50}{\numprint{428029}}}} & \mc{1}{|r}{\cellcolor{lightgray!50}{0.11}} & \mc{1}{r}{\textbf{\cellcolor{lightgray!50}{\numprint{428029}}}} & \mc{1}{|r}{\cellcolor{lightgray!50}{0.18}} & \mc{1}{r}{\cellcolor{lightgray!50}{\numprint{428016}}} \\

  \mc{1}{l}{\texttt{fe\_body}} & \mc{1}{|r}{504.31} & \mc{1}{r}{\numprint{1678616}} & \mc{1}{|r}{499.03} & \mc{1}{r}{\numprint{1678496}} & \mc{1}{|r}{806.46} & \mc{1}{r}{\numprint{1678708}} & \mc{1}{|r}{0.51} & \mc{1}{r}{\textbf{\numprint{1680182}}} & \mc{1}{|r}{0.86} & \mc{1}{r}{\numprint{1680117}} & \mc{1}{|r}{0.27} & \mc{1}{r}{\numprint{1680133}} \\

  \mc{1}{l}{{\texttt{fe\_ocean}}} & \mc{1}{|r}{{983.53}} & \mc{1}{r}{{\numprint{7222521}}} & \mc{1}{|r}{{379.75}} & \mc{1}{r}{{\numprint{7220128}}} & \mc{1}{|r}{{999.57}} & \mc{1}{r}{{\numprint{7069279}}} & \mc{1}{|r}{{18.85}} & \mc{1}{r}{{\numprint{6591832}}} & \mc{1}{|r}{{19.04}} & \mc{1}{r}{{\numprint{6591537}}} & \mc{1}{|r}{{18.85}} & \mc{1}{r}{{\numprint{6597698}}} \\

  \mc{1}{l}{\texttt{fe\_pwt}} & \mc{1}{|r}{814.23} & \mc{1}{r}{\numprint{1176721}} & \mc{1}{|r}{320.05} & \mc{1}{r}{\textbf{\numprint{1176784}}} & \mc{1}{|r}{932.43} & \mc{1}{r}{\numprint{1175754}} & \mc{1}{|r}{3.03} & \mc{1}{r}{\numprint{1162232}} & \mc{1}{|r}{5.45} & \mc{1}{r}{\numprint{888959}} & \mc{1}{|r}{1.57} & \mc{1}{r}{\numprint{1151777}} \\

  \mc{1}{l}{\texttt{fe\_rotor}} & \mc{1}{|r}{961.76} & \mc{1}{r}{\textbf{\numprint{2659653}}} & \mc{1}{|r}{874.68} & \mc{1}{r}{\numprint{2659473}} & \mc{1}{|r}{973.92} & \mc{1}{r}{\numprint{2650132}} & \mc{1}{|r}{13.95} & \mc{1}{r}{\numprint{2531152}} & \mc{1}{|r}{20.55} & \mc{1}{r}{\numprint{2538117}} & \mc{1}{|r}{13.56} & \mc{1}{r}{\numprint{2532168}} \\

  \mc{1}{l}{\cellcolor{lightgray!50}{\texttt{fe\_sphere}}} & \mc{1}{|r}{\cellcolor{lightgray!50}{875.87}} & \mc{1}{r}{\cellcolor{lightgray!50}{\numprint{616978}}} & \mc{1}{|r}{\cellcolor{lightgray!50}{872.36}} & \mc{1}{r}{\cellcolor{lightgray!50}{\numprint{616978}}} & \mc{1}{|r}{\cellcolor{lightgray!50}{843.67}} & \mc{1}{r}{\cellcolor{lightgray!50}{\numprint{616528}}} & \mc{1}{|r}{\cellcolor{lightgray!50}{0.63}} & \mc{1}{r}{\textbf{\cellcolor{lightgray!50}{\numprint{617816}}}} & \mc{1}{|r}{\cellcolor{lightgray!50}{0.67}} & \mc{1}{r}{\textbf{\cellcolor{lightgray!50}{\numprint{617816}}}} & \mc{1}{|r}{\cellcolor{lightgray!50}{0.46}} & \mc{1}{r}{\cellcolor{lightgray!50}{\numprint{617585}}} \\

  \mc{1}{l}{\cellcolor{lightgray!50}{\texttt{fe\_tooth}}} & \mc{1}{|r}{\cellcolor{lightgray!50}{353.21}} & \mc{1}{r}{\cellcolor{lightgray!50}{\numprint{3031269}}} & \mc{1}{|r}{\cellcolor{lightgray!50}{619.96}} & \mc{1}{r}{\cellcolor{lightgray!50}{\numprint{3031385}}} & \mc{1}{|r}{\cellcolor{lightgray!50}{994.97}} & \mc{1}{r}{\cellcolor{lightgray!50}{\numprint{3032819}}} & \mc{1}{|r}{\cellcolor{lightgray!50}{0.26}} & \mc{1}{r}{\textbf{\cellcolor{lightgray!50}{\numprint{3033298}}}} & \mc{1}{|r}{\cellcolor{lightgray!50}{0.26}} & \mc{1}{r}{\textbf{\cellcolor{lightgray!50}{\numprint{3033298}}}} & \mc{1}{|r}{\cellcolor{lightgray!50}{0.27}} & \mc{1}{r}{\textbf{\cellcolor{lightgray!50}{\numprint{3033298}}}} \\

  \hline \hline  

  \mc{1}{l}{OSM instances} & \mc{2}{|c}{DynWVC1} & \mc{2}{|c}{DynWVC2} & \mc{2}{|c}{HILS} & \mc{2}{|c}{Cyclic-Fast} & \mc{2}{|c}{Cyclic-Strong} & \mc{2}{|c}{Non-Increasing} \\
  \hline  

  \mc{1}{l}{\cellcolor{lightgray!50}{\texttt{alabama-AM2}}} & \mc{1}{|r}{\cellcolor{lightgray!50}{0.18}} & \mc{1}{r}{\cellcolor{lightgray!50}{\numprint{174252}}} & \mc{1}{|r}{\cellcolor{lightgray!50}{0.24}} & \mc{1}{r}{\cellcolor{lightgray!50}{\numprint{174269}}} & \mc{1}{|r}{\cellcolor{lightgray!50}{0.03}} & \mc{1}{r}{\textbf{\cellcolor{lightgray!50}{\numprint{174309}}}} & \mc{1}{|r}{\cellcolor{lightgray!50}{0.01}} & \mc{1}{r}{\textbf{\cellcolor{lightgray!50}{\numprint{174309}}}} & \mc{1}{|r}{\cellcolor{lightgray!50}{0.01}} & \mc{1}{r}{\textbf{\cellcolor{lightgray!50}{\numprint{174309}}}} & \mc{1}{|r}{\cellcolor{lightgray!50}{0.01}} & \mc{1}{r}{\textbf{\cellcolor{lightgray!50}{\numprint{174309}}}} \\

  \mc{1}{l}{\cellcolor{lightgray!50}{\texttt{alabama-AM3}}} & \mc{1}{|r}{\cellcolor{lightgray!50}{725.34}} & \mc{1}{r}{\cellcolor{lightgray!50}{\numprint{185518}}} & \mc{1}{|r}{\cellcolor{lightgray!50}{199.94}} & \mc{1}{r}{\cellcolor{lightgray!50}{\numprint{185655}}} & \mc{1}{|r}{\cellcolor{lightgray!50}{0.58}} & \mc{1}{r}{\textbf{\cellcolor{lightgray!50}{\numprint{185744}}}} & \mc{1}{|r}{\cellcolor{lightgray!50}{1.76}} & \mc{1}{r}{\textbf{\cellcolor{lightgray!50}{\numprint{185744}}}} & \mc{1}{|r}{\cellcolor{lightgray!50}{32.42}} & \mc{1}{r}{\textbf{\cellcolor{lightgray!50}{\numprint{185744}}}} & \mc{1}{|r}{\cellcolor{lightgray!50}{0.60}} & \mc{1}{r}{\textbf{\cellcolor{lightgray!50}{\numprint{185744}}}} \\

  \mc{1}{l}{\cellcolor{lightgray!50}{\texttt{district-of-columbia-AM1}}} & \mc{1}{|r}{\cellcolor{lightgray!50}{23.96}} & \mc{1}{r}{\textbf{\cellcolor{lightgray!50}{\numprint{196475}}}} & \mc{1}{|r}{\cellcolor{lightgray!50}{28.42}} & \mc{1}{r}{\textbf{\cellcolor{lightgray!50}{\numprint{196475}}}} & \mc{1}{|r}{\cellcolor{lightgray!50}{0.14}} & \mc{1}{r}{\textbf{\cellcolor{lightgray!50}{\numprint{196475}}}} & \mc{1}{|r}{\cellcolor{lightgray!50}{0.32}} & \mc{1}{r}{\textbf{\cellcolor{lightgray!50}{\numprint{196475}}}} & \mc{1}{|r}{\cellcolor{lightgray!50}{3.52}} & \mc{1}{r}{\textbf{\cellcolor{lightgray!50}{\numprint{196475}}}} & \mc{1}{|r}{\cellcolor{lightgray!50}{0.06}} & \mc{1}{r}{\textbf{\cellcolor{lightgray!50}{\numprint{196475}}}} \\

  \mc{1}{l}{\texttt{district-of-columbia-AM2}} & \mc{1}{|r}{159.08} & \mc{1}{r}{\numprint{208989}} & \mc{1}{|r}{915.18} & \mc{1}{r}{\numprint{208977}} & \mc{1}{|r}{400.69} & \mc{1}{r}{\textbf{\numprint{209132}}} & \mc{1}{|r}{4.21} & \mc{1}{r}{\textbf{\numprint{209132}}} & \mc{1}{|r}{84.21} & \mc{1}{r}{\numprint{209131}} & \mc{1}{|r}{686.26} & \mc{1}{r}{\numprint{174114}} \\

  \mc{1}{l}{\texttt{district-of-columbia-AM3}} & \mc{1}{|r}{461.10} & \mc{1}{r}{\numprint{224760}} & \mc{1}{|r}{313.17} & \mc{1}{r}{\numprint{223955}} & \mc{1}{|r}{849.37} & \mc{1}{r}{\textbf{\numprint{227613}}} & \mc{1}{|r}{904.91} & \mc{1}{r}{\numprint{142454}} & \mc{1}{|r}{804.79} & \mc{1}{r}{\numprint{156967}} & \mc{1}{|r}{168.55} & \mc{1}{r}{\numprint{120366}} \\

  \mc{1}{l}{\cellcolor{lightgray!50}{\texttt{florida-AM2}}} & \mc{1}{|r}{\cellcolor{lightgray!50}{0.18}} & \mc{1}{r}{\textbf{\cellcolor{lightgray!50}{\numprint{230595}}}} & \mc{1}{|r}{\cellcolor{lightgray!50}{0.53}} & \mc{1}{r}{\textbf{\cellcolor{lightgray!50}{\numprint{230595}}}} & \mc{1}{|r}{\cellcolor{lightgray!50}{0.04}} & \mc{1}{r}{\textbf{\cellcolor{lightgray!50}{\numprint{230595}}}} & \mc{1}{|r}{\cellcolor{lightgray!50}{0.00}} & \mc{1}{r}{\textbf{\cellcolor{lightgray!50}{\numprint{230595}}}} & \mc{1}{|r}{\cellcolor{lightgray!50}{0.00}} & \mc{1}{r}{\textbf{\cellcolor{lightgray!50}{\numprint{230595}}}} & \mc{1}{|r}{\cellcolor{lightgray!50}{0.00}} & \mc{1}{r}{\textbf{\cellcolor{lightgray!50}{\numprint{230595}}}} \\

  \mc{1}{l}{\cellcolor{lightgray!50}{\texttt{florida-AM3}}} & \mc{1}{|r}{\cellcolor{lightgray!50}{425.87}} & \mc{1}{r}{\cellcolor{lightgray!50}{\numprint{237229}}} & \mc{1}{|r}{\cellcolor{lightgray!50}{862.04}} & \mc{1}{r}{\cellcolor{lightgray!50}{\numprint{237120}}} & \mc{1}{|r}{\cellcolor{lightgray!50}{3.98}} & \mc{1}{r}{\textbf{\cellcolor{lightgray!50}{\numprint{237333}}}} & \mc{1}{|r}{\cellcolor{lightgray!50}{1.57}} & \mc{1}{r}{\textbf{\cellcolor{lightgray!50}{\numprint{237333}}}} & \mc{1}{|r}{\cellcolor{lightgray!50}{40.97}} & \mc{1}{r}{\textbf{\cellcolor{lightgray!50}{\numprint{237333}}}} & \mc{1}{|r}{\cellcolor{lightgray!50}{2.08}} & \mc{1}{r}{\textbf{\cellcolor{lightgray!50}{\numprint{237333}}}} \\

  \mc{1}{l}{\cellcolor{lightgray!50}{\texttt{georgia-AM3}}} & \mc{1}{|r}{\cellcolor{lightgray!50}{0.42}} & \mc{1}{r}{\textbf{\cellcolor{lightgray!50}{\numprint{222652}}}} & \mc{1}{|r}{\cellcolor{lightgray!50}{1.31}} & \mc{1}{r}{\textbf{\cellcolor{lightgray!50}{\numprint{222652}}}} & \mc{1}{|r}{\cellcolor{lightgray!50}{0.04}} & \mc{1}{r}{\textbf{\cellcolor{lightgray!50}{\numprint{222652}}}} & \mc{1}{|r}{\cellcolor{lightgray!50}{0.98}} & \mc{1}{r}{\textbf{\cellcolor{lightgray!50}{\numprint{222652}}}} & \mc{1}{|r}{\cellcolor{lightgray!50}{12.97}} & \mc{1}{r}{\textbf{\cellcolor{lightgray!50}{\numprint{222652}}}} & \mc{1}{|r}{\cellcolor{lightgray!50}{14.56}} & \mc{1}{r}{\textbf{\cellcolor{lightgray!50}{\numprint{222652}}}} \\

  \mc{1}{l}{\texttt{greenland-AM3}} & \mc{1}{|r}{58.88} & \mc{1}{r}{\numprint{14007}} & \mc{1}{|r}{640.46} & \mc{1}{r}{\numprint{14010}} & \mc{1}{|r}{1.18} & \mc{1}{r}{\numprint{14011}} & \mc{1}{|r}{10.95} & \mc{1}{r}{\numprint{14011}} & \mc{1}{|r}{58.24} & \mc{1}{r}{\numprint{14008}} & \mc{1}{|r}{5.06} & \mc{1}{r}{\textbf{\numprint{14012}}} \\

  \mc{1}{l}{\cellcolor{lightgray!50}{\texttt{hawaii-AM2}}} & \mc{1}{|r}{\cellcolor{lightgray!50}{1.89}} & \mc{1}{r}{\cellcolor{lightgray!50}{\numprint{125270}}} & \mc{1}{|r}{\cellcolor{lightgray!50}{1.63}} & \mc{1}{r}{\cellcolor{lightgray!50}{\numprint{125270}}} & \mc{1}{|r}{\cellcolor{lightgray!50}{0.20}} & \mc{1}{r}{\textbf{\cellcolor{lightgray!50}{\numprint{125284}}}} & \mc{1}{|r}{\cellcolor{lightgray!50}{0.09}} & \mc{1}{r}{\textbf{\cellcolor{lightgray!50}{\numprint{125284}}}} & \mc{1}{|r}{\cellcolor{lightgray!50}{0.10}} & \mc{1}{r}{\textbf{\cellcolor{lightgray!50}{\numprint{125284}}}} & \mc{1}{|r}{\cellcolor{lightgray!50}{0.13}} & \mc{1}{r}{\textbf{\cellcolor{lightgray!50}{\numprint{125284}}}} \\

  \mc{1}{l}{\texttt{hawaii-AM3}} & \mc{1}{|r}{406.57} & \mc{1}{r}{\numprint{140656}} & \mc{1}{|r}{887.44} & \mc{1}{r}{\numprint{140595}} & \mc{1}{|r}{213.32} & \mc{1}{r}{\textbf{\numprint{141035}}} & \mc{1}{|r}{152.38} & \mc{1}{r}{\numprint{116202}} & \mc{1}{|r}{681.39} & \mc{1}{r}{\numprint{121222}} & \mc{1}{|r}{155.21} & \mc{1}{r}{\numprint{107879}} \\

  \mc{1}{l}{\texttt{idaho-AM3}} & \mc{1}{|r}{79.67} & \mc{1}{r}{\textbf{\numprint{77145}}} & \mc{1}{|r}{58.83} & \mc{1}{r}{\textbf{\numprint{77145}}} & \mc{1}{|r}{0.78} & \mc{1}{r}{\textbf{\numprint{77145}}} & \mc{1}{|r}{11.95} & \mc{1}{r}{\numprint{77141}} & \mc{1}{|r}{40.71} & \mc{1}{r}{\numprint{77144}} & \mc{1}{|r}{8.89} & \mc{1}{r}{\numprint{77144}} \\

  \mc{1}{l}{\cellcolor{lightgray!50}{\texttt{kansas-AM3}}} & \mc{1}{|r}{\cellcolor{lightgray!50}{333.60}} & \mc{1}{r}{\textbf{\cellcolor{lightgray!50}{\numprint{87976}}}} & \mc{1}{|r}{\cellcolor{lightgray!50}{276.26}} & \mc{1}{r}{\textbf{\cellcolor{lightgray!50}{\numprint{87976}}}} & \mc{1}{|r}{\cellcolor{lightgray!50}{0.55}} & \mc{1}{r}{\textbf{\cellcolor{lightgray!50}{\numprint{87976}}}} & \mc{1}{|r}{\cellcolor{lightgray!50}{2.25}} & \mc{1}{r}{\textbf{\cellcolor{lightgray!50}{\numprint{87976}}}} & \mc{1}{|r}{\cellcolor{lightgray!50}{110.41}} & \mc{1}{r}{\textbf{\cellcolor{lightgray!50}{\numprint{87976}}}} & \mc{1}{|r}{\cellcolor{lightgray!50}{337.83}} & \mc{1}{r}{\textbf{\cellcolor{lightgray!50}{\numprint{87976}}}} \\

  \mc{1}{l}{\cellcolor{lightgray!50}{\texttt{kentucky-AM2}}} & \mc{1}{|r}{\cellcolor{lightgray!50}{3.23}} & \mc{1}{r}{\textbf{\cellcolor{lightgray!50}{\numprint{97397}}}} & \mc{1}{|r}{\cellcolor{lightgray!50}{2.92}} & \mc{1}{r}{\textbf{\cellcolor{lightgray!50}{\numprint{97397}}}} & \mc{1}{|r}{\cellcolor{lightgray!50}{0.26}} & \mc{1}{r}{\textbf{\cellcolor{lightgray!50}{\numprint{97397}}}} & \mc{1}{|r}{\cellcolor{lightgray!50}{0.23}} & \mc{1}{r}{\textbf{\cellcolor{lightgray!50}{\numprint{97397}}}} & \mc{1}{|r}{\cellcolor{lightgray!50}{0.44}} & \mc{1}{r}{\textbf{\cellcolor{lightgray!50}{\numprint{97397}}}} & \mc{1}{|r}{\cellcolor{lightgray!50}{0.26}} & \mc{1}{r}{\textbf{\cellcolor{lightgray!50}{\numprint{97397}}}} \\

  \mc{1}{l}{\texttt{kentucky-AM3}} & \mc{1}{|r}{951.91} & \mc{1}{r}{\numprint{100476}} & \mc{1}{|r}{96.83} & \mc{1}{r}{\numprint{100455}} & \mc{1}{|r}{515.99} & \mc{1}{r}{\numprint{100507}} & \mc{1}{|r}{354.45} & \mc{1}{r}{\textbf{\numprint{100510}}} & \mc{1}{|r}{776.69} & \mc{1}{r}{\textbf{\numprint{100510}}} & \mc{1}{|r}{305.01} & \mc{1}{r}{\numprint{100497}} \\

  \mc{1}{l}{\cellcolor{lightgray!50}{\texttt{louisiana-AM3}}} & \mc{1}{|r}{\cellcolor{lightgray!50}{8.63}} & \mc{1}{r}{\textbf{\cellcolor{lightgray!50}{\numprint{60024}}}} & \mc{1}{|r}{\cellcolor{lightgray!50}{0.18}} & \mc{1}{r}{\cellcolor{lightgray!50}{\numprint{60002}}} & \mc{1}{|r}{\cellcolor{lightgray!50}{0.01}} & \mc{1}{r}{\textbf{\cellcolor{lightgray!50}{\numprint{60024}}}} & \mc{1}{|r}{\cellcolor{lightgray!50}{0.05}} & \mc{1}{r}{\textbf{\cellcolor{lightgray!50}{\numprint{60024}}}} & \mc{1}{|r}{\cellcolor{lightgray!50}{0.11}} & \mc{1}{r}{\textbf{\cellcolor{lightgray!50}{\numprint{60024}}}} & \mc{1}{|r}{\cellcolor{lightgray!50}{0.15}} & \mc{1}{r}{\textbf{\cellcolor{lightgray!50}{\numprint{60024}}}} \\

  \mc{1}{l}{\cellcolor{lightgray!50}{\texttt{maryland-AM3}}} & \mc{1}{|r}{\cellcolor{lightgray!50}{0.79}} & \mc{1}{r}{\textbf{\cellcolor{lightgray!50}{\numprint{45496}}}} & \mc{1}{|r}{\cellcolor{lightgray!50}{0.59}} & \mc{1}{r}{\textbf{\cellcolor{lightgray!50}{\numprint{45496}}}} & \mc{1}{|r}{\cellcolor{lightgray!50}{0.01}} & \mc{1}{r}{\textbf{\cellcolor{lightgray!50}{\numprint{45496}}}} & \mc{1}{|r}{\cellcolor{lightgray!50}{0.11}} & \mc{1}{r}{\textbf{\cellcolor{lightgray!50}{\numprint{45496}}}} & \mc{1}{|r}{\cellcolor{lightgray!50}{0.15}} & \mc{1}{r}{\textbf{\cellcolor{lightgray!50}{\numprint{45496}}}} & \mc{1}{|r}{\cellcolor{lightgray!50}{0.14}} & \mc{1}{r}{\textbf{\cellcolor{lightgray!50}{\numprint{45496}}}} \\

  \mc{1}{l}{\cellcolor{lightgray!50}{\texttt{massachusetts-AM2}}} & \mc{1}{|r}{\cellcolor{lightgray!50}{0.25}} & \mc{1}{r}{\textbf{\cellcolor{lightgray!50}{\numprint{140095}}}} & \mc{1}{|r}{\cellcolor{lightgray!50}{0.74}} & \mc{1}{r}{\textbf{\cellcolor{lightgray!50}{\numprint{140095}}}} & \mc{1}{|r}{\cellcolor{lightgray!50}{0.01}} & \mc{1}{r}{\textbf{\cellcolor{lightgray!50}{\numprint{140095}}}} & \mc{1}{|r}{\cellcolor{lightgray!50}{0.04}} & \mc{1}{r}{\textbf{\cellcolor{lightgray!50}{\numprint{140095}}}} & \mc{1}{|r}{\cellcolor{lightgray!50}{0.05}} & \mc{1}{r}{\textbf{\cellcolor{lightgray!50}{\numprint{140095}}}} & \mc{1}{|r}{\cellcolor{lightgray!50}{0.03}} & \mc{1}{r}{\textbf{\cellcolor{lightgray!50}{\numprint{140095}}}} \\

  \mc{1}{l}{\texttt{massachusetts-AM3}} & \mc{1}{|r}{980.11} & \mc{1}{r}{\numprint{145852}} & \mc{1}{|r}{270.28} & \mc{1}{r}{\numprint{145862}} & \mc{1}{|r}{0.77} & \mc{1}{r}{\textbf{\numprint{145866}}} & \mc{1}{|r}{1.39} & \mc{1}{r}{\textbf{\numprint{145866}}} & \mc{1}{|r}{31.04} & \mc{1}{r}{\textbf{\numprint{145866}}} & \mc{1}{|r}{0.76} & \mc{1}{r}{\textbf{\numprint{145866}}} \\

  \mc{1}{l}{\cellcolor{lightgray!50}{\texttt{mexico-AM3}}} & \mc{1}{|r}{\cellcolor{lightgray!50}{0.71}} & \mc{1}{r}{\textbf{\cellcolor{lightgray!50}{\numprint{97663}}}} & \mc{1}{|r}{\cellcolor{lightgray!50}{2.28}} & \mc{1}{r}{\textbf{\cellcolor{lightgray!50}{\numprint{97663}}}} & \mc{1}{|r}{\cellcolor{lightgray!50}{0.02}} & \mc{1}{r}{\textbf{\cellcolor{lightgray!50}{\numprint{97663}}}} & \mc{1}{|r}{\cellcolor{lightgray!50}{0.96}} & \mc{1}{r}{\textbf{\cellcolor{lightgray!50}{\numprint{97663}}}} & \mc{1}{|r}{\cellcolor{lightgray!50}{21.19}} & \mc{1}{r}{\textbf{\cellcolor{lightgray!50}{\numprint{97663}}}} & \mc{1}{|r}{\cellcolor{lightgray!50}{0.67}} & \mc{1}{r}{\textbf{\cellcolor{lightgray!50}{\numprint{97663}}}} \\

  \mc{1}{l}{\cellcolor{lightgray!50}{\texttt{new-hampshire-AM3}}} & \mc{1}{|r}{\cellcolor{lightgray!50}{0.08}} & \mc{1}{r}{\textbf{\cellcolor{lightgray!50}{\numprint{116060}}}} & \mc{1}{|r}{\cellcolor{lightgray!50}{1.63}} & \mc{1}{r}{\textbf{\cellcolor{lightgray!50}{\numprint{116060}}}} & \mc{1}{|r}{\cellcolor{lightgray!50}{0.03}} & \mc{1}{r}{\textbf{\cellcolor{lightgray!50}{\numprint{116060}}}} & \mc{1}{|r}{\cellcolor{lightgray!50}{0.05}} & \mc{1}{r}{\textbf{\cellcolor{lightgray!50}{\numprint{116060}}}} & \mc{1}{|r}{\cellcolor{lightgray!50}{0.08}} & \mc{1}{r}{\textbf{\cellcolor{lightgray!50}{\numprint{116060}}}} & \mc{1}{|r}{\cellcolor{lightgray!50}{0.06}} & \mc{1}{r}{\textbf{\cellcolor{lightgray!50}{\numprint{116060}}}} \\

  \mc{1}{l}{\cellcolor{lightgray!50}{\texttt{north-carolina-AM3}}} & \mc{1}{|r}{\cellcolor{lightgray!50}{0.58}} & \mc{1}{r}{\cellcolor{lightgray!50}{\numprint{49694}}} & \mc{1}{|r}{\cellcolor{lightgray!50}{114.45}} & \mc{1}{r}{\textbf{\cellcolor{lightgray!50}{\numprint{49720}}}} & \mc{1}{|r}{\cellcolor{lightgray!50}{0.03}} & \mc{1}{r}{\textbf{\cellcolor{lightgray!50}{\numprint{49720}}}} & \mc{1}{|r}{\cellcolor{lightgray!50}{0.74}} & \mc{1}{r}{\textbf{\cellcolor{lightgray!50}{\numprint{49720}}}} & \mc{1}{|r}{\cellcolor{lightgray!50}{45.82}} & \mc{1}{r}{\textbf{\cellcolor{lightgray!50}{\numprint{49720}}}} & \mc{1}{|r}{\cellcolor{lightgray!50}{0.47}} & \mc{1}{r}{\textbf{\cellcolor{lightgray!50}{\numprint{49720}}}} \\

  \mc{1}{l}{\cellcolor{lightgray!50}{\texttt{oregon-AM2}}} & \mc{1}{|r}{\cellcolor{lightgray!50}{0.62}} & \mc{1}{r}{\textbf{\cellcolor{lightgray!50}{\numprint{165047}}}} & \mc{1}{|r}{\cellcolor{lightgray!50}{0.37}} & \mc{1}{r}{\textbf{\cellcolor{lightgray!50}{\numprint{165047}}}} & \mc{1}{|r}{\cellcolor{lightgray!50}{0.02}} & \mc{1}{r}{\textbf{\cellcolor{lightgray!50}{\numprint{165047}}}} & \mc{1}{|r}{\cellcolor{lightgray!50}{0.01}} & \mc{1}{r}{\textbf{\cellcolor{lightgray!50}{\numprint{165047}}}} & \mc{1}{|r}{\cellcolor{lightgray!50}{0.01}} & \mc{1}{r}{\textbf{\cellcolor{lightgray!50}{\numprint{165047}}}} & \mc{1}{|r}{\cellcolor{lightgray!50}{0.01}} & \mc{1}{r}{\textbf{\cellcolor{lightgray!50}{\numprint{165047}}}} \\

  \mc{1}{l}{\texttt{oregon-AM3}} & \mc{1}{|r}{174.64} & \mc{1}{r}{\numprint{175059}} & \mc{1}{|r}{511.10} & \mc{1}{r}{\numprint{175067}} & \mc{1}{|r}{4.65} & \mc{1}{r}{\textbf{\numprint{175078}}} & \mc{1}{|r}{9.50} & \mc{1}{r}{\textbf{\numprint{175078}}} & \mc{1}{|r}{39.78} & \mc{1}{r}{\numprint{175077}} & \mc{1}{|r}{21.29} & \mc{1}{r}{\textbf{\numprint{175078}}} \\

  \mc{1}{l}{\cellcolor{lightgray!50}{\texttt{pennsylvania-AM3}}} & \mc{1}{|r}{\cellcolor{lightgray!50}{0.06}} & \mc{1}{r}{\textbf{\cellcolor{lightgray!50}{\numprint{143870}}}} & \mc{1}{|r}{\cellcolor{lightgray!50}{0.14}} & \mc{1}{r}{\textbf{\cellcolor{lightgray!50}{\numprint{143870}}}} & \mc{1}{|r}{\cellcolor{lightgray!50}{0.02}} & \mc{1}{r}{\textbf{\cellcolor{lightgray!50}{\numprint{143870}}}} & \mc{1}{|r}{\cellcolor{lightgray!50}{0.07}} & \mc{1}{r}{\textbf{\cellcolor{lightgray!50}{\numprint{143870}}}} & \mc{1}{|r}{\cellcolor{lightgray!50}{0.12}} & \mc{1}{r}{\textbf{\cellcolor{lightgray!50}{\numprint{143870}}}} & \mc{1}{|r}{\cellcolor{lightgray!50}{0.16}} & \mc{1}{r}{\textbf{\cellcolor{lightgray!50}{\numprint{143870}}}} \\

  \mc{1}{l}{\cellcolor{lightgray!50}{\texttt{rhode-island-AM2}}} & \mc{1}{|r}{\cellcolor{lightgray!50}{7.75}} & \mc{1}{r}{\cellcolor{lightgray!50}{\numprint{184537}}} & \mc{1}{|r}{\cellcolor{lightgray!50}{13.90}} & \mc{1}{r}{\cellcolor{lightgray!50}{\numprint{184576}}} & \mc{1}{|r}{\cellcolor{lightgray!50}{0.24}} & \mc{1}{r}{\textbf{\cellcolor{lightgray!50}{\numprint{184596}}}} & \mc{1}{|r}{\cellcolor{lightgray!50}{0.41}} & \mc{1}{r}{\textbf{\cellcolor{lightgray!50}{\numprint{184596}}}} & \mc{1}{|r}{\cellcolor{lightgray!50}{4.37}} & \mc{1}{r}{\textbf{\cellcolor{lightgray!50}{\numprint{184596}}}} & \mc{1}{|r}{\cellcolor{lightgray!50}{0.27}} & \mc{1}{r}{\textbf{\cellcolor{lightgray!50}{\numprint{184596}}}} \\

  \mc{1}{l}{\texttt{rhode-island-AM3}} & \mc{1}{|r}{230.53} & \mc{1}{r}{\numprint{201470}} & \mc{1}{|r}{711.97} & \mc{1}{r}{\numprint{201359}} & \mc{1}{|r}{30.15} & \mc{1}{r}{\textbf{\numprint{201758}}} & \mc{1}{|r}{44.88} & \mc{1}{r}{\numprint{167162}} & \mc{1}{|r}{82.02} & \mc{1}{r}{\numprint{167162}} & \mc{1}{|r}{45.46} & \mc{1}{r}{\numprint{166103}} \\

  \mc{1}{l}{\cellcolor{lightgray!50}{\texttt{utah-AM3}}} & \mc{1}{|r}{\cellcolor{lightgray!50}{215.88}} & \mc{1}{r}{\cellcolor{lightgray!50}{\numprint{98802}}} & \mc{1}{|r}{\cellcolor{lightgray!50}{136.90}} & \mc{1}{r}{\textbf{\cellcolor{lightgray!50}{\numprint{98847}}}} & \mc{1}{|r}{\cellcolor{lightgray!50}{0.07}} & \mc{1}{r}{\textbf{\cellcolor{lightgray!50}{\numprint{98847}}}} & \mc{1}{|r}{\cellcolor{lightgray!50}{0.09}} & \mc{1}{r}{\textbf{\cellcolor{lightgray!50}{\numprint{98847}}}} & \mc{1}{|r}{\cellcolor{lightgray!50}{0.27}} & \mc{1}{r}{\textbf{\cellcolor{lightgray!50}{\numprint{98847}}}} & \mc{1}{|r}{\cellcolor{lightgray!50}{0.44}} & \mc{1}{r}{\textbf{\cellcolor{lightgray!50}{\numprint{98847}}}} \\

  \mc{1}{l}{\texttt{vermont-AM3}} & \mc{1}{|r}{28.77} & \mc{1}{r}{\numprint{63234}} & \mc{1}{|r}{768.43} & \mc{1}{r}{\numprint{63248}} & \mc{1}{|r}{979.14} & \mc{1}{r}{\numprint{63310}} & \mc{1}{|r}{145.39} & \mc{1}{r}{\textbf{\numprint{63312}}} & \mc{1}{|r}{448.54} & \mc{1}{r}{\textbf{\numprint{63312}}} & \mc{1}{|r}{217.67} & \mc{1}{r}{\textbf{\numprint{63312}}} \\

  \mc{1}{l}{\cellcolor{lightgray!50}{\texttt{virginia-AM2}}} & \mc{1}{|r}{\cellcolor{lightgray!50}{0.53}} & \mc{1}{r}{\cellcolor{lightgray!50}{\numprint{295758}}} & \mc{1}{|r}{\cellcolor{lightgray!50}{20.50}} & \mc{1}{r}{\cellcolor{lightgray!50}{\numprint{295638}}} & \mc{1}{|r}{\cellcolor{lightgray!50}{0.07}} & \mc{1}{r}{\textbf{\cellcolor{lightgray!50}{\numprint{295867}}}} & \mc{1}{|r}{\cellcolor{lightgray!50}{0.02}} & \mc{1}{r}{\textbf{\cellcolor{lightgray!50}{\numprint{295867}}}} & \mc{1}{|r}{\cellcolor{lightgray!50}{0.02}} & \mc{1}{r}{\textbf{\cellcolor{lightgray!50}{\numprint{295867}}}} & \mc{1}{|r}{\cellcolor{lightgray!50}{0.02}} & \mc{1}{r}{\textbf{\cellcolor{lightgray!50}{\numprint{295867}}}} \\

  \mc{1}{l}{\texttt{virginia-AM3}} & \mc{1}{|r}{754.86} & \mc{1}{r}{\numprint{307782}} & \mc{1}{|r}{809.24} & \mc{1}{r}{\numprint{307907}} & \mc{1}{|r}{2.52} & \mc{1}{r}{\textbf{\numprint{308305}}} & \mc{1}{|r}{34.42} & \mc{1}{r}{\textbf{\numprint{308305}}} & \mc{1}{|r}{200.13} & \mc{1}{r}{\textbf{\numprint{308305}}} & \mc{1}{|r}{49.42} & \mc{1}{r}{\textbf{\numprint{308305}}} \\

  \mc{1}{l}{\cellcolor{lightgray!50}{\texttt{washington-AM2}}} & \mc{1}{|r}{\cellcolor{lightgray!50}{1.24}} & \mc{1}{r}{\textbf{\cellcolor{lightgray!50}{\numprint{305619}}}} & \mc{1}{|r}{\cellcolor{lightgray!50}{13.35}} & \mc{1}{r}{\textbf{\cellcolor{lightgray!50}{\numprint{305619}}}} & \mc{1}{|r}{\cellcolor{lightgray!50}{0.25}} & \mc{1}{r}{\textbf{\cellcolor{lightgray!50}{\numprint{305619}}}} & \mc{1}{|r}{\cellcolor{lightgray!50}{0.06}} & \mc{1}{r}{\textbf{\cellcolor{lightgray!50}{\numprint{305619}}}} & \mc{1}{|r}{\cellcolor{lightgray!50}{0.07}} & \mc{1}{r}{\textbf{\cellcolor{lightgray!50}{\numprint{305619}}}} & \mc{1}{|r}{\cellcolor{lightgray!50}{0.08}} & \mc{1}{r}{\textbf{\cellcolor{lightgray!50}{\numprint{305619}}}} \\

  \mc{1}{l}{\texttt{washington-AM3}} & \mc{1}{|r}{37.94} & \mc{1}{r}{\numprint{313689}} & \mc{1}{|r}{383.62} & \mc{1}{r}{\numprint{313844}} & \mc{1}{|r}{10.17} & \mc{1}{r}{\textbf{\numprint{314288}}} & \mc{1}{|r}{3.60} & \mc{1}{r}{\numprint{284684}} & \mc{1}{|r}{72.84} & \mc{1}{r}{\numprint{288116}} & \mc{1}{|r}{4.56} & \mc{1}{r}{\numprint{282020}} \\

  \mc{1}{l}{\cellcolor{lightgray!50}{\texttt{west-virginia-AM3}}} & \mc{1}{|r}{\cellcolor{lightgray!50}{2.75}} & \mc{1}{r}{\textbf{\cellcolor{lightgray!50}{\numprint{47927}}}} & \mc{1}{|r}{\cellcolor{lightgray!50}{2.84}} & \mc{1}{r}{\textbf{\cellcolor{lightgray!50}{\numprint{47927}}}} & \mc{1}{|r}{\cellcolor{lightgray!50}{0.07}} & \mc{1}{r}{\textbf{\cellcolor{lightgray!50}{\numprint{47927}}}} & \mc{1}{|r}{\cellcolor{lightgray!50}{2.88}} & \mc{1}{r}{\textbf{\cellcolor{lightgray!50}{\numprint{47927}}}} & \mc{1}{|r}{\cellcolor{lightgray!50}{41.73}} & \mc{1}{r}{\textbf{\cellcolor{lightgray!50}{\numprint{47927}}}} & \mc{1}{|r}{\cellcolor{lightgray!50}{2.60}} & \mc{1}{r}{\textbf{\cellcolor{lightgray!50}{\numprint{47927}}}} \\

\end{tabu}
\end{ThreePartTable}

\caption{Best solution found by each algorithm on FE and OSM instances and time (in seconds) required to compute it. The global best solution is highlighted in bold. Rows are highlighted in gray if one of our exact solvers is able to solve the corresponding instances.}
\label{best_weight_table_fe}
\end{table}

\begin{table*}[h]
\centering
\scriptsize
\begin{ThreePartTable}
\begin{tabu}{ccccccccccccc}

  \mc{1}{l}{Graph} & \mc{1}{|r}{$t_{max}$} & \mc{1}{r}{$w_{max}$} & \mc{1}{|r}{$t_{max}$} & \mc{1}{r}{$w_{max}$} & \mc{1}{|r}{$t_{max}$} & \mc{1}{r}{$w_{max}$} & \mc{1}{|r}{$t_{max}$} & \mc{1}{r}{$w_{max}$} & \mc{1}{|r}{$t_{max}$} & \mc{1}{r}{$w_{max}$} & \mc{1}{|r}{$t_{max}$} & \mc{1}{r}{$w_{max}$} \\
  \hline  

  \mc{1}{l}{mesh instances} & \mc{2}{|c}{DynWVC1} & \mc{2}{|c}{DynWVC2} & \mc{2}{|c}{HILS} & \mc{2}{|c}{Cyclic-Fast} & \mc{2}{|c}{Cyclic-Strong} & \mc{2}{|c}{Non-Increasing} \\
  \hline  

  \mc{1}{l}{\cellcolor{lightgray!50}{\texttt{beethoven}}} & \mc{1}{|r}{\cellcolor{lightgray!50}{8.86}} & \mc{1}{r}{\cellcolor{lightgray!50}{\numprint{238726}}} & \mc{1}{|r}{\cellcolor{lightgray!50}{8.79}} & \mc{1}{r}{\cellcolor{lightgray!50}{\numprint{238726}}} & \mc{1}{|r}{\cellcolor{lightgray!50}{462.31}} & \mc{1}{r}{\cellcolor{lightgray!50}{\numprint{238746}}} & \mc{1}{|r}{\cellcolor{lightgray!50}{0.00}} & \mc{1}{r}{\textbf{\cellcolor{lightgray!50}{\numprint{238794}}}} & \mc{1}{|r}{\cellcolor{lightgray!50}{0.00}} & \mc{1}{r}{\textbf{\cellcolor{lightgray!50}{\numprint{238794}}}} & \mc{1}{|r}{\cellcolor{lightgray!50}{0.00}} & \mc{1}{r}{\textbf{\cellcolor{lightgray!50}{\numprint{238794}}}} \\

  \mc{1}{l}{\cellcolor{lightgray!50}{\texttt{blob}}} & \mc{1}{|r}{\cellcolor{lightgray!50}{39.91}} & \mc{1}{r}{\cellcolor{lightgray!50}{\numprint{854843}}} & \mc{1}{|r}{\cellcolor{lightgray!50}{40.00}} & \mc{1}{r}{\cellcolor{lightgray!50}{\numprint{854843}}} & \mc{1}{|r}{\cellcolor{lightgray!50}{351.91}} & \mc{1}{r}{\cellcolor{lightgray!50}{\numprint{855004}}} & \mc{1}{|r}{\cellcolor{lightgray!50}{0.02}} & \mc{1}{r}{\textbf{\cellcolor{lightgray!50}{\numprint{855547}}}} & \mc{1}{|r}{\cellcolor{lightgray!50}{0.02}} & \mc{1}{r}{\textbf{\cellcolor{lightgray!50}{\numprint{855547}}}} & \mc{1}{|r}{\cellcolor{lightgray!50}{0.02}} & \mc{1}{r}{\textbf{\cellcolor{lightgray!50}{\numprint{855547}}}} \\

  \mc{1}{l}{\cellcolor{lightgray!50}{\texttt{buddha}}} & \mc{1}{|r}{\cellcolor{lightgray!50}{879.42}} & \mc{1}{r}{\cellcolor{lightgray!50}{\numprint{56757052}}} & \mc{1}{|r}{\cellcolor{lightgray!50}{797.35}} & \mc{1}{r}{\cellcolor{lightgray!50}{\numprint{56757052}}} & \mc{1}{|r}{\cellcolor{lightgray!50}{999.94}} & \mc{1}{r}{\cellcolor{lightgray!50}{\numprint{55490134}}} & \mc{1}{|r}{\cellcolor{lightgray!50}{1.75}} & \mc{1}{r}{\textbf{\cellcolor{lightgray!50}{\numprint{57555880}}}} & \mc{1}{|r}{\cellcolor{lightgray!50}{1.77}} & \mc{1}{r}{\textbf{\cellcolor{lightgray!50}{\numprint{57555880}}}} & \mc{1}{|r}{\cellcolor{lightgray!50}{2.24}} & \mc{1}{r}{\textbf{\cellcolor{lightgray!50}{\numprint{57555880}}}} \\

  \mc{1}{l}{\cellcolor{lightgray!50}{\texttt{bunny}}} & \mc{1}{|r}{\cellcolor{lightgray!50}{702.13}} & \mc{1}{r}{\cellcolor{lightgray!50}{\numprint{3683000}}} & \mc{1}{|r}{\cellcolor{lightgray!50}{695.55}} & \mc{1}{r}{\cellcolor{lightgray!50}{\numprint{3683000}}} & \mc{1}{|r}{\cellcolor{lightgray!50}{964.60}} & \mc{1}{r}{\cellcolor{lightgray!50}{\numprint{3681696}}} & \mc{1}{|r}{\cellcolor{lightgray!50}{0.11}} & \mc{1}{r}{\textbf{\cellcolor{lightgray!50}{\numprint{3686960}}}} & \mc{1}{|r}{\cellcolor{lightgray!50}{0.13}} & \mc{1}{r}{\textbf{\cellcolor{lightgray!50}{\numprint{3686960}}}} & \mc{1}{|r}{\cellcolor{lightgray!50}{0.11}} & \mc{1}{r}{\textbf{\cellcolor{lightgray!50}{\numprint{3686960}}}} \\

  \mc{1}{l}{\cellcolor{lightgray!50}{\texttt{cow}}} & \mc{1}{|r}{\cellcolor{lightgray!50}{62.04}} & \mc{1}{r}{\cellcolor{lightgray!50}{\numprint{269340}}} & \mc{1}{|r}{\cellcolor{lightgray!50}{61.40}} & \mc{1}{r}{\cellcolor{lightgray!50}{\numprint{269340}}} & \mc{1}{|r}{\cellcolor{lightgray!50}{935.58}} & \mc{1}{r}{\cellcolor{lightgray!50}{\numprint{269464}}} & \mc{1}{|r}{\cellcolor{lightgray!50}{0.01}} & \mc{1}{r}{\textbf{\cellcolor{lightgray!50}{\numprint{269543}}}} & \mc{1}{|r}{\cellcolor{lightgray!50}{0.01}} & \mc{1}{r}{\textbf{\cellcolor{lightgray!50}{\numprint{269543}}}} & \mc{1}{|r}{\cellcolor{lightgray!50}{0.01}} & \mc{1}{r}{\textbf{\cellcolor{lightgray!50}{\numprint{269543}}}} \\

  \mc{1}{l}{\cellcolor{lightgray!50}{\texttt{dragon}}} & \mc{1}{|r}{\cellcolor{lightgray!50}{970.34}} & \mc{1}{r}{\cellcolor{lightgray!50}{\numprint{7943911}}} & \mc{1}{|r}{\cellcolor{lightgray!50}{981.51}} & \mc{1}{r}{\cellcolor{lightgray!50}{\numprint{7944042}}} & \mc{1}{|r}{\cellcolor{lightgray!50}{996.01}} & \mc{1}{r}{\cellcolor{lightgray!50}{\numprint{7940422}}} & \mc{1}{|r}{\cellcolor{lightgray!50}{0.21}} & \mc{1}{r}{\textbf{\cellcolor{lightgray!50}{\numprint{7956530}}}} & \mc{1}{|r}{\cellcolor{lightgray!50}{0.22}} & \mc{1}{r}{\textbf{\cellcolor{lightgray!50}{\numprint{7956530}}}} & \mc{1}{|r}{\cellcolor{lightgray!50}{0.22}} & \mc{1}{r}{\textbf{\cellcolor{lightgray!50}{\numprint{7956530}}}} \\

  \mc{1}{l}{\cellcolor{lightgray!50}{\texttt{dragonsub}}} & \mc{1}{|r}{\cellcolor{lightgray!50}{323.07}} & \mc{1}{r}{\cellcolor{lightgray!50}{\numprint{31762035}}} & \mc{1}{|r}{\cellcolor{lightgray!50}{379.11}} & \mc{1}{r}{\cellcolor{lightgray!50}{\numprint{31762035}}} & \mc{1}{|r}{\cellcolor{lightgray!50}{999.54}} & \mc{1}{r}{\cellcolor{lightgray!50}{\numprint{31304363}}} & \mc{1}{|r}{\cellcolor{lightgray!50}{1.10}} & \mc{1}{r}{\textbf{\cellcolor{lightgray!50}{\numprint{32213898}}}} & \mc{1}{|r}{\cellcolor{lightgray!50}{1.11}} & \mc{1}{r}{\textbf{\cellcolor{lightgray!50}{\numprint{32213898}}}} & \mc{1}{|r}{\cellcolor{lightgray!50}{1.88}} & \mc{1}{r}{\textbf{\cellcolor{lightgray!50}{\numprint{32213898}}}} \\

  \mc{1}{l}{\cellcolor{lightgray!50}{\texttt{ecat}}} & \mc{1}{|r}{\cellcolor{lightgray!50}{565.03}} & \mc{1}{r}{\cellcolor{lightgray!50}{\numprint{36129804}}} & \mc{1}{|r}{\cellcolor{lightgray!50}{542.87}} & \mc{1}{r}{\cellcolor{lightgray!50}{\numprint{36129804}}} & \mc{1}{|r}{\cellcolor{lightgray!50}{999.91}} & \mc{1}{r}{\cellcolor{lightgray!50}{\numprint{35512644}}} & \mc{1}{|r}{\cellcolor{lightgray!50}{2.19}} & \mc{1}{r}{\textbf{\cellcolor{lightgray!50}{\numprint{36650298}}}} & \mc{1}{|r}{\cellcolor{lightgray!50}{2.29}} & \mc{1}{r}{\textbf{\cellcolor{lightgray!50}{\numprint{36650298}}}} & \mc{1}{|r}{\cellcolor{lightgray!50}{2.44}} & \mc{1}{r}{\textbf{\cellcolor{lightgray!50}{\numprint{36650298}}}} \\

  \mc{1}{l}{\cellcolor{lightgray!50}{\texttt{face}}} & \mc{1}{|r}{\cellcolor{lightgray!50}{87.05}} & \mc{1}{r}{\cellcolor{lightgray!50}{\numprint{1218510}}} & \mc{1}{|r}{\cellcolor{lightgray!50}{86.38}} & \mc{1}{r}{\cellcolor{lightgray!50}{\numprint{1218510}}} & \mc{1}{|r}{\cellcolor{lightgray!50}{228.77}} & \mc{1}{r}{\cellcolor{lightgray!50}{\numprint{1218565}}} & \mc{1}{|r}{\cellcolor{lightgray!50}{0.03}} & \mc{1}{r}{\textbf{\cellcolor{lightgray!50}{\numprint{1219418}}}} & \mc{1}{|r}{\cellcolor{lightgray!50}{0.03}} & \mc{1}{r}{\textbf{\cellcolor{lightgray!50}{\numprint{1219418}}}} & \mc{1}{|r}{\cellcolor{lightgray!50}{0.03}} & \mc{1}{r}{\textbf{\cellcolor{lightgray!50}{\numprint{1219418}}}} \\

  \mc{1}{l}{\cellcolor{lightgray!50}{\texttt{fandisk}}} & \mc{1}{|r}{\cellcolor{lightgray!50}{8.26}} & \mc{1}{r}{\cellcolor{lightgray!50}{\numprint{462950}}} & \mc{1}{|r}{\cellcolor{lightgray!50}{8.42}} & \mc{1}{r}{\cellcolor{lightgray!50}{\numprint{462950}}} & \mc{1}{|r}{\cellcolor{lightgray!50}{232.96}} & \mc{1}{r}{\cellcolor{lightgray!50}{\numprint{463090}}} & \mc{1}{|r}{\cellcolor{lightgray!50}{0.01}} & \mc{1}{r}{\textbf{\cellcolor{lightgray!50}{\numprint{463288}}}} & \mc{1}{|r}{\cellcolor{lightgray!50}{0.01}} & \mc{1}{r}{\textbf{\cellcolor{lightgray!50}{\numprint{463288}}}} & \mc{1}{|r}{\cellcolor{lightgray!50}{0.01}} & \mc{1}{r}{\textbf{\cellcolor{lightgray!50}{\numprint{463288}}}} \\

  \mc{1}{l}{\cellcolor{lightgray!50}{\texttt{feline}}} & \mc{1}{|r}{\cellcolor{lightgray!50}{730.80}} & \mc{1}{r}{\cellcolor{lightgray!50}{\numprint{2204925}}} & \mc{1}{|r}{\cellcolor{lightgray!50}{734.34}} & \mc{1}{r}{\cellcolor{lightgray!50}{\numprint{2204925}}} & \mc{1}{|r}{\cellcolor{lightgray!50}{640.98}} & \mc{1}{r}{\cellcolor{lightgray!50}{\numprint{2204911}}} & \mc{1}{|r}{\cellcolor{lightgray!50}{0.09}} & \mc{1}{r}{\textbf{\cellcolor{lightgray!50}{\numprint{2207219}}}} & \mc{1}{|r}{\cellcolor{lightgray!50}{0.08}} & \mc{1}{r}{\textbf{\cellcolor{lightgray!50}{\numprint{2207219}}}} & \mc{1}{|r}{\cellcolor{lightgray!50}{0.09}} & \mc{1}{r}{\textbf{\cellcolor{lightgray!50}{\numprint{2207219}}}} \\

  \mc{1}{l}{\cellcolor{lightgray!50}{\texttt{gameguy}}} & \mc{1}{|r}{\cellcolor{lightgray!50}{519.12}} & \mc{1}{r}{\cellcolor{lightgray!50}{\numprint{2323941}}} & \mc{1}{|r}{\cellcolor{lightgray!50}{525.93}} & \mc{1}{r}{\cellcolor{lightgray!50}{\numprint{2323941}}} & \mc{1}{|r}{\cellcolor{lightgray!50}{736.64}} & \mc{1}{r}{\cellcolor{lightgray!50}{\numprint{2322824}}} & \mc{1}{|r}{\cellcolor{lightgray!50}{0.05}} & \mc{1}{r}{\textbf{\cellcolor{lightgray!50}{\numprint{2325878}}}} & \mc{1}{|r}{\cellcolor{lightgray!50}{0.05}} & \mc{1}{r}{\textbf{\cellcolor{lightgray!50}{\numprint{2325878}}}} & \mc{1}{|r}{\cellcolor{lightgray!50}{0.05}} & \mc{1}{r}{\textbf{\cellcolor{lightgray!50}{\numprint{2325878}}}} \\

  \mc{1}{l}{\cellcolor{lightgray!50}{\texttt{gargoyle}}} & \mc{1}{|r}{\cellcolor{lightgray!50}{29.25}} & \mc{1}{r}{\cellcolor{lightgray!50}{\numprint{1058496}}} & \mc{1}{|r}{\cellcolor{lightgray!50}{29.11}} & \mc{1}{r}{\cellcolor{lightgray!50}{\numprint{1058496}}} & \mc{1}{|r}{\cellcolor{lightgray!50}{724.41}} & \mc{1}{r}{\cellcolor{lightgray!50}{\numprint{1058652}}} & \mc{1}{|r}{\cellcolor{lightgray!50}{0.03}} & \mc{1}{r}{\textbf{\cellcolor{lightgray!50}{\numprint{1059559}}}} & \mc{1}{|r}{\cellcolor{lightgray!50}{0.03}} & \mc{1}{r}{\textbf{\cellcolor{lightgray!50}{\numprint{1059559}}}} & \mc{1}{|r}{\cellcolor{lightgray!50}{0.03}} & \mc{1}{r}{\textbf{\cellcolor{lightgray!50}{\numprint{1059559}}}} \\

  \mc{1}{l}{\cellcolor{lightgray!50}{\texttt{turtle}}} & \mc{1}{|r}{\cellcolor{lightgray!50}{982.00}} & \mc{1}{r}{\cellcolor{lightgray!50}{\numprint{14215429}}} & \mc{1}{|r}{\cellcolor{lightgray!50}{976.57}} & \mc{1}{r}{\cellcolor{lightgray!50}{\numprint{14213516}}} & \mc{1}{|r}{\cellcolor{lightgray!50}{999.68}} & \mc{1}{r}{\cellcolor{lightgray!50}{\numprint{14151616}}} & \mc{1}{|r}{\cellcolor{lightgray!50}{0.42}} & \mc{1}{r}{\textbf{\cellcolor{lightgray!50}{\numprint{14263005}}}} & \mc{1}{|r}{\cellcolor{lightgray!50}{0.43}} & \mc{1}{r}{\textbf{\cellcolor{lightgray!50}{\numprint{14263005}}}} & \mc{1}{|r}{\cellcolor{lightgray!50}{0.56}} & \mc{1}{r}{\textbf{\cellcolor{lightgray!50}{\numprint{14263005}}}} \\

  \mc{1}{l}{\cellcolor{lightgray!50}{\texttt{venus}}} & \mc{1}{|r}{\cellcolor{lightgray!50}{559.29}} & \mc{1}{r}{\cellcolor{lightgray!50}{\numprint{305571}}} & \mc{1}{|r}{\cellcolor{lightgray!50}{556.38}} & \mc{1}{r}{\cellcolor{lightgray!50}{\numprint{305571}}} & \mc{1}{|r}{\cellcolor{lightgray!50}{130.83}} & \mc{1}{r}{\cellcolor{lightgray!50}{\numprint{305724}}} & \mc{1}{|r}{\cellcolor{lightgray!50}{0.01}} & \mc{1}{r}{\textbf{\cellcolor{lightgray!50}{\numprint{305749}}}} & \mc{1}{|r}{\cellcolor{lightgray!50}{0.01}} & \mc{1}{r}{\textbf{\cellcolor{lightgray!50}{\numprint{305749}}}} & \mc{1}{|r}{\cellcolor{lightgray!50}{0.01}} & \mc{1}{r}{\textbf{\cellcolor{lightgray!50}{\numprint{305749}}}} \\

  \hline \hline  

  \mc{1}{l}{SNAP instances} & \mc{2}{|c}{DynWVC1} & \mc{2}{|c}{DynWVC2} & \mc{2}{|c}{HILS} & \mc{2}{|c}{Cyclic-Fast} & \mc{2}{|c}{Cyclic-Strong} & \mc{2}{|c}{Non-Increasing} \\
  \hline  

  \mc{1}{l}{\texttt{as-skitter}} & \mc{1}{|r}{989.05} & \mc{1}{r}{\numprint{123613404}} & \mc{1}{|r}{383.97} & \mc{1}{r}{\numprint{123273938}} & \mc{1}{|r}{999.32} & \mc{1}{r}{\numprint{122658804}} & \mc{1}{|r}{346.69} & \mc{1}{r}{\numprint{124137148}} & \mc{1}{|r}{354.71} & \mc{1}{r}{\textbf{\numprint{124137365}}} & \mc{1}{|r}{431.90} & \mc{1}{r}{\numprint{124136621}} \\

  \mc{1}{l}{\cellcolor{lightgray!50}{\texttt{ca-AstroPh}}} & \mc{1}{|r}{\cellcolor{lightgray!50}{32.46}} & \mc{1}{r}{\cellcolor{lightgray!50}{\numprint{797475}}} & \mc{1}{|r}{\cellcolor{lightgray!50}{125.05}} & \mc{1}{r}{\cellcolor{lightgray!50}{\numprint{797480}}} & \mc{1}{|r}{\cellcolor{lightgray!50}{13.47}} & \mc{1}{r}{\textbf{\cellcolor{lightgray!50}{\numprint{797510}}}} & \mc{1}{|r}{\cellcolor{lightgray!50}{0.02}} & \mc{1}{r}{\textbf{\cellcolor{lightgray!50}{\numprint{797510}}}} & \mc{1}{|r}{\cellcolor{lightgray!50}{0.02}} & \mc{1}{r}{\textbf{\cellcolor{lightgray!50}{\numprint{797510}}}} & \mc{1}{|r}{\cellcolor{lightgray!50}{0.02}} & \mc{1}{r}{\textbf{\cellcolor{lightgray!50}{\numprint{797510}}}} \\

  \mc{1}{l}{\cellcolor{lightgray!50}{\texttt{ca-CondMat}}} & \mc{1}{|r}{\cellcolor{lightgray!50}{114.85}} & \mc{1}{r}{\cellcolor{lightgray!50}{\numprint{1147814}}} & \mc{1}{|r}{\cellcolor{lightgray!50}{27.75}} & \mc{1}{r}{\cellcolor{lightgray!50}{\numprint{1147845}}} & \mc{1}{|r}{\cellcolor{lightgray!50}{50.90}} & \mc{1}{r}{\textbf{\cellcolor{lightgray!50}{\numprint{1147950}}}} & \mc{1}{|r}{\cellcolor{lightgray!50}{0.01}} & \mc{1}{r}{\textbf{\cellcolor{lightgray!50}{\numprint{1147950}}}} & \mc{1}{|r}{\cellcolor{lightgray!50}{0.01}} & \mc{1}{r}{\textbf{\cellcolor{lightgray!50}{\numprint{1147950}}}} & \mc{1}{|r}{\cellcolor{lightgray!50}{0.01}} & \mc{1}{r}{\textbf{\cellcolor{lightgray!50}{\numprint{1147950}}}}\\

  \mc{1}{l}{\cellcolor{lightgray!50}{\texttt{ca-GrQc}}} & \mc{1}{|r}{\cellcolor{lightgray!50}{4.87}} & \mc{1}{r}{\textbf{\cellcolor{lightgray!50}{\numprint{286489}}}} & \mc{1}{|r}{\cellcolor{lightgray!50}{1.93}} & \mc{1}{r}{\textbf{\cellcolor{lightgray!50}{\numprint{286489}}}} & \mc{1}{|r}{\cellcolor{lightgray!50}{0.34}} & \mc{1}{r}{\textbf{\cellcolor{lightgray!50}{\numprint{286489}}}} & \mc{1}{|r}{\cellcolor{lightgray!50}{0.00}} & \mc{1}{r}{\textbf{\cellcolor{lightgray!50}{\numprint{286489}}}} & \mc{1}{|r}{\cellcolor{lightgray!50}{0.00}} & \mc{1}{r}{\textbf{\cellcolor{lightgray!50}{\numprint{286489}}}} & \mc{1}{|r}{\cellcolor{lightgray!50}{0.00}} & \mc{1}{r}{\textbf{\cellcolor{lightgray!50}{\numprint{286489}}}} \\

  \mc{1}{l}{\cellcolor{lightgray!50}{\texttt{ca-HepPh}}} & \mc{1}{|r}{\cellcolor{lightgray!50}{13.21}} & \mc{1}{r}{\cellcolor{lightgray!50}{\numprint{581014}}} & \mc{1}{|r}{\cellcolor{lightgray!50}{17.34}} & \mc{1}{r}{\cellcolor{lightgray!50}{\numprint{581028}}} & \mc{1}{|r}{\cellcolor{lightgray!50}{7.73}} & \mc{1}{r}{\textbf{\cellcolor{lightgray!50}{\numprint{581039}}}} & \mc{1}{|r}{\cellcolor{lightgray!50}{0.01}} & \mc{1}{r}{\textbf{\cellcolor{lightgray!50}{\numprint{581039}}}} & \mc{1}{|r}{\cellcolor{lightgray!50}{0.01}} & \mc{1}{r}{\textbf{\cellcolor{lightgray!50}{\numprint{581039}}}} & \mc{1}{|r}{\cellcolor{lightgray!50}{0.01}} & \mc{1}{r}{\textbf{\cellcolor{lightgray!50}{\numprint{581039}}}} \\

  \mc{1}{l}{\cellcolor{lightgray!50}{\texttt{ca-HepTh}}} & \mc{1}{|r}{\cellcolor{lightgray!50}{6.57}} & \mc{1}{r}{\cellcolor{lightgray!50}{\numprint{561982}}} & \mc{1}{|r}{\cellcolor{lightgray!50}{5.30}} & \mc{1}{r}{\cellcolor{lightgray!50}{\numprint{561974}}} & \mc{1}{|r}{\cellcolor{lightgray!50}{4.68}} & \mc{1}{r}{\textbf{\cellcolor{lightgray!50}{\numprint{562004}}}} & \mc{1}{|r}{\cellcolor{lightgray!50}{0.00}} & \mc{1}{r}{\textbf{\cellcolor{lightgray!50}{\numprint{562004}}}} & \mc{1}{|r}{\cellcolor{lightgray!50}{0.00}} & \mc{1}{r}{\textbf{\cellcolor{lightgray!50}{\numprint{562004}}}} & \mc{1}{|r}{\cellcolor{lightgray!50}{0.00}} & \mc{1}{r}{\textbf{\cellcolor{lightgray!50}{\numprint{562004}}}} \\

  \mc{1}{l}{\cellcolor{lightgray!50}{\texttt{email-Enron}}} & \mc{1}{|r}{\cellcolor{lightgray!50}{454.49}} & \mc{1}{r}{\cellcolor{lightgray!50}{\numprint{2464887}}} & \mc{1}{|r}{\cellcolor{lightgray!50}{594.93}} & \mc{1}{r}{\cellcolor{lightgray!50}{\numprint{2464890}}} & \mc{1}{|r}{\cellcolor{lightgray!50}{71.07}} & \mc{1}{r}{\cellcolor{lightgray!50}{\numprint{2464922}}} & \mc{1}{|r}{\cellcolor{lightgray!50}{0.02}} & \mc{1}{r}{\textbf{\cellcolor{lightgray!50}{\numprint{2464935}}}} & \mc{1}{|r}{\cellcolor{lightgray!50}{0.03}} & \mc{1}{r}{\textbf{\cellcolor{lightgray!50}{\numprint{2464935}}}} & \mc{1}{|r}{\cellcolor{lightgray!50}{0.02}} & \mc{1}{r}{\textbf{\cellcolor{lightgray!50}{\numprint{2464935}}}} \\

  \mc{1}{l}{\cellcolor{lightgray!50}{\texttt{email-EuAll}}} & \mc{1}{|r}{\cellcolor{lightgray!50}{134.83}} & \mc{1}{r}{\textbf{\cellcolor{lightgray!50}{\numprint{25286322}}}} & \mc{1}{|r}{\cellcolor{lightgray!50}{132.62}} & \mc{1}{r}{\textbf{\cellcolor{lightgray!50}{\numprint{25286322}}}} & \mc{1}{|r}{\cellcolor{lightgray!50}{338.14}} & \mc{1}{r}{\textbf{\cellcolor{lightgray!50}{\numprint{25286322}}}} & \mc{1}{|r}{\cellcolor{lightgray!50}{0.07}} & \mc{1}{r}{\textbf{\cellcolor{lightgray!50}{\numprint{25286322}}}} & \mc{1}{|r}{\cellcolor{lightgray!50}{0.07}} & \mc{1}{r}{\textbf{\cellcolor{lightgray!50}{\numprint{25286322}}}} & \mc{1}{|r}{\cellcolor{lightgray!50}{0.06}} & \mc{1}{r}{\textbf{\cellcolor{lightgray!50}{\numprint{25286322}}}} \\

  \mc{1}{l}{\cellcolor{lightgray!50}{\texttt{p2p-Gnutella04}}} & \mc{1}{|r}{\cellcolor{lightgray!50}{1.46}} & \mc{1}{r}{\cellcolor{lightgray!50}{\numprint{679105}}} & \mc{1}{|r}{\cellcolor{lightgray!50}{2.34}} & \mc{1}{r}{\textbf{\cellcolor{lightgray!50}{\numprint{679111}}}} & \mc{1}{|r}{\cellcolor{lightgray!50}{94.12}} & \mc{1}{r}{\textbf{\cellcolor{lightgray!50}{\numprint{679111}}}} & \mc{1}{|r}{\cellcolor{lightgray!50}{0.01}} & \mc{1}{r}{\textbf{\cellcolor{lightgray!50}{\numprint{679111}}}} & \mc{1}{|r}{\cellcolor{lightgray!50}{0.01}} & \mc{1}{r}{\textbf{\cellcolor{lightgray!50}{\numprint{679111}}}} & \mc{1}{|r}{\cellcolor{lightgray!50}{0.01}} & \mc{1}{r}{\textbf{\cellcolor{lightgray!50}{\numprint{679111}}}} \\

  \mc{1}{l}{\cellcolor{lightgray!50}{\texttt{p2p-Gnutella05}}} & \mc{1}{|r}{\cellcolor{lightgray!50}{1.15}} & \mc{1}{r}{\cellcolor{lightgray!50}{\numprint{554926}}} & \mc{1}{|r}{\cellcolor{lightgray!50}{3.55}} & \mc{1}{r}{\cellcolor{lightgray!50}{\numprint{554931}}} & \mc{1}{|r}{\cellcolor{lightgray!50}{135.17}} & \mc{1}{r}{\textbf{\cellcolor{lightgray!50}{\numprint{554943}}}} & \mc{1}{|r}{\cellcolor{lightgray!50}{0.01}} & \mc{1}{r}{\textbf{\cellcolor{lightgray!50}{\numprint{554943}}}} & \mc{1}{|r}{\cellcolor{lightgray!50}{0.01}} & \mc{1}{r}{\textbf{\cellcolor{lightgray!50}{\numprint{554943}}}} & \mc{1}{|r}{\cellcolor{lightgray!50}{0.01}} & \mc{1}{r}{\textbf{\cellcolor{lightgray!50}{\numprint{554943}}}} \\

  \mc{1}{l}{\cellcolor{lightgray!50}{\texttt{p2p-Gnutella06}}} & \mc{1}{|r}{\cellcolor{lightgray!50}{525.35}} & \mc{1}{r}{\cellcolor{lightgray!50}{\numprint{548611}}} & \mc{1}{|r}{\cellcolor{lightgray!50}{186.97}} & \mc{1}{r}{\cellcolor{lightgray!50}{\numprint{548611}}} & \mc{1}{|r}{\cellcolor{lightgray!50}{1.29}} & \mc{1}{r}{\textbf{\cellcolor{lightgray!50}{\numprint{548612}}}} & \mc{1}{|r}{\cellcolor{lightgray!50}{0.01}} & \mc{1}{r}{\textbf{\cellcolor{lightgray!50}{\numprint{548612}}}} & \mc{1}{|r}{\cellcolor{lightgray!50}{0.01}} & \mc{1}{r}{\textbf{\cellcolor{lightgray!50}{\numprint{548612}}}} & \mc{1}{|r}{\cellcolor{lightgray!50}{0.01}} & \mc{1}{r}{\textbf{\cellcolor{lightgray!50}{\numprint{548612}}}} \\

  \mc{1}{l}{\cellcolor{lightgray!50}{\texttt{p2p-Gnutella08}}} & \mc{1}{|r}{\cellcolor{lightgray!50}{0.15}} & \mc{1}{r}{\cellcolor{lightgray!50}{\numprint{434575}}} & \mc{1}{|r}{\cellcolor{lightgray!50}{0.18}} & \mc{1}{r}{\textbf{\cellcolor{lightgray!50}{\numprint{434577}}}} & \mc{1}{|r}{\cellcolor{lightgray!50}{0.12}} & \mc{1}{r}{\textbf{\cellcolor{lightgray!50}{\numprint{434577}}}} & \mc{1}{|r}{\cellcolor{lightgray!50}{0.00}} & \mc{1}{r}{\textbf{\cellcolor{lightgray!50}{\numprint{434577}}}} & \mc{1}{|r}{\cellcolor{lightgray!50}{0.00}} & \mc{1}{r}{\textbf{\cellcolor{lightgray!50}{\numprint{434577}}}} & \mc{1}{|r}{\cellcolor{lightgray!50}{0.00}} & \mc{1}{r}{\textbf{\cellcolor{lightgray!50}{\numprint{434577}}}} \\

  \mc{1}{l}{\cellcolor{lightgray!50}{\texttt{p2p-Gnutella09}}} & \mc{1}{|r}{\cellcolor{lightgray!50}{0.39}} & \mc{1}{r}{\textbf{\cellcolor{lightgray!50}{\numprint{568439}}}} & \mc{1}{|r}{\cellcolor{lightgray!50}{0.28}} & \mc{1}{r}{\textbf{\cellcolor{lightgray!50}{\numprint{568439}}}} & \mc{1}{|r}{\cellcolor{lightgray!50}{0.09}} & \mc{1}{r}{\textbf{\cellcolor{lightgray!50}{\numprint{568439}}}} & \mc{1}{|r}{\cellcolor{lightgray!50}{0.00}} & \mc{1}{r}{\textbf{\cellcolor{lightgray!50}{\numprint{568439}}}} & \mc{1}{|r}{\cellcolor{lightgray!50}{0.00}} & \mc{1}{r}{\textbf{\cellcolor{lightgray!50}{\numprint{568439}}}} & \mc{1}{|r}{\cellcolor{lightgray!50}{0.00}} & \mc{1}{r}{\textbf{\cellcolor{lightgray!50}{\numprint{568439}}}} \\

  \mc{1}{l}{\cellcolor{lightgray!50}{\texttt{p2p-Gnutella24}}} & \mc{1}{|r}{\cellcolor{lightgray!50}{8.01}} & \mc{1}{r}{\textbf{\cellcolor{lightgray!50}{\numprint{1984567}}}} & \mc{1}{|r}{\cellcolor{lightgray!50}{5.51}} & \mc{1}{r}{\textbf{\cellcolor{lightgray!50}{\numprint{1984567}}}} & \mc{1}{|r}{\cellcolor{lightgray!50}{3.17}} & \mc{1}{r}{\textbf{\cellcolor{lightgray!50}{\numprint{1984567}}}} & \mc{1}{|r}{\cellcolor{lightgray!50}{0.01}} & \mc{1}{r}{\textbf{\cellcolor{lightgray!50}{\numprint{1984567}}}} & \mc{1}{|r}{\cellcolor{lightgray!50}{0.01}} & \mc{1}{r}{\textbf{\cellcolor{lightgray!50}{\numprint{1984567}}}} & \mc{1}{|r}{\cellcolor{lightgray!50}{0.01}} & \mc{1}{r}{\textbf{\cellcolor{lightgray!50}{\numprint{1984567}}}} \\

  \mc{1}{l}{\cellcolor{lightgray!50}{\texttt{p2p-Gnutella25}}} & \mc{1}{|r}{\cellcolor{lightgray!50}{2.66}} & \mc{1}{r}{\textbf{\cellcolor{lightgray!50}{\numprint{1701967}}}} & \mc{1}{|r}{\cellcolor{lightgray!50}{2.20}} & \mc{1}{r}{\textbf{\cellcolor{lightgray!50}{\numprint{1701967}}}} & \mc{1}{|r}{\cellcolor{lightgray!50}{1.17}} & \mc{1}{r}{\textbf{\cellcolor{lightgray!50}{\numprint{1701967}}}} & \mc{1}{|r}{\cellcolor{lightgray!50}{0.01}} & \mc{1}{r}{\textbf{\cellcolor{lightgray!50}{\numprint{1701967}}}} & \mc{1}{|r}{\cellcolor{lightgray!50}{0.01}} & \mc{1}{r}{\textbf{\cellcolor{lightgray!50}{\numprint{1701967}}}} & \mc{1}{|r}{\cellcolor{lightgray!50}{0.01}} & \mc{1}{r}{\textbf{\cellcolor{lightgray!50}{\numprint{1701967}}}} \\

  \mc{1}{l}{\cellcolor{lightgray!50}{\texttt{p2p-Gnutella30}}} & \mc{1}{|r}{\cellcolor{lightgray!50}{8.83}} & \mc{1}{r}{\cellcolor{lightgray!50}{\numprint{2787903}}} & \mc{1}{|r}{\cellcolor{lightgray!50}{132.71}} & \mc{1}{r}{\cellcolor{lightgray!50}{\numprint{2787899}}} & \mc{1}{|r}{\cellcolor{lightgray!50}{15.14}} & \mc{1}{r}{\textbf{\cellcolor{lightgray!50}{\numprint{2787907}}}} & \mc{1}{|r}{\cellcolor{lightgray!50}{0.01}} & \mc{1}{r}{\textbf{\cellcolor{lightgray!50}{\numprint{2787907}}}} & \mc{1}{|r}{\cellcolor{lightgray!50}{0.01}} & \mc{1}{r}{\textbf{\cellcolor{lightgray!50}{\numprint{2787907}}}} & \mc{1}{|r}{\cellcolor{lightgray!50}{0.02}} & \mc{1}{r}{\textbf{\cellcolor{lightgray!50}{\numprint{2787907}}}} \\

  \mc{1}{l}{\cellcolor{lightgray!50}{\texttt{p2p-Gnutella31}}} & \mc{1}{|r}{\cellcolor{lightgray!50}{70.88}} & \mc{1}{r}{\cellcolor{lightgray!50}{\numprint{4776960}}} & \mc{1}{|r}{\cellcolor{lightgray!50}{47.97}} & \mc{1}{r}{\cellcolor{lightgray!50}{\numprint{4776961}}} & \mc{1}{|r}{\cellcolor{lightgray!50}{115.01}} & \mc{1}{r}{\textbf{\cellcolor{lightgray!50}{\numprint{4776986}}}} & \mc{1}{|r}{\cellcolor{lightgray!50}{0.02}} & \mc{1}{r}{\textbf{\cellcolor{lightgray!50}{\numprint{4776986}}}} & \mc{1}{|r}{\cellcolor{lightgray!50}{0.02}} & \mc{1}{r}{\textbf{\cellcolor{lightgray!50}{\numprint{4776986}}}} & \mc{1}{|r}{\cellcolor{lightgray!50}{0.03}} & \mc{1}{r}{\textbf{\cellcolor{lightgray!50}{\numprint{4776986}}}} \\

  \mc{1}{l}{\cellcolor{lightgray!50}{\texttt{roadNet-CA}}} & \mc{1}{|r}{\cellcolor{lightgray!50}{999.98}} & \mc{1}{r}{\cellcolor{lightgray!50}{\numprint{109586054}}} & \mc{1}{|r}{\cellcolor{lightgray!50}{999.90}} & \mc{1}{r}{\cellcolor{lightgray!50}{\numprint{109582579}}} & \mc{1}{|r}{\cellcolor{lightgray!50}{\numprint{1000.00}}} & \mc{1}{r}{\cellcolor{lightgray!50}{\numprint{106584645}}} & \mc{1}{|r}{\cellcolor{lightgray!50}{1.94}} & \mc{1}{r}{\textbf{\cellcolor{lightgray!50}{\numprint{111360828}}}} & \mc{1}{|r}{\cellcolor{lightgray!50}{1.86}} & \mc{1}{r}{\textbf{\cellcolor{lightgray!50}{\numprint{111360828}}}} & \mc{1}{|r}{\cellcolor{lightgray!50}{4.09}} & \mc{1}{r}{\textbf{\cellcolor{lightgray!50}{\numprint{111360828}}}} \\

  \mc{1}{l}{\cellcolor{lightgray!50}{\texttt{roadNet-PA}}} & \mc{1}{|r}{\cellcolor{lightgray!50}{511.59}} & \mc{1}{r}{\cellcolor{lightgray!50}{\numprint{60990177}}} & \mc{1}{|r}{\cellcolor{lightgray!50}{469.18}} & \mc{1}{r}{\cellcolor{lightgray!50}{\numprint{60990177}}} & \mc{1}{|r}{\cellcolor{lightgray!50}{999.94}} & \mc{1}{r}{\cellcolor{lightgray!50}{\numprint{60037011}}} & \mc{1}{|r}{\cellcolor{lightgray!50}{0.96}} & \mc{1}{r}{\textbf{\cellcolor{lightgray!50}{\numprint{61731589}}}} & \mc{1}{|r}{\cellcolor{lightgray!50}{1.04}} & \mc{1}{r}{\textbf{\cellcolor{lightgray!50}{\numprint{61731589}}}} & \mc{1}{|r}{\cellcolor{lightgray!50}{1.83}} & \mc{1}{r}{\textbf{\cellcolor{lightgray!50}{\numprint{61731589}}}} \\

  \mc{1}{l}{\cellcolor{lightgray!50}{\texttt{roadNet-TX}}} & \mc{1}{|r}{\cellcolor{lightgray!50}{789.43}} & \mc{1}{r}{\cellcolor{lightgray!50}{\numprint{77672388}}} & \mc{1}{|r}{\cellcolor{lightgray!50}{694.33}} & \mc{1}{r}{\cellcolor{lightgray!50}{\numprint{77672388}}} & \mc{1}{|r}{\cellcolor{lightgray!50}{999.97}} & \mc{1}{r}{\cellcolor{lightgray!50}{\numprint{76347666}}} & \mc{1}{|r}{\cellcolor{lightgray!50}{1.29}} & \mc{1}{r}{\textbf{\cellcolor{lightgray!50}{\numprint{78599946}}}} & \mc{1}{|r}{\cellcolor{lightgray!50}{1.29}} & \mc{1}{r}{\textbf{\cellcolor{lightgray!50}{\numprint{78599946}}}} & \mc{1}{|r}{\cellcolor{lightgray!50}{3.42}} & \mc{1}{r}{\textbf{\cellcolor{lightgray!50}{\numprint{78599946}}}} \\

  \mc{1}{l}{\cellcolor{lightgray!50}{\texttt{soc-Epinions1}}} & \mc{1}{|r}{\cellcolor{lightgray!50}{290.84}} & \mc{1}{r}{\cellcolor{lightgray!50}{\numprint{5690651}}} & \mc{1}{|r}{\cellcolor{lightgray!50}{272.56}} & \mc{1}{r}{\cellcolor{lightgray!50}{\numprint{5690773}}} & \mc{1}{|r}{\cellcolor{lightgray!50}{253.10}} & \mc{1}{r}{\cellcolor{lightgray!50}{\numprint{5690874}}} & \mc{1}{|r}{\cellcolor{lightgray!50}{0.08}} & \mc{1}{r}{\textbf{\cellcolor{lightgray!50}{\numprint{5690970}}}} & \mc{1}{|r}{\cellcolor{lightgray!50}{0.08}} & \mc{1}{r}{\textbf{\cellcolor{lightgray!50}{\numprint{5690970}}}} & \mc{1}{|r}{\cellcolor{lightgray!50}{0.08}} & \mc{1}{r}{\textbf{\cellcolor{lightgray!50}{\numprint{5690970}}}} \\

  \mc{1}{l}{\texttt{soc-LiveJournal1}} & \mc{1}{|r}{999.99} & \mc{1}{r}{\numprint{279150686}} & \mc{1}{|r}{999.99} & \mc{1}{r}{\numprint{279231875}} & \mc{1}{|r}{\numprint{1000.00}} & \mc{1}{r}{\numprint{255079926}} & \mc{1}{|r}{51.33} & \mc{1}{r}{\numprint{284036222}} & \mc{1}{|r}{44.19} & \mc{1}{r}{\textbf{\numprint{284036239}}} & \mc{1}{|r}{39.36} & \mc{1}{r}{\numprint{283970295}} \\

  \mc{1}{l}{\cellcolor{lightgray!50}{\texttt{soc-Slashdot0811}}} & \mc{1}{|r}{\cellcolor{lightgray!50}{238.18}} & \mc{1}{r}{\cellcolor{lightgray!50}{\numprint{5660385}}} & \mc{1}{|r}{\cellcolor{lightgray!50}{880.68}} & \mc{1}{r}{\cellcolor{lightgray!50}{\numprint{5660555}}} & \mc{1}{|r}{\cellcolor{lightgray!50}{446.95}} & \mc{1}{r}{\cellcolor{lightgray!50}{\numprint{5660787}}} & \mc{1}{|r}{\cellcolor{lightgray!50}{0.09}} & \mc{1}{r}{\textbf{\cellcolor{lightgray!50}{\numprint{5660899}}}} & \mc{1}{|r}{\cellcolor{lightgray!50}{0.08}} & \mc{1}{r}{\textbf{\cellcolor{lightgray!50}{\numprint{5660899}}}} & \mc{1}{|r}{\cellcolor{lightgray!50}{0.08}} & \mc{1}{r}{\textbf{\cellcolor{lightgray!50}{\numprint{5660899}}}} \\

  \mc{1}{l}{\cellcolor{lightgray!50}{\texttt{soc-Slashdot0902}}} & \mc{1}{|r}{\cellcolor{lightgray!50}{270.85}} & \mc{1}{r}{\cellcolor{lightgray!50}{\numprint{5971308}}} & \mc{1}{|r}{\cellcolor{lightgray!50}{435.90}} & \mc{1}{r}{\cellcolor{lightgray!50}{\numprint{5971476}}} & \mc{1}{|r}{\cellcolor{lightgray!50}{604.07}} & \mc{1}{r}{\cellcolor{lightgray!50}{\numprint{5971664}}} & \mc{1}{|r}{\cellcolor{lightgray!50}{0.11}} & \mc{1}{r}{\textbf{\cellcolor{lightgray!50}{\numprint{5971849}}}} & \mc{1}{|r}{\cellcolor{lightgray!50}{0.11}} & \mc{1}{r}{\textbf{\cellcolor{lightgray!50}{\numprint{5971849}}}} & \mc{1}{|r}{\cellcolor{lightgray!50}{0.12}} & \mc{1}{r}{\textbf{\cellcolor{lightgray!50}{\numprint{5971849}}}} \\

  \mc{1}{l}{\texttt{soc-pokec-relationships}} & \mc{1}{|r}{999.85} & \mc{1}{r}{\textbf{\numprint{83223668}}} & \mc{1}{|r}{999.13} & \mc{1}{r}{\numprint{83155217}} & \mc{1}{|r}{\numprint{1000.00}} & \mc{1}{r}{\numprint{82021946}} & \mc{1}{|r}{254.59} & \mc{1}{r}{\numprint{76075111}} & \mc{1}{|r}{488.31} & \mc{1}{r}{\numprint{76075700}} & \mc{1}{|r}{228.07} & \mc{1}{r}{\numprint{76063476}} \\

  \mc{1}{l}{\cellcolor{lightgray!50}{\texttt{web-BerkStan}}} & \mc{1}{|r}{\cellcolor{lightgray!50}{194.20}} & \mc{1}{r}{\cellcolor{lightgray!50}{\numprint{43640833}}} & \mc{1}{|r}{\cellcolor{lightgray!50}{164.10}} & \mc{1}{r}{\cellcolor{lightgray!50}{\numprint{43637382}}} & \mc{1}{|r}{\cellcolor{lightgray!50}{998.73}} & \mc{1}{r}{\cellcolor{lightgray!50}{\numprint{43424373}}} & \mc{1}{|r}{\cellcolor{lightgray!50}{6.74}} & \mc{1}{r}{\textbf{\cellcolor{lightgray!50}{\numprint{43907482}}}} & \mc{1}{|r}{\cellcolor{lightgray!50}{8.05}} & \mc{1}{r}{\textbf{\cellcolor{lightgray!50}{\numprint{43907482}}}} & \mc{1}{|r}{\cellcolor{lightgray!50}{16.01}} & \mc{1}{r}{\textbf{\cellcolor{lightgray!50}{\numprint{43907482}}}} \\

  \mc{1}{l}{\cellcolor{lightgray!50}{\texttt{web-Google}}} & \mc{1}{|r}{\cellcolor{lightgray!50}{349.08}} & \mc{1}{r}{\cellcolor{lightgray!50}{\numprint{56209005}}} & \mc{1}{|r}{\cellcolor{lightgray!50}{324.65}} & \mc{1}{r}{\cellcolor{lightgray!50}{\numprint{56206250}}} & \mc{1}{|r}{\cellcolor{lightgray!50}{995.92}} & \mc{1}{r}{\cellcolor{lightgray!50}{\numprint{56008278}}} & \mc{1}{|r}{\cellcolor{lightgray!50}{1.72}} & \mc{1}{r}{\textbf{\cellcolor{lightgray!50}{\numprint{56326504}}}} & \mc{1}{|r}{\cellcolor{lightgray!50}{6.44}} & \mc{1}{r}{\textbf{\cellcolor{lightgray!50}{\numprint{56326504}}}} & \mc{1}{|r}{\cellcolor{lightgray!50}{2.17}} & \mc{1}{r}{\textbf{\cellcolor{lightgray!50}{\numprint{56326504}}}} \\

  \mc{1}{l}{\cellcolor{lightgray!50}{\texttt{web-NotreDame}}} & \mc{1}{|r}{\cellcolor{lightgray!50}{949.84}} & \mc{1}{r}{\cellcolor{lightgray!50}{\numprint{26010791}}} & \mc{1}{|r}{\cellcolor{lightgray!50}{905.72}} & \mc{1}{r}{\cellcolor{lightgray!50}{\numprint{26009287}}} & \mc{1}{|r}{\cellcolor{lightgray!50}{997.00}} & \mc{1}{r}{\cellcolor{lightgray!50}{\numprint{26002793}}} & \mc{1}{|r}{\cellcolor{lightgray!50}{1.60}} & \mc{1}{r}{\textbf{\cellcolor{lightgray!50}{\numprint{26016941}}}} & \mc{1}{|r}{\cellcolor{lightgray!50}{2.74}} & \mc{1}{r}{\textbf{\cellcolor{lightgray!50}{\numprint{26016941}}}} & \mc{1}{|r}{\cellcolor{lightgray!50}{1.36}} & \mc{1}{r}{\textbf{\cellcolor{lightgray!50}{\numprint{26016941}}}} \\

  \mc{1}{l}{\cellcolor{lightgray!50}{\texttt{web-Stanford}}} & \mc{1}{|r}{\cellcolor{lightgray!50}{943.85}} & \mc{1}{r}{\cellcolor{lightgray!50}{\numprint{17748798}}} & \mc{1}{|r}{\cellcolor{lightgray!50}{671.32}} & \mc{1}{r}{\cellcolor{lightgray!50}{\numprint{17741043}}} & \mc{1}{|r}{\cellcolor{lightgray!50}{999.50}} & \mc{1}{r}{\cellcolor{lightgray!50}{\numprint{17709827}}} & \mc{1}{|r}{\cellcolor{lightgray!50}{1.68}} & \mc{1}{r}{\textbf{\cellcolor{lightgray!50}{\numprint{17792930}}}} & \mc{1}{|r}{\cellcolor{lightgray!50}{1.86}} & \mc{1}{r}{\textbf{\cellcolor{lightgray!50}{\numprint{17792930}}}} & \mc{1}{|r}{\cellcolor{lightgray!50}{1.71}} & \mc{1}{r}{\textbf{\cellcolor{lightgray!50}{\numprint{17792930}}}} \\

  \mc{1}{l}{\cellcolor{lightgray!50}{\texttt{wiki-Talk}}} & \mc{1}{|r}{\cellcolor{lightgray!50}{951.51}} & \mc{1}{r}{\cellcolor{lightgray!50}{\numprint{235836837}}} & \mc{1}{|r}{\cellcolor{lightgray!50}{972.93}} & \mc{1}{r}{\cellcolor{lightgray!50}{\numprint{235836913}}} & \mc{1}{|r}{\cellcolor{lightgray!50}{999.69}} & \mc{1}{r}{\cellcolor{lightgray!50}{\numprint{235818823}}} & \mc{1}{|r}{\cellcolor{lightgray!50}{1.29}} & \mc{1}{r}{\textbf{\cellcolor{lightgray!50}{\numprint{235837346}}}} & \mc{1}{|r}{\cellcolor{lightgray!50}{1.29}} & \mc{1}{r}{\textbf{\cellcolor{lightgray!50}{\numprint{235837346}}}} & \mc{1}{|r}{\cellcolor{lightgray!50}{1.31}} & \mc{1}{r}{\textbf{\cellcolor{lightgray!50}{\numprint{235837346}}}} \\

  \mc{1}{l}{\cellcolor{lightgray!50}{\texttt{wiki-Vote}}} & \mc{1}{|r}{\cellcolor{lightgray!50}{188.76}} & \mc{1}{r}{\cellcolor{lightgray!50}{\numprint{500075}}} & \mc{1}{|r}{\cellcolor{lightgray!50}{0.32}} & \mc{1}{r}{\textbf{\cellcolor{lightgray!50}{\numprint{500079}}}} & \mc{1}{|r}{\cellcolor{lightgray!50}{10.34}} & \mc{1}{r}{\textbf{\cellcolor{lightgray!50}{\numprint{500079}}}} & \mc{1}{|r}{\cellcolor{lightgray!50}{0.02}} & \mc{1}{r}{\textbf{\cellcolor{lightgray!50}{\numprint{500079}}}} & \mc{1}{|r}{\cellcolor{lightgray!50}{0.02}} & \mc{1}{r}{\textbf{\cellcolor{lightgray!50}{\numprint{500079}}}} & \mc{1}{|r}{\cellcolor{lightgray!50}{0.02}} & \mc{1}{r}{\textbf{\cellcolor{lightgray!50}{\numprint{500079}}}} \\

\end{tabu}
\end{ThreePartTable}

\caption{Best solution found by each algorithm on mesh and SNAP instances and time (in seconds) required to compute it. The global best solution is highlighted in bold. Rows are highlighted in gray if one of our exact solvers is able to solve the corresponding instances.}
\label{best_weight_table_mesh}
\end{table*}
\end{landscape}


\end{appendix}
\end{document}